\documentclass[a4paper]{llncs}

\usepackage{times}
\usepackage{mathptmx}

\usepackage{graphicx}
\usepackage{float}
\usepackage{tabularx}

\usepackage{paralist}

\usepackage{url}

\usepackage{dsfont}

\usepackage{lmodern}
\usepackage[T1]{fontenc}

\usepackage{cancel}	
\usepackage[pdftex]{thumbpdf}
\usepackage{graphicx}
\usepackage{multicol}
\usepackage{colortbl}

\usepackage{todonotes}

\usepackage{amsfonts,amsmath,amscd,amssymb,amsbsy} 

\usepackage{dsfont} 
\usepackage{eufrak}

\usepackage[T1]{fontenc}

\usepackage{mathbbol} 

\def\bbA{\mathbb{A}} 
\def\bbP{\mathbb{P}} 
\def\bbS{\mathbb{S}} 
\def\bbV{\mathbb{V}} 
\def\bbW{\mathbb{W}} 

\def\mcL{\mathcal{L}} 
\def\mcO{\mathcal{O}} 
\def\mcI{\mathcal{I}} 
\def\mcF{\mathcal{F}} 
\def\mcD{\mathcal{D}} 
\def\mcE{\mathcal{E}} 
\def\mcP{\mathcal{P}} 
\def\mcR{\mathcal{R}} 
\def\mcC{\mathcal{C}} 

\def\dsN{\mathds{N}} 

\def\SP{\mathrm{SP}} 
\def\SWF{\mathrm{SF}} 
\def\SRF{\mathrm{SRF}} 
\def\S{\mathrm{S}} 
\def\pre{\mathrm{pre}} 
\def\alf{\mathrm{alph}} 
\newcommand{\abs}[1]{\left\lvert#1\right\rvert}

\def\I{\mathrm{I}} 
\def\II{\mathrm{II}} 
\def\III{\mathrm{III}} 
\def\IV{\mathrm{IV}} 
\def\V{\mathrm{V}} 

\usepackage{color}

\newcommand{\dotcup}{\ensuremath{\mathaccent\cdot\cup}}

\usepackage{shuffle}

\usepackage{mathptmx}

\usepackage{tabularx}

\definecolor{mybgcolor}{rgb}{0.95,0.95,0.95}
\definecolor{myframecolor}{rgb}{0.6,0.6,0.6}
\definecolor{darkgreen}{rgb}{0.15,0.6,0.15}
\definecolor{darkblue}{rgb}{0,0,0.54}
\definecolor{white}{rgb}{1,1,1}
\definecolor{blockgreen}{rgb}{0.9,0.95,0.9}

\usepackage{shuffle}

\usepackage{subfigure}

\usepackage{tikz}
\usetikzlibrary{arrows,shapes,snakes,automata,backgrounds,petri,positioning,fit}
\tikzstyle{every state}=[fill=none,draw=black,text=black,minimum size=0.4cm]
\tikzstyle{init} = [pin edge={to-,thin,black}]
 \tikzstyle{surround} = [fill=blue!0,
                         draw=black,rounded corners=2mm]
  \tikzstyle{warn}= [starburst,fill=white,draw]
  \tikzstyle{transition-root}=[rectangle,thick,draw=black!75,
  			  fill=black!0,
                          minimum size=4mm]
  \tikzstyle{transition-intern}=[rectangle,draw=black!75,
  			  fill=black!0,
                          minimum size=4mm]
  \tikzstyle{transition-in}=[rectangle,draw=black!75,
  			  fill=black!0,
                          minimum size=4mm]
  \tikzstyle{transition-out}=[rectangle,draw=black!75,
  			  fill=black!0,
                          minimum size=4mm]
  \tikzstyle{clabel}=[]
\tikzstyle{cstate}= [state,circle,inner sep=3pt]
\tikzstyle{rstate}= [state,rectangle, rounded corners=8pt,inner sep=4pt]

\usepackage[bookmarks,bookmarksopen,bookmarksdepth=2,colorlinks=true,linkcolor=blue,citecolor=blue]{hyperref} 

\let\orgautoref\autoref
\renewcommand{\autoref}
        {\def\equationautorefname{Equation}%
         \def\figureautorefname{Figure}%
         \def\definitionautorefname{Definition}%
         \def\assumptionautorefname{Assumption}%
         \def\theoremautorefname{Theorem}%
         \def\subfigureautorefname{Figure}%
         \def\Itemautorefname{Item}%
         \def\tableautorefname{Table}%
         \def\algorithmautorefname{Algorithm}%
         \def\paragraphautorefname{Paragraph}%
         \def\sectionautorefname{Section}%
         \def\subsectionautorefname{Section}%
         \def\subsubsectionautorefname{Section}%
         \def\chapterautorefname{Section}%
         \def\partautorefname{Part}%
         \def\goalautorefname{Goal}%
         \def\reqautorefname{Req.}%
         \def\adviceautorefname{Rule}%
         \def\parameterautorefname{Param.}%
         \def\lstnumberautorefname{Line }%
         \orgautoref}

\usepackage[OT2,T1]{fontenc}
\DeclareSymbolFont{cyrletters}{OT2}{wncyr}{m}{n}
\DeclareMathSymbol{\Sha}{\mathalpha}{cyrletters}{"58}

\begin{document}


\title{Pairs of Languages Closed under Shuffle Projection \\
}

    \author{Peter Ochsenschl\"ager\inst{1} and Roland Rieke\inst{1}\inst{2}}

 \institute{Fraunhofer Institute for Secure Information Technology SIT, Darmstadt, Germany\\
\and
 Philipps-Universit\"at Marburg, Germany\\
      \email{peter-ochsenschlaeger@t-online.de, roland.rieke@sit.fraunhofer.de}}

\maketitle

\begin{abstract}
Shuffle projection is motivated by the verification 
of safety properties of special parameterized systems. 
Basic definitions and 
properties, especially related to alphabetic homomorphisms, are presented. 
The 
relation between iterated shuffle products and shuffle projections is shown. 
A special class 
of multi-counter automata is introduced, to formulate 
shuffle projection in terms 
of computations of these automata represented by transductions. 
This reformulation of shuffle projection 
leads 
to construction principles for pairs of languages closed under shuffle projection. 
Additionally, it is shown 
that under certain conditions these transductions are rational, 
which implies decidability of closure 
against shuffle projection. 
Decidability of these conditions is proven 
for regular languages. 
Finally, without additional conditions, 
decidability of the question, 
whether a pair of regular languages is closed under shuffle projection, is shown.
In an appendix the relation between shuffle projection and the shuffle product of two languages is discussed. 
Additionally, a kind of shuffle product for computations in S-automata is defined.
\end{abstract}  

\section{Introduction and Motivation}
The definition of shuffle projection is motivated by our investigations of self-similarity
of scalable systems \cite{SysMea14}. Let us consider some examples:
\begin{example}\label{ex:sp-1}
A server answers requests of a family of clients. 
The actions of the server are considered in the following.
We assume w.r.t. each client that a request will be answered 
before a new request from this client is accepted.
If the family of clients consists of only one client, then the automaton in
Fig.~\ref{fig:4} describes the 
%
system behavior $S\subset \Sigma^*$, where $\Sigma =\{a,b\}$,
the label $a$ depicts the $request$, and $b$ depicts the $response$.

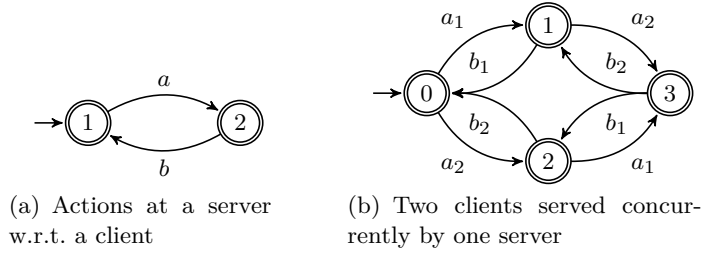
\begin{figure}[ht]
\centering
\subfigure[Actions at a server w.r.t. a client\label{fig:4}]{
\begin{tikzpicture}[->,>=stealth',shorten >=1pt,auto,node distance=1cm,semithick,initial text=]
   \node[cstate,initial,accepting] (k1) {$1$};
   \node[cstate,accepting] (k2) at (2.0,0){$2$};
   \path 
         (k1) edge [bend left] node {\small $a$} (k2)
         (k2) edge [bend left] node {\small $b$} (k1)
		;
\end{tikzpicture}
}\hspace{1cm}%
\subfigure[Two clients served concurrently by one server\label{fig:5}]
{
\begin{tikzpicture}[->,>=stealth',shorten >=1pt,auto,node distance=1cm,semithick,initial text=]
   \node[cstate,accepting,initial] (k0) {$0$};
   \node[cstate,accepting] (k1) at (1.6,0.9){$1$};
   \node[cstate,accepting] (k2) at (1.6,-0.9) {$2$};
   \node[cstate,accepting] (k3) at (3.2,0) {$3$};
   \path 
         (k0) edge [bend left] node {\small $a_1$} (k1)
         (k1) edge [bend left,swap] node {\small $b_1$} (k0)
         (k1) edge [bend left] node {\small $a_2$} (k3)
         (k3) edge [bend left,swap] node {\small $b_2$} (k1)
         (k0) edge [bend right,swap] node {\small $a_2$} (k2)
         (k2) edge [bend right] node {\small $b_2$} (k0)
         (k2) edge [bend right,swap] node {\small $a_1$} (k3)
         (k3) edge [bend right] node {\small $b_1$} (k2)
		;
\end{tikzpicture}
}
\caption{Scalable client-server system}
\end{figure}
\end{example}

\begin{example}\label{ex:sp-2}
Fig.~\ref{fig:5} now describes the system behavior $S_{\{1,2\}}\subset \Sigma_{\{1,2\}}^*$ for
two clients $1$ and $2$, under the
assumption that the server handles the requests of different clients
non-restricted concurrently.

\end{example}
\ \\
For $\emptyset\neq I$ and $i \in I$ let $\Sigma_{\{i\}}$ denote 
pairwise disjoint copies of $\Sigma$. 
The elements of $\Sigma_{\{i\}}$ are denoted by $a_{i}$ and 
$\Sigma_{I}:= \dot{\bigcup\limits_{i \in I}}\Sigma_{\{i\}}$.
Additionally let
$\Sigma_\emptyset:= \emptyset$, and $\Sigma_\emptyset ^*:= \{\varepsilon\}$.   
The index $i$ describes the bijection $a\leftrightarrow a_{i}$ for $a\in \Sigma$ 
and $a_{i} \in \Sigma_{\{i\}}$. 
\ \\

\begin{example}\label{ex:sp-3}
For $\emptyset\neq I\subset \dsN$ with finite $I$, let now $S_I\subset \Sigma_I^*$ denote the system behavior w.r.t. the client set $I$. 
For each $i\in \dsN$  $S_{\{i\}}$ is isomorphic to $S$, and
$S_I$ consists of the non-restricted concurrent run of all $S_{\{i\}}$ with $i\in I$. 
\ \\
Let $\mcI_1$ denote the set of all finite non-empty subsets of $\dsN$ (the set of all possible clients). 
Then, the family $(S_I)_{I\in\mcI_1}$ has the following properties:

\begin{itemize}
\item $I\subset K$ implies $S_I \subset S_K$ (monotony)
\item $I \approx K$ implies $S_I \approx S_K$ (uniform parameterization)
\end{itemize}
Such families are called \emph{scalable systems} \cite{SysMea14}.
\ \\
Here $\approx$ denotes isomorphic.
Notice, each bijection  $\iota:I\to K$ defines an isomorphism
$\iota_K^I:\Sigma_I^*\to\Sigma_K^*$ .
\end{example}
%


In section 2 the basic definitions and 
properties, especially related to alphabetic homomorphisms, are presented. Section 3 shows the 
relations between iterated shuffle products and shuffle projections. In section 4 a special class 
of multi-counter automata are introduced, to formulate in section 5 shuffle projection in terms 
of computations of these automata. This reformulation of shuffle projection 
leads in section 6 to construction principles for pairs of languages closed under shuffle projection. 
In section 7 the results of section 5 are represented by transductions. Additionally, it is shown 
that under certain conditions these transductions are rational, which imply decidability of closure 
against shuffle projection. In section 8 decidability of these conditions is proven 
for regular languages. Finally, without the restrictions of section 7, decidability of the question, 
whether a pair of regular languages is closed under shuffle projection, is shown.
In an appendix the relation between shuffle projection and the shuffle product of two languages is discussed. 
Additionally a kind of shuffle product for computations in S-automata is defined, which shows the results 
of section 5 from another point of view.    

\section{Basic Definitions and Homomorphic Properties}

\begin{definition}\label{def:tau}\ \\
For $I\subset N$ and $n\in N$ 
let $\tau_n^I:\Sigma_I^*\to\Sigma^*$ be the
homomorphisms defined by 
\[
\tau_n^I(a_i)
= \left\{ \begin{array}{r@{\ }cl}
a & |\ & a_{i} \in \Sigma_{I\cap\{n\}}\\
\varepsilon & |\ & a_{i} \in \Sigma_{I\setminus \{n\}}
\end{array} \right. .
\]
\end{definition}
For a singleton set \{n\},
$\tau_n^{\{n\}}:\Sigma_{\{n\}}^*\to \Sigma^*$ is an isomorphism.
\ \\
For $I\in\mcI_1$ holds
\[S_I=\bigcap\limits_{n\in I}(\tau_n^I)^{-1}(S) . \]  

\begin{definition}[\ \ $(\dot{\mcL}(L)_I)_{I\in\mcI_1}$\ \ ]\label{def:L-dot}
\ \\
Let $\emptyset\neq L\subset \Sigma^*$ be prefix closed and
\[\dot{\mcL}(L)_I := \bigcap\limits_{i\in I}(\tau_i^I)^{-1}(L) 
\]
for $I\in\mcI_1$.
\end{definition}
The systems $\dot{\mcL}(L)_I$ consist of the 
``non-restricted concurrent run'' of all systems 
$(\tau_i^{\{i\}})^{-1}(L)\subset \Sigma_{\{i\}}^*$ with $i\in I$. 
Because $\tau_i^{\{i\}}: \Sigma_{\{i\}}^*\to\Sigma^*$ are isomorphisms,
$(\tau_i^{\{i\}})^{-1}(L)$ are pairwise disjoint copies of $L$.

\begin{theorem}\label{thm:L-dot scalable}
\ \\
$(\dot{\mcL}(L)_I)_{I\in\mcI_1}$ is a scalable system  \cite{SysMea14}.
\end{theorem}
Now we show how to construct well-behaved systems by restricting concurrency in the behaviour-family $\dot{\mcL}$.
In Example~\ref{ex:sp-3} holds 
$S_I=\dot{\mcL}(S)_I$
for $I\in\mcI_1$.
%
If, in Example~\ref{ex:sp-3}, the server needs specific resources for the processing of 
a request, then - on account of restricted resources - an non-restricted concurrent
processing of requests is not possible. 
Thus, restrictions of concurrency in terms of synchronization conditions are necessary.
One possible but very strong restriction is the requirement that the server
handles the requests of different clients in the same way as it handles
the requests of a single client, namely, on the request follows the 
response and vice versa. This synchronization condition can be formalized with the help of $S$
and the homomorphisms $\Theta^I$.
\begin{definition}\label{def:theta}\ \\
For a set $I$ let the homomorphism
\[\Theta^I:\Sigma_I^*\to\Sigma^* \mbox{ be defined by } \Theta^I(a_i):= a , \]
$\mbox{for } i\in I \mbox{ and } a\in \Sigma$.
\end{definition}

\begin{example}\label{ex:sp-5}
Restriction of concurrency on account of restricted resources: one ``task'' after another.
All behaviors with respect to  $i\in I$ \emph{influence} each other.
Let
\[\bar{S}_I:=S_I\cap(\Theta^I)^{-1}(S)=\bigcap\limits_{i\in I}(\tau_i^I)^{-1}(S)\cap(\Theta^I)^{-1}(S)\]
for $I\in \mcI_1$. 
\end{example}
From the automaton in Fig.~\ref{fig:5} it is evident that $\bar{S}_{\{1,2\}}$ will be 
recognized by the automaton in Fig.~\ref{fig:6}.
\begin{figure}[ht]
\centering
\subfigure[Automaton recognizing $\bar{S}_{\{1,2\}}$ \label{fig:6}]{
\begin{tikzpicture}[->,>=stealth',shorten >=1pt,auto,node distance=1cm,semithick,initial text=]
  \useasboundingbox (-2.5,-0.8) rectangle (2.5,0.8);
   \node[cstate,initial,accepting] (k0) {$0$};
   \node[cstate,accepting] (k1) at (1.8,0){$1$};
   \node[cstate,accepting] (k2) at (-1.8,0){$2$};
   \path 
         (k0) edge [bend left] node {\small $a_1$} (k1)
         (k1) edge [bend left] node {\small $b_1$} (k0)
         (k0) edge [bend right,swap] node {\small $a_2$} (k2)
         (k2) edge [bend right,swap] node {\small $b_2$} (k0)
		;
\end{tikzpicture}
}\hspace{0.1cm}
\subfigure[\label{fig:7}]{
\begin{tikzpicture}[->,>=stealth',shorten >=1pt,auto,node distance=1cm,semithick,initial text=]
   \node[cstate,initial,accepting] (k1) {$0$};
   \node[cstate,accepting] (k2) at (1.8,0){$i$};
   \path 
         (k1) edge [bend left] node {\small $a_i$} (k2)
         (k2) edge [bend left] node {\small $b_i$} (k1)
		;
\end{tikzpicture}
}
\caption{}
\end{figure}
Given an arbitrary $I\in\mcI_1$, then $\bar{S}_I$ is recognized by an automaton with 
state set $\{0\}\cup I$ and state transition relation given by Fig.~\ref{fig:7} 
for each $i\in I$. 
From this automaton it is evident that $(\bar{S}_I)_{I\in\mcI_1}$ is a scalable system \cite{SysMea14}.
\begin{definition}[\ \ $(\bar{\mcL}(L,V)_I)_{I\in\mcI_1}$]\label{def:L-bar}
Let $\emptyset \neq L\subset V\subset \Sigma^*$ be prefix closed and
\[\bar{\mcL}(L,V)_I := 
\bigcap\limits_{n\in I}(\tau_n^I)^{-1}(L)\cap(\Theta^I)^{-1}(V) \mbox{ for } I\in\mcI_1 .\]
\end{definition}
In \cite{SysMea14} it is shown
\begin{theorem}\label{thm:L-bar scalable}
$(\bar{\mcL}(L,V)_I)_{I\in\mcI_1}$ is a scalable system.
\end{theorem}
To consider arbitrary scalable systems $(\mcL_I)_{I\in\mcI}$ general parameter structures have to be defined:
\begin{definition}[parameter structure]\label{def:parameter-structure} 
Let $N$ be a countable (infinite) set and 
$\emptyset \neq \mcI \subset 2^N \setminus \{\emptyset\}$.
$\mcI$ is called a \emph{parameter structure} based on $N$.
\end{definition}

\begin{definition}[self-similar scalable system]\label{def:self-similar} 
For arbitrary sets $I' \subset I$ let
$\Pi_{I'}^{I}: \Sigma_{I}^* \to \Sigma_{I'}^*
\mbox{ with }$
\[\Pi_{I'}^{I}(a_{i})
= \left\{ \begin{array}{r@{\ }cl}
a_{i} & |\ & a_{i} \in \Sigma_{I'}\\
\varepsilon & |\ & a_{i} \in \Sigma_{I}\setminus \Sigma_{I'}.
\end{array} \right. \]
A scalable system $(\mcL_I)_{I\in\mcI}$ is called \emph{self-similar} iff 
\[\Pi_{I'}^I (\mcL_I)=\mcL_{I'} \mbox{ for each } I, I' \in \mcI \mbox{ with } I'\subset I .\]
\end{definition}
\ \\
Examples: In \cite{SysMea14} it is shown that $(S_I)_{I\in \mcI_1}$ and $(\bar{S}_I)_{I\in \mcI_1}$ 
are self-similar scalable systems.
\ \\
In \cite{OR2011} it is shown that for self-similar scalable systems a large class of safety properties 
(uniformly parameterized safety properties) can be verified by inspecting only one corresponding ``prototype system'' 
instead of inspecting the whole family of systems. This demonstrates the importance of self-similarity for scalable systems.
\ \\
The following example shows that not each $(\bar{\mcL}(L,V)_I)_{I\in\mcI_1}$ is self-similar.
\begin{example}\label{counter-example}
Let $G\subset \{a,b,c\}^*$ the prefix closed language that is 
recognized by the automaton in Fig.~\ref{fig:8}.
Let $H\subset \{a,b,c\}^*$ the prefix closed language that is 
recognized by the automaton in Fig.~\ref{fig:9}.
It holds $\emptyset\neq G\subset H$, however, 
$(\mcL(G,\bar{\mcE}_{\mcI_1},H)_I)_{I\in \mcI_1}$ is not self-similar, 
e.g., 
\[\Pi_{\{2,3\}}^{\{1,2,3\}}(\mcL(G,\bar{\mcE}_{\mcI_1},H)_{\{1,2,3\}}) \neq (\mcL(G,\bar{\mcE}_{\mcI_1},H)_{\{2,3\}} 
\mbox{ because }\]
$a_{1}b_{1}a_{2}a_{3}\in \mcL(G,\bar{\mcE}_{\mcI_1},H)_{\{1,2,3\}}$, and hence 
$a_{2}a_{3}\in \Pi_{\{2,3\}}^{\{1,2,3\}}(\mcL(G,\bar{\mcE}_{\mcI_1},H)_{\{1,2,3\}})$,
but $a_{2}a_{3}\notin (\mcL(G,\bar{\mcE}_{\mcI_1},H)_{\{2,3\}}$.

\begin{figure}[ht]
\centering
\subfigure[Automaton recognizing $G$\label{fig:8}]{
\scalebox{0.9} {
\begin{tikzpicture}[->,>=stealth',shorten >=1pt,auto,node distance=2cm,semithick,initial text=]
  \useasboundingbox (-2.5,-1.5) rectangle (2.5,0);
   \node[state,accepting,initial] (k1) {\small 1};
   \node[state,accepting] (k2) at (-0.6,-1.2) {\small 2};
   \node[state,accepting] (k3) at (0.6,-1.2) {\small 3};
   \path 
         (k1) edge  node [swap] {\small $a$} (k2)
         (k2) edge  node {\small $b$} (k3)
         (k3) edge  node [swap] {\small $c$} (k1)
		;
\end{tikzpicture}
}}\hspace{0.01cm}%
\subfigure[Automaton recognizing $H$\label{fig:9}]{
\scalebox{0.9} {
\begin{tikzpicture}[->,>=stealth',shorten >=1pt,auto,node distance=2cm,semithick,initial text=]
  \useasboundingbox (-3.5,-2.5) rectangle (3.5,0);
   \node[state,accepting,initial] (k1) {\small 1};
   \node[state,accepting] (k2) at (-1,-1) {\small 2};
   \node[state,accepting] (k3) at (1,-1) {\small 3};
   \node[state,accepting] (k4) at (0,-2) {\small 4};
   \node[state,accepting] (k5) [right of =k4,node distance=1.4cm] {\small 5};
   \node[state,accepting] (k6) at (2.6,-2) {\small 6};
   \node[state,accepting] (k7) [above of =k6,node distance=1cm] {\small 7};
   \node[state,accepting] (k9) [right of =k1,node distance=1.4cm] {\small 9};
   \node[state,accepting] (k8) at (2.6,0) {\small 8};
   \path 
         (k1) edge  node [swap] {\small $a$} (k2)
         (k2) edge  node {\small $b$} (k3)
         (k3) edge  node [swap] {\small $c$} (k1)
         (k3) edge  node {\small $a$} (k4)
         (k4) edge  node {\small $c$} (k2)
         (k4) edge  [swap] node {\small $a$} (k5)
         (k5) edge  [swap] node {\small $b$} (k6)
         (k6) edge  node [swap] {\small $b$} (k7)
         (k7) edge  node [swap] {\small $c$} (k8)
         (k8) edge  node [swap] {\small $c$} (k9)
         (k9) edge  node [swap] {\small $c$} (k1)
		;
\end{tikzpicture}
}}
\caption{Counterexample}
\end{figure}
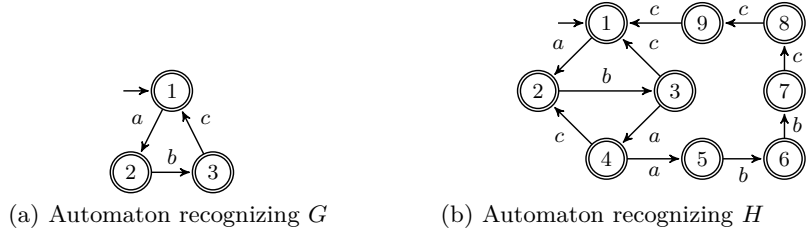
\end{example}

\begin{theorem}\label{thm:L-bar(L,V) and SP(L,V)}
Let $\emptyset \neq L\subset V\subset \Sigma^*$ be prefix closed and $(\bar{\mcL}(L,V)_I)_{I\in\mcI_1}$ self-similar. 
Then 
\[\Pi_K^\dsN[(\bigcap\limits_{n\in \dsN}(\tau_n^\dsN)^{-1}(L)) \cap (\Theta^\dsN)^{-1}(V)]\subset (\Theta^\dsN)^{-1}(V)\]
for each subset $K\subset \dsN$.
\end{theorem}

\begin{proof}
Let $w\in\bigcap\limits_{n\in \dsN}(\tau_n^\dsN)^{-1}(L) \cap (\Theta^\dsN)^{-1}(V)$, then there exists
$J\in\mcI_1$  with $w \in \Sigma_{J}^*$ and therefore
\begin{equation}\label{eq:sp-1}
w \in \bar{\mcL}(L,V)_J .
\end{equation}
\ \\
Now
\begin{equation}\label{eq:sp-2}
\Pi_K^\dsN (w) = \Pi_{K \cap J}^J (w).
\end{equation}
\ \\
If $ K \cap J = \emptyset $, then
\begin{equation}\label{eq:sp-3}
\Pi_K^\dsN (w) = \varepsilon \in (\Theta^\dsN)^{-1}(V).
\end{equation}
\ \\
If $K \cap J \neq \emptyset $, then $K \cap J \in\mcI_1$.
Now \eqref{eq:sp-1}, \eqref{eq:sp-2} and self-similarity of $(\bar{\mcL}(L,V)_I)_{I\in\mcI_1}$ implies
\begin{equation}\label{eq:sp-4}
\Pi_K^\dsN (w) \in \bar{\mcL}(L,V)_{K \cap J} \subset (\Theta^{K \cap J})^{-1}(V) \subset (\Theta^\dsN)^{-1}(V).
\end{equation}
\ \\
\eqref{eq:sp-3} and \eqref{eq:sp-4} completes the proof of Theorem~\ref{thm:L-bar(L,V) and SP(L,V)}.

\end{proof}
%
%
%
In \cite{SysMea14} it is shown that 
\[\Pi_K^\dsN[(\bigcap\limits_{n\in \dsN}(\tau_n^\dsN)^{-1}(L)) \cap (\Theta^\dsN)^{-1}(V)]\subset (\Theta^\dsN)^{-1}(V)
\mbox{ for each subset } \emptyset \neq K\subset \dsN \]
is a sufficient condition for self-similarity of a large class of scalable systems including $(\bar{\mcL}(L,V)_I)_{I\in\mcI_1}$.
So we define:

\begin{definition}[{closed under shuffle projection}]\label{def:shuffle-projection}
Let $U, V \subset \Sigma^*$. $V$ is \emph{closed under shuffle projection
with respect to $U$}, iff 
\[\Pi_K^\dsN[(\bigcap\limits_{n\in \dsN}(\tau_n^\dsN)^{-1}(U)) \cap (\Theta^\dsN)^{-1}(V)]\subset (\Theta^\dsN)^{-1}(V)
\mbox{ for each subset } \emptyset \neq K\subset \dsN.\]
We abbreviate this by $\SP(U,V)$.
\end{definition}
Now it holds

\begin{corollary}\label{cor:sp-1}
Let $\emptyset \neq L\subset V\subset \Sigma^*$ be prefix closed. 
Then $\SP(L,V)$ is equivalent to self-similarity of $(\bar{\mcL}(L,V)_I)_{I\in\mcI_1}$. 
\end{corollary}
\noindent \textbf{Remark.} It is easy to see that in Definition~\ref{def:shuffle-projection} $\dsN$ can be replaced by any 
set $N$ having the same cardinality as $\dsN$ \cite{SysMea14}.
\ \\
\ \\
In the last section of this paper decidability of $\SP(U,V)$ will be proven for regular languages $U$ and $V$. 
In preparation for this proof and supplementary to this result, first 
we investigate sufficient conditions for $\SP(U,V)$ and equivalent formulations of $\SP(U,V)$. \\
\ \\
By simple set theory the definition of $\SP(U,V)$ has some immediate consequences:
\begin{equation}\label{eq:sp-5}
\SP(U,V) \mbox{ implies } \SP(U',V) \mbox{ for each } U'\subset U.
\end{equation}
\ \\
Let $\emptyset \neq I$. Then $\SP(U,V_i)$ for each $i \in I$ implies
\begin{equation}\label{eq:sp-6}
\SP(U,\bigcap\limits_{i\in I}V_i) \mbox{ and } \SP(U,\bigcup\limits_{i\in I}V_i).
\end{equation}
\ \\
In \cite{SysMea14} the following theorem has been proven:
\begin{theorem}\label{thm:inverse abstraction}
Let $\varphi:\Sigma^*\to\Phi^*$ be an alphabetic homomorphism
and $W, X\subset\Phi^*$, then 
$\SP(W,X) \mbox{ implies } \SP(\varphi^{-1}(W),\varphi^{-1}(X))$.
\end{theorem}
Because of \eqref{eq:sp-5} and Theorem~\ref{thm:inverse abstraction}
\begin{equation}\label{eq:sp-6.1}
\SP(\varphi(U),V) \mbox{ implies } \SP(U,\varphi^{-1}(V)).\\
\end{equation}
\ \\
The inverse of implication \eqref{eq:sp-6.1} also holds.
For its proof additional notations and a lemma is needed:
\ \\
\ \\
Let $K$ be a non-empty set.
Each alphabetic homomorphism $\varphi:\Sigma^*\to\Phi^*$
defines a homomorphism
$\varphi^K:\Sigma_K^*\to\Phi_K^*$ by
\begin{equation}\label{eq:pp-62}
\varphi^K(a_n) := (\varphi(a))_n \mbox{ for } a_n\in\Sigma_K,
\mbox{ where } (\varepsilon)_n:=\varepsilon.
\end{equation}
If $\bar{\tau}_n^K:\Phi_K^*\to\Phi^*$ and $\bar{\Theta}^K:\Phi_K^*\to\Phi^*$ are defined
analogously to $\tau_n^K$ and $\Theta^K$, then
\begin{equation}\label{eq:pp-63}
\varphi\circ\tau_n^K=\bar{\tau}_n^K\circ\varphi^K, \mbox{ and } 
\varphi\circ\Theta^K=\bar{\Theta}^K\circ\varphi^K .
\end{equation}
Let $ K \subset N$ and
$\bar{\Pi}_K^N:\Phi_N^*\to \Phi_K^*$ be defined analogously to $\Pi_K^N$, then 
\begin{equation}\label{eq:pp-63.2}
\bar{\Pi}_K^N\circ\varphi^N=\varphi^K\circ\Pi_K^N.
\end{equation}
\begin{lemma} \label{lemma:sp-1} 
Let $\varphi:\Sigma^*\to\Phi^*$ be an alphabetic homomorphism,
$U\subset\Phi^*$ and $N$ be a non-empty set, then 
\[\varphi^N(\bigcap\limits_{t\in N}(\tau_t^N)^{-1}(U))=\bigcap\limits_{t\in N}(\bar{\tau}_t^N)^{-1}(\varphi(U)).\]
\end{lemma}

\begin{proof}

Because of \eqref{eq:pp-63}
for $x\in \bigcap\limits_{t \in N} (\tau_t^N)^{-1}(U)$ and $t\in  N$ holds
\[\bar{\tau}_t^N (\varphi^N(x)) = \varphi(\tau_t^N(x))\in \varphi(U) ,\]
 and therefore
\[\varphi^N(\bigcap\limits_{t \in N} (\tau_t^N)^{-1}(U)) \subset \bigcap\limits_{t \in N} (\bar{\tau}_t^N)^{-1}(\varphi(U)) .\]
\ \\
The contrary inclusion will be
proven by the following \emph{proposition}:\\
\ \\
For $y\in \Phi_N^*$ let $T(y)$ be the finite set defined by $T(y) := \{t\in N \ |\ \bar{\tau}_t^N(y)\neq \varepsilon\}$.
Then for each $y\in \Phi_N^*$ and $(u_t)_{t\in N}$
with $\bar{\tau}_t^N(y)=\varphi(u_t)$, $u_t\in \Sigma^+$ for $t\in T(y)$
and $u_t =\varepsilon$ for $t\in N\setminus T(y)$ exists an 
$x\in \Sigma_N^{*}$ with $y=\varphi^N(x)$ and 
$\tau_t^{N}(x)=u_t$ for each $t\in N$.

\begin{proof} [Proof of the proposition by induction.] \ \\
\emph{Induction base.}  \\
For $y=\varepsilon$ holds $T(y)=\emptyset$, and $x=\varepsilon$ satisfies the proposition.\\
\emph{Induction step.}  \\
Let $y=y'a'_s\in\Phi^{*}_N$ with $a'_s\in \Phi_{\{s\}}$ and
$\bar{\tau}_t^N(y)=\varphi(u_t)$ with $u_t\in\Sigma^+$ for $t\in T(y)$
as well as $u_t = \varepsilon$ for $t\in N \setminus T(y)$.\\
Then holds $s \in T(y)$, because $\bar{\tau}_s^N(y)=\bar{\tau}_s^N(y')a'_s\neq\varepsilon.$ \\
Let now $u_s=u'_sv'_s$ with $v'_s\in\Sigma^+$,
$a'_s=\bar{\tau}_s^N(a'_s)=\varphi(v'_s)\neq\varepsilon$ and $u'_s=\varepsilon$
when $\bar{\tau}_s^N(y')=\varepsilon$.\\
For $t\in N\setminus \{s\}$ let $u'_t:= u_t$.\\
$y'\in\Phi^{*}_N$ and $(u'_t)_{t\in N}$ now satisfy the induction
hypothesis. Therefore exists $x'\in\Sigma^{*}_N$ with $y'=\varphi^N(x')$ and
$\tau_t^N(x')=u'_t$ for each $t \in N$.\\
Because of the injectivity of $\tau_s^N$ on $\Sigma^*_{\{s\}}$ exists now 
exactly one $\tilde{v}_s\in\Sigma^+_{\{s\}}$ with $\tau_s^N(\tilde{v}_s)=v'_s$.\\
According to the definition of $\varphi^N$ now for $\tilde{v}_s$ holds:\\
$\varphi^N(\tilde{v}_s)=a'_s$, hence 
 $\varphi^N(x'\tilde{v}_s)=\varphi^N(x')\varphi^N(\tilde{v}_s)=y'a'_s=y$.\\
Because $\tau_t^N(x'\tilde{v}_s)=\tau_t^N(x')=u'_t=u_t$ for
$t\in N\setminus\{s\}$ and 
$\tau_s^N(x'\tilde{v}_s)=\tau_s^N(x')\tau_s^N(\tilde{v}_s)=u'_sv'_s=u_s$ is then 
$x:=x'\tilde{v}_s$ a proper $x\in \Sigma^*_N$ for 
$y=y'a'_s\in \Phi_N^{*}$ for the induction step.
Therewith the proof of the proposition is completed.
\end{proof}
From the above proposition follows the inclusion 
\[\bigcap\limits_{t \in N} (\bar{\tau}_t^N)^{-1}(\varphi(U)) \subset \varphi^N(\bigcap\limits_{t \in N} (\tau_t^N)^{-1}(U)) ,\]
 which completes the proof of Lemma~\ref{lemma:sp-1}.
\end{proof}

\begin{theorem}\label{thm:inverse abstraction-2}
Let $\varphi:\Sigma^*\to\Phi^*$ be an alphabetic homomorphism, $U\subset\Sigma^*$  
and $V\subset\Phi^*$, then 
$\SP(\varphi(U),V) \mbox{ iff } \SP(U,\varphi^{-1}(V))$.
\end{theorem}

\begin{proof} 
\ \\
On account of \eqref{eq:sp-6.1} it only has to be proven that   
$\SP(U,\varphi^{-1}(V))$ implies $\SP(\varphi(U),V)$.\\
\ \\
For each mapping $f:X \to Y$, $A \subset X$ and $B \subset Y$ holds

\begin{equation}\label{eq:sp-6.0}
f(A) \cap B = f(A \cap f^{-1}(B)).
\end{equation}
\ \\
Now Lemma~\ref{lemma:sp-1}, \eqref{eq:pp-63} and \eqref{eq:sp-6.0} imply

\begin{align}\label{eq:sp-63.1}
&\bigcap\limits_{t\in N}(\bar{\tau}_t^N)^{-1}(\varphi(u)) \cap (\bar{\Theta}^N)^{-1}(V) \nonumber \\
&=\varphi^N[\bigcap\limits_{t\in N}(\tau_t^N)^{-1}(U) \cap (\varphi^N)^{-1}((\bar{\Theta}^N)^{-1}(V))] \nonumber \\
&=\varphi^N[\bigcap\limits_{t\in N}(\tau_t^N)^{-1}(U) \cap (\Theta^N)^{-1}(\varphi^{-1}(V))] 
\end{align}
for each non-empty set $N$.\\
\ \\
Because of $\varphi^K(w)=\varphi^N(w)$ for $w\in\Sigma_K^*\subset\Sigma_N^*$
and $\emptyset\neq K\subset N$, 
\eqref{eq:pp-63}, \eqref{eq:pp-63.2}, \eqref{eq:sp-63.1} and $\SP(U,\varphi^{-1}(V))$ imply

\begin{align}\label{eq:sp-63.3}
&\bar{\Pi}_K^N[\bigcap\limits_{t\in N}(\bar{\tau}_t^N)^{-1}(\varphi(u)) \cap (\bar{\Theta}^N)^{-1}(V)] \nonumber \\
&=\varphi^N(\Pi_K^N[\bigcap\limits_{t\in N}(\tau_t^N)^{-1}(U) \cap (\Theta^N)^{-1}(\varphi^{-1}(V))]) \nonumber \\
&\subset \varphi^N((\Theta^N)^{-1}(\varphi^{-1}(V))) \nonumber \\
&=\varphi^N((\varphi^N)^{-1}((\bar{\Theta}^N)^{-1}(V))) \subset (\bar{\Theta}^N)^{-1}(V).
\end{align}
\ \\
\eqref{eq:sp-63.3} shows $\SP(\varphi(U),V)$, which completes the proof of Theorem~\ref{thm:inverse abstraction-2}.

\end{proof}
\ \\
$\SP(U,V)$ can be reduced to a simpler condition than Definition~\ref{def:shuffle-projection}. 
For that purpose an additional notion and lemma is needed.
\ \\
\ \\
Generally for a word $w \in \Sigma_N^*$ $\kappa(w)$ denotes the smallest subset of $N$
such that $w \in \Sigma_{\kappa(w)}^*$. More precisely
\begin{equation}\label{eq:def-kappa}
\kappa(\varepsilon) := \emptyset \mbox{ and } \kappa(wa) := \kappa(w) \cup \{i\}
\mbox{ for } w \in \Sigma_N^* \mbox{ and } a \in \Sigma_{\{i\}} \mbox{ with } i \in N. 
\end{equation} 
\begin{lemma} \label{lemma:struc-rep-proj} 
Let $N$ be an infinite set, $K \subset N$ and
$U \subset \Sigma^*$. Then 
\[\Pi_K^N(\bigcap\limits_{t\in N}(\tau_t^N)^{-1}(U)) \subset \bigcap\limits_{t\in N}(\tau_t^N)^{-1}(U).\]
\end{lemma}
\begin{proof}\ \\
If $\varepsilon \notin U$, then $\bigcap\limits_{t\in N}(\tau_t^N)^{-1}(U) = \emptyset$, and
therefore 
\[ \Pi_K^N(\bigcap\limits_{t\in N}(\tau_t^N)^{-1}(U)) = \emptyset 
\subset \bigcap\limits_{t\in N}(\tau_t^N)^{-1}(U). \]
\ \\
Let now $\varepsilon \in U$, and 
$x \in \Pi_K^N(\bigcap\limits_{t\in N}(\tau_t^N)^{-1}(U))$, then 
$\tau_t^N(x) = \varepsilon \in U$ for $t \in N \setminus K$, and 
$\tau_t^N(x) = \tau_t^N(w) \in U$ for $w \in\bigcap\limits_{t\in N}(\tau_t^N)^{-1}(U)$  
with $\Pi_K^N(w) =x$ and $t \in K$, 
which implies $x \in \bigcap\limits_{t\in N}(\tau_t^N)^{-1}(U)$.
This completes the proof of the lemma.
\end{proof}
\begin{theorem}\label{thm:mod-shuffle-projection}
Let $U,V \in\Sigma^*$, then 
$\SP(U,V)$, iff there exists an infinite countable set $N$ such that 
\begin{equation}\label{eq:mod-shuffle-projection}
\Pi_{N \setminus \{n\}}^N [(\bigcap\limits_{t\in N}(\tau_t^N)^{-1}(U)) 
\cap (\Theta^N)^{-1}(V)] \subset (\Theta^N)^{-1}(V)
\end{equation}
for each $n \in N$.
\end{theorem}
\begin{proof}
Let $K \subset N$ and $w \in \Sigma_N^*$, then by \eqref{eq:def-kappa} holds 
\[\Pi_K^N(w) = \Pi_{(N \setminus \kappa(w)) \cup K}^N(w) = 
\Pi_{N \setminus (\kappa(w) \setminus K)}^N(w).\] 
Therefore $\SP(U,V)$ iff there exists an infinite countable set $N$ such that
\begin{equation}\label{eq:mod-shuffle-projection-126}
\Pi_{N \setminus R}^N [(\bigcap\limits_{t\in N}(\tau_t^N)^{-1}(U)) 
\cap (\Theta^N)^{-1}(V)] \subset (\Theta^N)^{-1}(V)
\end{equation}
for each finite subset $R \subset N$.\ \\
\ \\
Now it is sufficient to show that \eqref{eq:mod-shuffle-projection-126} 
follows from \eqref{eq:mod-shuffle-projection}. 
\begin{proof} [by induction on the cardinality of $R \subset N$] \ \\
\emph{Induction base.}  \\
\eqref{eq:mod-shuffle-projection-126} holds for $R = \emptyset$.\\
\emph{Induction step.}  \\
Let $R = R' \cup \{n\}$ with $n \in N \setminus R'$, then
\begin{equation}\label{eq:mod-shuffle-projection-126'}
\Pi_{N \setminus R}^N = \Pi_{N \setminus R'}^{N \setminus \{n\}} \circ \Pi_{N \setminus \{n\}}^N.
\end{equation}
On account of $\Pi_{N \setminus R'}^N(L) = \Pi_{N \setminus R'}^{N \setminus \{n\}}(L)$ 
for $L \subset \Sigma_{N \setminus \{n\}}^* \subset \Sigma_N^*$ 
\eqref{eq:mod-shuffle-projection-126'} implies 
\begin{align}\label{eq:mod-shuffle-projection-127}
&\Pi_{N \setminus R}^N [(\bigcap\limits_{t\in N}(\tau_t^N)^{-1}(U)) 
\cap (\Theta^N)^{-1}(V)] = \nonumber \\   
&\Pi_{N \setminus R'}^N[\Pi_{N \setminus \{n\}}^N[(\bigcap\limits_{t\in N}(\tau_t^N)^{-1}(U)) 
\cap (\Theta^N)^{-1}(V)]].
\end{align}
By Lemma~\ref{lemma:struc-rep-proj} holds
\[\Pi_K^N(\bigcap\limits_{t\in N}(\tau_t^N)^{-1}(U)) \subset \bigcap\limits_{t\in N}(\tau_t^N)^{-1}(U),\]
and therefore
 \begin{equation}\label{eq:mod-shuffle-projection-129}
\Pi_{N \setminus \{n\}}^N[(\bigcap\limits_{t\in N}(\tau_t^N)^{-1}(U)) 
\cap (\Theta^N)^{-1}(V)] 
\subset \bigcap\limits_{t\in N}(\tau_t^N)^{-1}(U).
\end{equation}
Now \eqref{eq:mod-shuffle-projection} and \eqref{eq:mod-shuffle-projection-129}
imply
 \begin{equation}\label{eq:mod-shuffle-projection-130}
\Pi_{N \setminus \{n\}}^N[(\bigcap\limits_{t\in N}(\tau_t^N)^{-1}(U)) 
\cap (\Theta^N)^{-1}(V)] 
\subset (\bigcap\limits_{t\in N}(\tau_t^N)^{-1}(U)) 
\cap (\Theta^N)^{-1}(V).
\end{equation}
From \eqref{eq:mod-shuffle-projection-127}, \eqref{eq:mod-shuffle-projection-130} 
and the induction hypothesis it follows
\begin{align}
&\Pi_{N \setminus R}^N[(\bigcap\limits_{t\in N}(\tau_t^N)^{-1}(U)) 
\cap (\Theta^N)^{-1}(V)] 
\subset \nonumber \\ 
&\Pi_{N \setminus R'}^N[(\bigcap\limits_{t\in N}(\tau_t^N)^{-1}(U)) 
\cap (\Theta^N)^{-1}(V)] 
\subset (\Theta^N)^{-1}(V), \nonumber
\end{align}
which completes the induction step and the proof of Theorem~\ref{thm:mod-shuffle-projection}.
\end{proof}
\end{proof}
\ \\
\ \\
Because of $\bigcap\limits_{t\in \dsN}(\tau_t^\dsN)^{-1}(U) = \emptyset$ for 
$U \subset \Sigma^+$, then trivially holds $\SP(U,V)$ for each $V \subset \Sigma^*$. 
Therefore in the following sections we consider $\SP(P \cup \{\varepsilon\},V)$ 
for $P,V \subset \Sigma^*$.
\section{Iterated Shuffle Products}
Definition~\ref{def:shuffle-projection} and the examples of scalable systems considered so far are related to iterated 
shuffle products.

\begin{definition}[iterated shuffle product $P^\shuffle$]\label{def:itshuff}
For $P\subset \Sigma^*$ let\\ 
\ \\
$P^\shuffle := \Theta^\dsN [\bigcap\limits_{t \in \dsN} (\tau_t^\dsN)^{-1}(P \cup \{\varepsilon\})]$.\\ 
$P^\shuffle$ is called the \emph{iterated shuffle product of $P$}.  
\end{definition}
An immediate consequence of this definition is
\begin{equation}\label{eq:rem0-itshuff}
\emptyset^\shuffle = \{\varepsilon\}^\shuffle = \{\varepsilon\}, \ 
P \cup \{\varepsilon\} \subset P^\shuffle \mbox{ and } 
P^\shuffle \subset L^\shuffle \mbox{ for } P \subset L\subset \Sigma^*.
\end{equation}
For an alphabetic homomorphism $\varphi:\Sigma^*\to\Phi^*$ and $L \subset \Sigma^*$  holds $xy \in \varphi(L)$ 
iff there exist  $u,v\in\Sigma^*$ with $x=\varphi(u)$, $y=\varphi(v)$ and $uv \in L$. This implies
\begin{equation}\label{eq:sp-98'}
\varphi(\pre(L)) = \pre(\varphi(L) \mbox{ for each } L \subset \Sigma^*.
\end{equation}
where $\pre(M)$ denotes the set of all prefixes of words $w \in M$.
\ \\
\ \\
As $\Theta^\dsN$ and $\tau_t^\dsN$ are alphabetic homomorphisms, \eqref{eq:sp-98'} implies  
\begin{equation}\label{eq:rem-itshuff}
\pre(P^\shuffle)=(\pre(P))^\shuffle.
\end{equation}
\begin{example}\label{ex:itshuff}
Let $P = \{ab\}$, then
$aabb \in P^\shuffle$, because $aabb=\Theta^\dsN(a_1a_2b_2b_1)$,
$\tau_1^\dsN(a_1a_2b_2b_1)=\tau_2^\dsN(a_1a_2b_2b_1)=ab\in P$ and 
$\tau_t^\dsN(a_1a_2b_2b_1)=\varepsilon$ for $t\in \dsN \setminus \{1,2\}$.\\
\ \\
$a_1a_2b_2b_1$ is called a \emph{structured representation} of $aabb$.\\
In this term $\SP(P\cup \{\varepsilon\},V)$ is a property of a certain set of structured representations, which implies
\end{example}
\begin{theorem}\label{thm:SP(itshuff)}
Let $P, V\subset\Sigma^*$, then 
$\SP(P\cup \{\varepsilon\},V) \mbox{ implies } \SP(P^\shuffle,V)$.
\end{theorem}
For the proof of Theorem~\ref{thm:SP(itshuff)} additional notions and three lemmas from \cite{OR2010t} are needed.
Let $S$ and $T$ be non-empty sets.
For each $\emptyset\neq S'\subset S$ and $\emptyset\neq T'\subset T$ let
\[\Theta_{S'}^{S'\times T'}: 
\Sigma_{S'\times T'}^* \to \Sigma_{S'}^* \mbox{ with } 
\Theta_{S'}^{S'\times T'}(a_{(s,t)}):=a_s \mbox{ for each } a_{(s,t)}\in \Sigma_{S'\times T'} \mbox{ and}\]
\[\Theta_{T'}^{S'\times T'}: 
\Sigma_{S'\times T'}^* \to \Sigma_{T'}^* \mbox{ with }
\Theta_{T'}^{S'\times T'}(a_{(s,t)}):=a_t \mbox{ for each } a_{(s,t)}\in \Sigma_{S'\times T'}.\]

\begin{lemma} [Shuffle-lemma 1] \label{lemma:shuffle-lemma1}\ \\
Let $S$, $T$ be non-empty sets and $M\subset \Sigma^*$, then\\
$\bigcap\limits_{s\in S}(\tau_s^S)^{-1}[\Theta^T(\bigcap\limits_{t\in T}(\tau_t^T)^{-1}(M))] =
\Theta_S^{S\times T}[\bigcap\limits_{(s,t)\in S \times T}(\tau_{(s,t)}^{S \times T})^{-1}(M)]$,\\ 
which implies\\
$\Theta^{S} [\bigcap\limits_{s\in S}(\tau_s^S)^{-1}[\Theta^T(\bigcap\limits_{t\in T}(\tau_t^T)^{-1}(M))]] =
\Theta^{S\times T}[\bigcap\limits_{(s,t)\in S \times T}(\tau_{(s,t)}^{S \times T})^{-1}(M)]$,\\
because of\\
$\Theta^{S\times T}= \Theta^{S} \circ \Theta_S^{S\times T}$.
\end{lemma}

\begin{lemma} [Shuffle-lemma 2] \label{lemma:shuffle-lemma2}\ \\
Let $S$, $T$ be non-empty sets and $M\subset \Sigma^*$. If a bijection between $S$ and $T$ exists, then 
$\Theta^S[\bigcap\limits_{s\in S}(\tau_s^S)^{-1}(M)] =\Theta^T[\bigcap\limits_{t\in T}(\tau_t^T)^{-1}(M)]$.
\end{lemma}

\begin{definition}[structured representation]\label{def:strucrep}\ \\
Let $S$ be a non-empty set and $M\subset \Sigma^*$.
For each $x\in \Theta^S [\bigcap\limits_{s\in S}(\tau_s^S)^{-1}(M)]$ there exists
$u\in \bigcap\limits_{s\in S}(\tau_s^S)^{-1}(M)$ such that $x=\Theta^S(u)$.
We call $u$ a \emph{structured representation of $x$ w.r.t. $S$ and $M$}.\\
\ \\
For $x\in \Sigma^*$ let $SR_M^S(x):=(\Theta^S)^{-1}(x)\cap [\bigcap\limits_{s\in S}(\tau_s^S)^{-1}(M)]$.
It is the set of all structured representations of $x$ w.r.t. $S$ and $M$.
\end{definition}
\noindent \textbf{Remark.} Now $x\in P^\shuffle$ iff there exists an infinite countable set $S$ with
$SR^S_{(P\cup\{\epsilon\})}(x)\neq \emptyset$. Therefore in Definition~\ref{def:itshuff} $\dsN$ 
can be replaced by any infinite countable set $N$.

\begin{lemma} [Shuffle-lemma 3] \label{lemma:shuffle-lemma3}\ \\
Let $S$, $T$ be non-empty sets, $M\subset \Sigma^*$, and 
$y\in \Sigma_{S\times T}^*$ with 
$\tau_{(s,t)}^{S\times T}(y)\in M$ for each\\
$(s,t)\in S\times T$ and $x=\Theta_S^{S\times T}(y)\in \Sigma_S^*$, then\\ 
$\Pi_{S'\times T}^{S\times T}(y)\in SR_M^{S'\times T}(\Theta^{S'}(\Pi_{S'}^S(x)))$
for each $\emptyset\neq S'\subset S$.
\end{lemma}
\noindent \textbf{Remark.} The hypotheses of this lemma are given by
lemma~\ref{lemma:shuffle-lemma1}.
\begin{proof} [Proof of Theorem~\ref{thm:SP(itshuff)}]\ \\
\ \\
Let $x\in\bigcap\limits_{s\in S}(\tau_s^S)^{-1}(P^\shuffle) \cap (\Theta^S)^{-1}(V)$, where $S$ is a countable infinite set, then
$x\in\bigcap\limits_{s\in S}(\tau_s^S)^{-1}[\Theta^T(\bigcap\limits_{t\in T}(\tau_t^T)^{-1}(P\cup \{\varepsilon\}))]$, 
where $T$ is a countable infinite set.
By Lemma~\ref{lemma:shuffle-lemma1} there exists 
$y \in \bigcap\limits_{(s,t)\in S \times T}(\tau_{(s,t)}^{S \times T})^{-1}(P\cup \{\varepsilon\})$
with $x = \Theta_S^{S \times T}(y)$.
This implies 
$y \in \bigcap\limits_{(s,t)\in S \times T}(\tau_{(s,t)}^{S \times T})^{-1}(P\cup \{\varepsilon\}) \cap (\Theta^{S \times T})^{-1}(V)$ 
because of $\Theta^S(x) \in V$ and $\Theta^{S\times T}= \Theta^{S} \circ \Theta_S^{S\times T}$. 
Now, by the assumption $\SP(P\cup \{\varepsilon\},V)$ holds
\begin{equation}\label{eq:sp-214}
\Pi_{S'\times T}^{S\times T}(y)\in (\Theta^{S \times T})^{-1}(V)
\mbox{ for each } \emptyset\neq S'\subset S.
\end{equation}
As now $x$ and $y$ fulfill the assumptions of Lemma~\ref{lemma:shuffle-lemma3}, it follows 
\begin{equation}\label{eq:sp-215}
\Theta^{S'\times T}(\Pi_{S'\times T}^{S\times T}(y)) = \Theta^{S'}(\Pi_{S'}^{S}(x)).
\end{equation}
Because of \[\Theta^{S'}(\Pi_{S'}^{S}(x)) = \Theta^{S}(\Pi_{S'}^{S}(x))\] and 
\[\Theta^{S'\times T}(\Pi_{S'\times T}^{S\times T}(y)) = \Theta^{S\times T}(\Pi_{S'\times T}^{S\times T}(y))\]
\eqref{eq:sp-214} and \eqref{eq:sp-215} imply $\Pi_{S'}^{S}(x) \in (\Theta^S)^{-1}(V)$
for each $\emptyset\neq S'\subset S$, which completes the proof of Theorem~\ref{thm:SP(itshuff)}.
\end{proof}
\eqref{eq:sp-5}, \eqref{eq:sp-6}, Theorem~\ref{thm:inverse abstraction} and Theorem~\ref{thm:SP(itshuff)} 
show that in many cases it is sufficient to prove $SP(U,V)$ for very simple $U$.
On account of our focus on system behavior, we are especially interested in
$SP(U,V)$ for prefix closed languages $U$ and $V$. 
\ \\
\ \\
In Definition~\ref{def:itshuff} the iterated shuffle product is represented by the homomorphic image of a set of 
structured representations. To get a deeper insight into the property $\SP(P\cup \{\varepsilon\},V)$, in the next section we will 
represent $P^\shuffle$ by an homomorphic image of a set of computations of a certain automaton. 
For this purpose we need a ``bracketed coding'' of words.
\begin{definition}\label{def:sp-92}\ \\
Together with an alphabet $\Sigma$ we consider 
four pairwise disjoint copies of $\Sigma$, namely
$\tilde{\Sigma}$, $\mathring{\Sigma}$, $\bar{\Sigma}$, $\tilde{\bar{\Sigma}}$,
and a homomorphism 
$\wedge : \hat{\Sigma}^*\to \Sigma^*$ with
$\hat{\Sigma} := \tilde{\Sigma} \ \dotcup\  \mathring{\Sigma}\  \dotcup\ \bar{\Sigma}\ \dotcup\ \tilde{\bar{\Sigma}}$ and 
$\wedge(\tilde{a}) := \wedge(\mathring{a}):= \wedge(\bar{a}) := \wedge(\tilde{\bar{a}}) := a$ for each $a \in \Sigma$, 
where $\tilde{a}$, $\mathring{a}$, $\bar{a}$ and $\tilde{\bar{a}}$
are the corresponding copies of a letter $a\in \Sigma$.
\end{definition}
For words $u \in P \subset \Sigma^+$ the four alphabets are used to characterize start-, inner-, end-, or start-end letters of $u$. 
\begin{definition}\label{def:length}\ \\
Let $|x| \in \dsN_0$ denotes the \emph{length} of a word $x \in \Sigma^*$, defined by $|\varepsilon|:=0$ and
$|xa|:=|x|+1$ for $a \in \Sigma$ and $x \in \Sigma^*$.
\end{definition}
The following definition depends on the fact that each 
$u \in \Sigma^*$ with $|u|>1$ can be uniquely represented by $u=awb$ with $a,b \in \Sigma$ and $w \in \Sigma^*$.  
\begin{definition}\label{def:sp-95}\ \\
Let $\langle \rangle: \Sigma^* \to \{\varepsilon\} \cup \tilde{\bar{\Sigma}} \cup \tilde{\Sigma}\mathring{\Sigma}^*\bar{\Sigma}$
be the mapping defined by $\langle \rangle(\varepsilon):= \varepsilon, \langle \rangle(a):= \tilde{\bar{a}}$ for $a \in \Sigma$ 
and $\langle \rangle(awb):= \tilde{a}\mathring{w}\bar{b}$ for $a,b \in \Sigma$ and $w \in \Sigma^*$, 
where $\mathring{w}$ is defined by $\mathring{w} \in \mathring{\Sigma}^*$ and $\wedge(\mathring{w}) = w$ for each $w \in \Sigma^*$.\\
For short we write $\langle u \rangle$ instead of $\langle \rangle(u)$ for each $u \in \Sigma^*$.
\end{definition}
For each $y \in \{\varepsilon\} \cup \tilde{\bar{\Sigma}} \cup \tilde{\Sigma}\mathring{\Sigma}^*\bar{\Sigma}$
holds $\langle \wedge(y) \rangle = y$, and for each $x \in \Sigma^*$ holds $\wedge(\langle x \rangle) = x$.
Therefore $\langle \rangle$ is a bijection with
\begin{equation}\label{eq:sp-96,97,98}
\langle \rangle^{-1} = \wedge_{|\{\varepsilon\} \cup \tilde{\bar{\Sigma}} \cup \tilde{\Sigma}\mathring{\Sigma}^*\bar{\Sigma}}
\mbox{ and } |\langle w \rangle| = w \mbox{ for each } w \in \Sigma^*. 
\end{equation}
The bijection $\langle \rangle$ formalizes the ``bracketed coding'' of words.
\ \\
\\
By \eqref{eq:sp-96,97,98} and \eqref{eq:sp-98'} holds $\wedge(\langle U \rangle) = U$
and $\wedge(\pre(\langle U \rangle)) = \pre(U)$ for each $U \subset \Sigma^*$.
Therefore Theorem~\ref{thm:inverse abstraction-2} implies
\begin{corollary}\label{cor:sp-4'}\ \\
For each $U,V\subset\Sigma^*$ holds    
$\SP(U,V)$ iff $\SP(\langle U \rangle,\wedge^{-1}(V))$, and
$\SP(\pre(U),V)$ iff $\SP(\pre(\langle U \rangle),\wedge^{-1}(V))$.
\end{corollary}
The following theorem together with its corollary prepares the automata representations 
of iterated shuffle products.
\begin{theorem} \label{thm:SP(hom-shuffle)} \ \\
Let $\varphi:\Sigma^*\to\Phi^*$ be an alphabetic homomorphism and $P\subset\Sigma^*$,
then holds $\varphi(P^\shuffle)=(\varphi(P))^\shuffle$.
\end{theorem}
\begin{proof} \ \\
Let $N$ be an infinite countable set. 
Let $\varphi^N:\Sigma_N^*\to\Phi_N^*$,
$\bar{\tau}_n^N:\Phi_N^*\to\Phi^*$ and $\bar{\Theta}^N:\Phi_N^*\to\Phi^*$
be defined as in context of Lemma~\ref{lemma:sp-1}.
Because of \eqref{eq:pp-63} holds 
\begin{equation}\label{eq:*x12'}
\mu (P^\shuffle) = \bar{\Theta}^N [\varphi^N(\bigcap\limits_{t \in N} (\tau_t^N)^{-1}(P \cup \{\varepsilon\}))].
\end{equation}
From this it follows that $\mu(P^\shuffle)=(\mu(P))^\shuffle$ if 
the following equation holds:
\begin{equation} \label{eq:*x13}
\varphi^N(\bigcap\limits_{t \in N} (\tau_t^N)^{-1}(P \cup \{\varepsilon\})) = 
\bigcap\limits_{t \in N} (\bar{\tau}_t^N)^{-1}(\varphi(P) \cup \{\varepsilon\})
\end{equation}
\begin{proof} Proof of equation~\eqref{eq:*x13}:\\
Because of \eqref{eq:pp-63} holds
\[\bar{\tau}_t^N (\varphi^N(x)) = \varphi(\tau_t^N(x))\in
\varphi(P\cup \{\varepsilon\})= \varphi(P)\cup \{\varepsilon\} \]
for each 
$x\in \bigcap\limits_{t \in N} (\tau_t^N)^{-1}(P \cup \{\varepsilon\})$ and $t\in N$,
which implies
\[\varphi^N(\bigcap\limits_{t \in N} (\tau_t^N)^{-1}(P \cup \{\varepsilon\})) 
\subset \bigcap\limits_{t \in N} (\bar{\tau}_t^N)^{-1}(\mu(P) \cup \{\varepsilon\}) .\]
The other inclusion of equation~\eqref{eq:*x13}follows from Lemma~\ref{lemma:sp-1}, 
which completes the proof of equation~\eqref{eq:*x13} and of Theorem~\ref{thm:SP(hom-shuffle)}.
\end{proof}
\end{proof}
Because of $P = \wedge(\langle P \rangle)$ and $\pre(P) = \wedge(\pre(\langle P \rangle))$
Theorem~\ref{thm:SP(hom-shuffle)} implies
\begin{corollary}\label{cor:SP(hom-shuffle)}\ \\
Let $P\subset\Sigma^*$,
then $P^\shuffle = \wedge(\langle P \rangle^\shuffle)$, and
$(\pre(P))^\shuffle = \wedge((\pre(\langle P \rangle))^\shuffle)$.
\end{corollary}
Therefore Corollary~\ref{cor:SP(hom-shuffle)} reduces automata representations of 
$P^\shuffle$ rsp. $(\pre(P))^\shuffle$ to automata representations of
$\langle P \rangle^\shuffle$ rsp. $(\pre(\langle P \rangle))^\shuffle$.
\section{Automata Representations of Iterated Shuffle Products}\label{sec:automata}
Automata representations of iterated shuffle products are well known. See for example 
\cite{BjorklundB07} and \cite{Jedrzejowicz99}, where multicounter automata are considered.
Therefore the purpose of this section is not to introduce a new automaton concept, but
to establish notions for further investigations of $\SP(P,V)$ based on computations of 
these automata.
On account of Corollary~\ref{cor:SP(hom-shuffle)} we start with an automaton representation
for $(\pre(\langle P \rangle))^\shuffle$.  
\ \\
\ \\
Let $P\subset \Sigma^*$ and $\bbP=(\Sigma,Q,\delta,q_0,F)$ be
a (not necessarily finite) deterministic
automaton recognizing $P$, where
$\delta: Q \times \Sigma \to Q$ is a partial function,
$q_0 \in Q$ and $F\subset Q$.
As usual, $\delta$ is extended to a partial function $\delta: Q \times \Sigma^* \to Q$.
For simplicity we assume 
$P \neq \emptyset$ and $\delta(q_0,\pre(P))=Q$.
\ \\
\ \\
Moreover, we take this set of conditions as a general assumption for the rest of the paper.
\ \\
\ \\
The idea to define a semiautomaton (automaton without final states \cite{berstel79}) 
$\hat{\bbP}_\shuffle$ recognizing $(\pre(\langle P \rangle))^\shuffle$ is the following:
Each computation in $\hat{\bbP}_\shuffle$ ``correspond'' to a ``shuffled run'' of several 
not necessarily recognizing computations in $\bbP$, which we call \emph{``elementary computations''}. 
For each $q \in Q$ the states of $\hat{\bbP}_\shuffle$ store the number of ``elementary computations'' 
which just have reached the state $q$ in such a ``shuffled run'' of ``elementary computations''.
\ \\
\ \\
Formally, the state set of $\hat{\bbP}_\shuffle$ is $\dsN_0^Q$,
the set of all functions $f : Q \to \dsN_0$.
\ \\
\ \\
Let $0 \in \dsN_0^Q$ be defined by $0(q) := 0$ for each $q \in Q$.
For $q \in Q$ and $k\in \dsN$ let $k_q \in \dsN_0^Q$ be defined by
$k_q(x) :=  \left\{ \begin{array}{r@{\ }cl}
k & |\ & x =q\\
0 & |\ & x \in Q \setminus \{q\}
\end{array} \right. .$
\ \\
\ \\
For $f, g \in \dsN_0^Q$ let
\begin{itemize}
\item $f \geqslant g$ iff $f(x) \geqslant g(x)$ for each $x \in Q$,
\item $f + g \in \dsN_0^Q$ with  $(f + g)(x) := f(x) + g(x)$ for each $x \in Q$, and
\item for $f \geqslant g$,
  $f - g \in \dsN_0^Q$ with  $(f - g)(x) := f(x) - g(x)$ for each $x \in Q$.
\end{itemize}
The state transition relation $\hat{\shuffle}_\bbP$ of $\hat{\bbP}_\shuffle$ is composed of 
four disjunct subsets whose elements describe 
\begin{itemize}
\item the ``entry into a new elementary computation'',
\item the ``transition within an open elementary computation'',
\item the ``completion of an open elementary computation'',
\item the ``entry into a new elementary computation with simultaneous completion of this elementary computation''.
\end{itemize}
\begin{definition} [$\hat{S}$-automaton $\hat{\bbP}_\shuffle$]\label{def:shuffle-automaton^}
\ \\
$\hat{\bbP}_\shuffle = (\hat{\Sigma},\dsN_0^Q,\hat{\shuffle}_\bbP,0)$ w.r.t. $\bbP$ is a semiautomaton with
an infinite state set $\dsN_0^Q$, the initial state $0$ and a state transition relation 
$ \hat{\shuffle}_\bbP \subset \dsN_0^Q \times \hat{\Sigma} \times \dsN_0^Q $ defined by
\begin{align*}
\hat{\shuffle}_\bbP :=& \tilde{\shuffle}_\bbP \ \dotcup \ \mathring{\shuffle}_\bbP \ \dotcup 
\ \bar{\shuffle}_\bbP \ \dotcup \ \tilde{\bar{\shuffle}}_\bbP 
\mbox{ with }\\
\tilde{\shuffle}_\bbP :=& \{(f,x,f+1_p) \in \dsN_0^Q \times \tilde{\Sigma} \times \dsN_0^Q \ | \ \delta(q_0,\wedge(x))=p 
\mbox{ and it exists } b\in \Sigma \mbox{ such } \\
&\mbox{ that } \delta(p,b) \mbox{ is defined} \},\\
\mathring{\shuffle}_\bbP :=& \{(f,x,f+1_p-1_q) \in \dsN_0^Q \times \mathring{\Sigma} \times \dsN_0^Q \ | \ f \geqslant 1_q, 
\delta(q,\wedge(x))=p \mbox{ and it exists } \\
&b\in \Sigma \mbox{ such that } \delta(p,b) \mbox{ is defined} \},\\
\bar{\shuffle}_\bbP :=& \{(f,x,f-1_q) \in \dsN_0^Q \times \bar{\Sigma} \times \dsN_0^Q \ | \ f \geqslant 1_q \mbox{ and } 
\delta(q,\wedge(x)) \in F \} \mbox{ and }\\
\tilde{\bar{\shuffle}}_\bbP :=& \{(f,x,f) \in \dsN_0^Q \times \tilde{\bar{\Sigma}} \times \dsN_0^Q \ |\ \delta(q_0,\wedge(x)) \in F \}.
\end{align*}
\end{definition}
Generally $\hat{\bbP}_\shuffle$ is an infinite nondeterministic semiautomaton.
\begin{example}\label{ex:example7}
$P = \{abc, abbc\}$
\begin{figure}[h]
\centering
\begin{tikzpicture}[->,>=stealth',shorten >=1pt,auto,node distance=2cm,semithick,initial text=]
   \node[state,initial,inner sep=1pt,minimum size=6mm] (k1)  {\small $\I$};
   \node[state,inner sep=1pt,minimum size=6mm] (k2) at (1.5,0) {\small $\II$};
   \node[state,inner sep=1pt,minimum size=6mm] (k3) at (3,0) {\small $\III$};
   \node[state,accepting,inner sep=1pt,minimum size=6mm] (k4) at (4.5,0) {\small $\IV$};
   \node[state,inner sep=1pt,minimum size=6mm] (k5) at (3,-1.5) {\small $\V$};
   \path 
         (k1) edge node {\small $a$} (k2)
         (k2) edge node {\small $b$} (k3)
         (k3) edge node {\small $c$} (k4)
         (k3) edge node {\small $b$} (k5)
         (k5) edge [swap] node {\small $c$} (k4)
		;
\end{tikzpicture}
\caption{Automaton $\bbP$ recognizing $P$}\label{fig:ex7}
\end{figure}
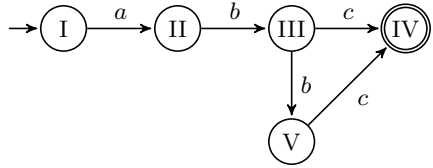
Two computations in $\hat{\bbP}_\shuffle$:
\[0 \stackrel{\tilde{a}}{\longrightarrow} 1_{\small \II} 
\stackrel{\mathring{b}}{\longrightarrow} 1_{\small \III}
\stackrel{\tilde{a}}{\longrightarrow} 1_{\small \III}+1_{\small \II}
\stackrel{\mathring{b}}{\longrightarrow} 1_{\small \V}+1_{\small \II} \ldots{}\]
\[0 \stackrel{\tilde{a}}{\longrightarrow} 1_{\small \II} 
\stackrel{\mathring{b}}{\longrightarrow} 1_{\small \III}
\stackrel{\tilde{a}}{\longrightarrow} 1_{\small \III}+1_{\small \II}
\stackrel{\mathring{b}}{\longrightarrow} 2_{\small \III} \ldots{}\]
\end{example}
$\hat{A}_\bbP \subset \hat{\shuffle}_\bbP^*$ denotes the set of all paths in $\hat{\bbP}_\shuffle$ 
starting with the initial state $0$ and including the empty path $\varepsilon$. 
For $w\in \hat{A}_\bbP$, $\hat{Z}_\bbP(w)$ denotes the final state of the path $w$ and 
$\hat{Z}_\bbP(\varepsilon):=0$. 
Formally the prefix closed language $\hat{A}_\bbP$ and the function
$\hat{Z}_\bbP : \hat{A}_\bbP \to \dsN_0^Q$ are defined inductively by 
\begin{equation} \label{eq:def A^,Z^}
\varepsilon\in \hat{A}_\bbP, \ \hat{Z}_\bbP(\varepsilon):=0, 
\ w(f,x,g)\in \hat{A}_\bbP \mbox{ and } 
\hat{Z}_\bbP(w(f,x,g)):=g 
\end{equation}
for $w\in \hat{A}_\bbP$, $\hat{Z}_\bbP(w)=f$ and $(f,x,g)\in \hat{\shuffle}_\bbP$.
\ \\
\ \\
Let the function $\hat{\alpha}_\bbP :\hat{A}_\bbP \to \hat{\Sigma}^*$ be inductively defined by
\begin{equation} \label{eq:def alpha^}
\hat{\alpha}_\bbP(\varepsilon):=\varepsilon \mbox{ and }
\hat{\alpha}_\bbP(w(f,x,g)):=\hat{\alpha}_\bbP(w)x
\end{equation}
for $w(f,x,g)\in \hat{A}_\bbP$ and $(f,x,g)\in \hat{\shuffle}_\bbP$.
$\hat{\alpha}_\bbP(u)$ is called the \emph{label of a path u}.
\begin{definition}\label{def:sp-93,94}\ \\
Let $N$ be an infinite countable set. 
For $I' \subset I \subset N$ and $t\in N$ let  
$\hat{\Sigma}_{\{t\}} := \tilde{\Sigma}_{\{t\}} \ \dotcup\  \mathring{\Sigma}_{\{t\}}\  
\dotcup\ \bar{\Sigma}_{\{t\}}\ \dotcup\ \tilde{\bar{\Sigma}}_{\{t\}}$ , 
$\hat{\tau}_t^I: \hat{\Sigma}_I^* \to \hat{\Sigma}^*$,  
$\hat{\Theta}^I: \hat{\Sigma}_I^* \to \hat{\Sigma}^*$ and 
$\hat{\Pi}_{I'}^{I}: \hat{\Sigma}_{I}^* \to \hat{\Sigma}_{I'}^*$
be defined according to the definitions of $\hat{\Sigma}$, $\Sigma_{\{t\}}$, $\tau_t^I$, $\Theta^I$ and $\Pi_{I'}^{I}$,
where
$\hat{\Sigma}_I := \dot{\bigcup\limits_{s \in I}} \hat{\Sigma}_{\{s\}}$. 
\end{definition}
The key to prove that $\hat{\bbP}_\shuffle$ recognizes 
$(\pre(\langle P \rangle))^\shuffle = 
\hat{\Theta}^N[\bigcap\limits_{t \in N} (\hat{\tau}_t^N)^{-1}(\pre(\langle P \rangle))]$ is 
to define an appropriate function
$\hat{c}_\bbP: \bigcap\limits_{t \in N} (\hat{\tau}_t^N)^{-1}(\pre(\langle P \rangle)) \to \hat{A}_\bbP$.
For that purpose we first consider the function
$\hat{n}_\bbP: \bigcap\limits_{t \in N} (\hat{\tau}_t^N)^{-1}(\pre(\langle P \rangle)) \to \dsN_0^Q$, 
defined by
\begin{equation}\label{eq:def-np^}
\hat{n}_\bbP(x)(q) := \#(\{t\in N \ |\ \delta(q_0,\wedge(\hat{\tau}_t^{N}(x)))=q \mbox{ and } 
\hat{\tau}_t^{N}(x) \notin \langle P \rangle \cup \{\varepsilon\}\})
\end{equation}
for each $x \in  \bigcap\limits_{t \in N} (\hat{\tau}_t^{N})^{-1}(\pre(\langle P \rangle))$ 
and $q \in Q$, where $\#(M)$ denotes the cardinality of a set $M$.
\ \\
\ \\
As in \eqref{eq:rem-itshuff} it holds 
\begin{equation}\label{eq:pre-itshuff^}
\pre(\bigcap\limits_{t \in N} (\hat{\tau}_t^N)^{-1}(\langle P \cup \{\varepsilon\} \rangle)) =
\bigcap\limits_{t \in N} (\hat{\tau}_t^N)^{-1}(\pre(\langle P \rangle)).
\end{equation}
This shows that $\bigcap\limits_{t \in N} (\hat{\tau}t^N)^{-1}(\pre(\langle P \rangle))$ is a prefix closed language.
\ \\
\ \\
The following property of $\hat{n}_\bbP$ is the key for the definition of $\hat{c}_\bbP$.
\begin{lemma} \label{lemma:def-cp^}\ \\ 
Let $x\hat{a} \in  \bigcap\limits_{t \in N} (\hat{\tau}_t^{N})^{-1}(\pre(\langle P \rangle))$ 
with $\hat{a} \in \hat{\Sigma}_N$, then 
$(\hat{n}_\bbP(x),\hat{\Theta}^N(\hat{a}),\hat{n}_\bbP(x\hat{a})) \in \hat{\shuffle}_\bbP$.
\end{lemma}
\begin{proof}\ \\
For $I \subset N$ an immediate consequence of Lemma~\ref{lemma:struc-rep-proj} 
and the definitions of $\hat{n}_\bbP$ and $\hat{\shuffle}_\bbP$ is\\
\ \\
\begin{equation}\label{eq:def-cp^1}
\hat{\Pi}_{I}^N[\bigcap\limits_{t \in N} (\hat{\tau}_t^N)^{-1}(\pre(\langle P \rangle))]
\subset \bigcap\limits_{t \in N} (\hat{\tau}_t^N)^{-1}(\pre(\langle P \rangle)) \cap \hat{\Sigma}_I^*, 
\end{equation}
\begin{equation}\label{eq:def-cp^2}
\hat{n}_\bbP (x) =  \hat{n}_\bbP (\hat{\Pi}_{I}^N(x)) + \hat{n}_\bbP (\hat{\Pi}_{N \setminus I}^N(x))
\mbox{ for } x \in \bigcap\limits_{t \in N} (\hat{\tau}_t^N)^{-1}(\pre(\langle P \rangle)) \mbox{ and} 
\end{equation}
\begin{equation}\label{eq:def-cp^3}
(f,b,g) \in \hat{\shuffle}_\bbP
\mbox{ implies } (f+h,b,g+h) \in \hat{\shuffle}_\bbP \mbox{ for } h \in \dsN_0^Q.
\end{equation}
\ \\
\ \\
For $x\hat{a} \in  \bigcap\limits_{t \in N} (\hat{\tau}_t^{N})^{-1}(\pre(\langle P \rangle))$ 
there exists $s \in N$ with $\hat{a} \in \hat{\Sigma}_{\{s\}}$, and therefore 
$\hat{\Pi}_{N \setminus \{s\}}^N(x\hat{a}) = \hat{\Pi}_{N \setminus \{s\}}^N(x)$. 
Now by \eqref{eq:def-cp^1} - \eqref{eq:def-cp^3} it is sufficient to prove the lemma for
$x\hat{a} \in  \bigcap\limits_{t \in N} (\hat{\tau}_t^{N})^{-1}(\pre(\langle P \rangle)) \cap \hat{\Sigma}_{\{s\}}^* =
(\hat{\tau}_s^{\{s\}})^{-1}(\pre(\langle P \rangle)) $, where 
$\hat{\tau}_s^{\{s\}}:\hat{\Sigma}_{\{s\}}^* \to \hat{\Sigma}^*$ is a bijection.
\ \\
\ \\
For $w \in (\hat{\tau}_s^{\{s\}})^{-1}(\pre(\langle P \rangle) \setminus \langle P \rangle \cup \{\varepsilon\}) $ 
holds $\hat{n}_\bbP(w) = 1_q$, with $\delta(q_0,\wedge(\hat{\tau}_s^{\{s\}}(w)))=q$ and 
for $w \in (\hat{\tau}_s^{\{s\}})^{-1}(\langle P \rangle \cup \{\varepsilon\})$ 
holds $\hat{n}_\bbP(w) = 0$. Therefore the definition of $\hat{\shuffle}_\bbP$ immediately implies 
$(\hat{n}_\bbP(x),\hat{\Theta}^N(\hat{a}),\hat{n}_\bbP(x\hat{a})) = 
(\hat{n}_\bbP(x),\hat{\tau}_s^{\{s\}}(\hat{a}),\hat{n}_\bbP(x\hat{a}))
\in \hat{\shuffle}_\bbP$ 
for $x\hat{a} \in (\hat{\tau}_s^{\{s\}})^{-1}(\pre(\langle P \rangle))$, 
which completes the proof of the lemma. 
\end{proof}
Lemma~\ref{lemma:def-cp^} makes the following definition sound:
\begin{definition}\label{def:def-cp^}\ \\
Let the function
$\hat{c}_\bbP: \bigcap\limits_{t \in N} (\hat{\tau}_t^N)^{-1}(\pre(\langle P \rangle)) \to \hat{A}_\bbP$ 
be inductively defined by $\hat{c}_\bbP(\varepsilon) := \varepsilon $ and 
$\hat{c}_\bbP(x\hat{a}) := \hat{c}_\bbP(x)(\hat{n}_\bbP(x),\hat{\Theta}^N(\hat{a}),\hat{n}_\bbP(x\hat{a})) $ 
for 
$x\hat{a} \in  \bigcap\limits_{t \in N} (\hat{\tau}_t^{N})^{-1}(\pre(\langle P \rangle))$ 
with $\hat{a} \in \hat{\Sigma}_N$.
\end{definition}
This definition immediately implies
\begin{theorem}\label{thm:def-cp^}
Let $x \in  \bigcap\limits_{t \in N} (\hat{\tau}_t^{N})^{-1}(\pre(\langle P \rangle))$ then
\begin{subequations}\label{subeq:def-cp^}
\begin{align}
&  \hat{Z}_\bbP(\hat{c}_\bbP(x)) = \hat{n}_\bbP(x) , \label{subeq:def-cp^-1}\\
&  \hat{\alpha}_\bbP(\hat{c}_\bbP(x)) = \hat{\Theta}^N(x) , \label{subeq:def-cp^-2}\\ 
&  |\hat{c}_\bbP(x)| = |x| , \mbox{ and} \label{subeq:def-cp^-3}\\ 
&  \pre(\hat{c}_\bbP(x)) = \hat{c}_\bbP(\pre(x)) . \label{subeq:def-cp^-4}
\end{align}
\end{subequations}
\end{theorem}
To prove surjectivity of $\hat{c}_\bbP$ we need a counterpart of Lemma~\ref{lemma:def-cp^}:
\begin{lemma} \label{lemma:cp^-surj} 
Let $c(f,\hat{b},g) \in \hat{A}_\bbP$ with $(f,\hat{b},g) \in \hat{\shuffle}_\bbP$, 
and $w \in  \bigcap\limits_{t \in N} (\hat{\tau}_t^{N})^{-1}(\pre(\langle P \rangle))$ 
with $\hat{c}_\bbP(w) = c$.\\
\ \\
If $\hat{b} \in \tilde{\Sigma} \dotcup \tilde{\bar{\Sigma}}$, then for each 
$\hat{a} \in \hat{\Sigma}_{N \setminus \kappa(w)}$ with $\hat{\Theta}^N(\hat{a}) = \hat{b}$
holds $w\hat{a} \in  \bigcap\limits_{t \in N} (\hat{\tau}_t^{N})^{-1}(\pre(\langle P \rangle))$ 
and $\hat{c}_\bbP(w\hat{a}) = c(f,\hat{b},g)$.
\ \\
\ \\
If $\hat{b} \in \mathring{\Sigma} \dotcup \bar{\Sigma}$, then 
there exists $\hat{a} \in \hat{\Sigma}_{\kappa(w)}$ with $\hat{\Theta}^N(\hat{a}) = \hat{b}$
such that $w\hat{a} \in  \bigcap\limits_{t \in N} (\hat{\tau}_t^{N})^{-1}(\pre(\langle P \rangle))$ 
and $\hat{c}_\bbP(w\hat{a}) = c(f,\hat{b},g)$.
\end{lemma}
\begin{proof}
\ \\
By the definition of $\hat{\shuffle}_\bbP$ each $(f,\hat{b},g) \in \tilde{\shuffle}_\bbP \dotcup \tilde{\bar{\shuffle}}_\bbP$ 
can be represented by
\begin{equation}\label{eq:cp^-surj-0}
(f,\hat{b},g) = (f,\hat{b},f+h) \mbox{ with } 
(0,\hat{b},h) \in \tilde{\shuffle}_\bbP \dotcup \tilde{\bar{\shuffle}}_\bbP, 
\end{equation}
and each $(f,\hat{b},g) \in \mathring{\shuffle}_\bbP \dotcup \bar{\shuffle}_\bbP$ 
can be represented by
\begin{equation}\label{eq:cp^-surj-1}
(f,\hat{b},g) = (f'+1_q,\hat{b},f'+k) \mbox{ with } q \in Q \mbox{ and }
(1_q,\hat{b},k) \in \mathring{\shuffle}_\bbP \dotcup \bar{\shuffle}_\bbP. 
\end{equation} 
In these representations $h$ is uniquely determined by $\hat{b}$, and 
$k$ is uniquely determined by $q$ and $\hat{b}$. More precisely: 
There exist partial functions $\tilde{\bar{\delta}}_\bbP: \tilde{\Sigma} \dotcup \tilde{\bar{\Sigma}} \to \dsN_0^Q$ 
and $\mathring{\bar{\delta}}_\bbP: Q \times (\mathring{\Sigma} \dotcup \bar{\Sigma}) \to \dsN_0^Q$ such that
\begin{equation}\label{eq:cp^-surj-2}
(0,\hat{b},h) \in \tilde{\shuffle}_\bbP \dotcup \tilde{\bar{\shuffle}}_\bbP \mbox{ iff }
h = \tilde{\bar{\delta}}_\bbP(\hat{b}), \mbox{ and }
(1_q,\hat{b},k) \in \mathring{\shuffle}_\bbP \dotcup \bar{\shuffle}_\bbP \mbox{ iff }
k = \mathring{\bar{\delta}}_\bbP(q,\hat{b}).
\end{equation}
Let $w\hat{a} \in  \bigcap\limits_{t \in N} (\hat{\tau}_t^{N})^{-1}(\pre(\langle P \rangle))$ with 
$\hat{\Theta}^N(\hat{a}) = \hat{b} \in \tilde{\Sigma} \dotcup \tilde{\bar{\Sigma}}$, then 
$\hat{a} \in \hat{\Sigma}_{N \setminus \kappa(w)}$.
Now \eqref{eq:def-cp^2} and the definition of $\tilde{\bar{\delta}}_\bbP$ imply
\begin{equation}\label{eq:cp^-surj-3}
\hat{n}_\bbP(w\hat{a}) = \hat{n}_\bbP(w) + \tilde{\bar{\delta}}_\bbP(\hat{b}).
\end{equation}
Let $w\hat{a} \in  \bigcap\limits_{t \in N} (\hat{\tau}_t^{N})^{-1}(\pre(\langle P \rangle))$ with 
$\hat{\Theta}^N(\hat{a}) = \hat{b} \in \mathring{\Sigma} \dotcup \bar{\Sigma}$, then 
there exists $s \in \kappa(w)$ such that $\hat{\tau}_s^{N}(w) \neq \varepsilon$ and 
$\hat{\tau}_s^{N}(w)\hat{b} \in \pre(\langle P \rangle)$. 
Now \eqref{eq:def-cp^2} and the definition of $\mathring{\bar{\delta}}_\bbP$ imply
\begin{equation}\label{eq:cp^-surj-4}
\hat{n}_\bbP(w\hat{a}) = \hat{n}_\bbP (\hat{\Pi}_{N \setminus \{s\}}^N(w) + 
\mathring{\bar{\delta}}_\bbP(\delta(q_0,\wedge(\hat{\tau}_s^{N}(w))),\hat{b}).
\end{equation}
Let $\hat{b} \in \tilde{\Sigma} \dotcup \tilde{\bar{\Sigma}}$, then $(f,\hat{b},g) = 
(f,\hat{b},f+\tilde{\bar{\delta}}_\bbP(\hat{b})) 
\in \tilde{\shuffle}_\bbP \dotcup \tilde{\bar{\shuffle}}_\bbP$ implies 
$\hat{b} \in \pre(\langle P \rangle)$, and therefore
$w \in  \bigcap\limits_{t \in N} (\hat{\tau}_t^{N})^{-1}(\pre(\langle P \rangle))$ implies
\begin{equation}\label{eq:cp^-surj-5}
w\hat{a} \in  \bigcap\limits_{t \in N} (\hat{\tau}_t^{N})^{-1}(\pre(\langle P \rangle)) 
\mbox{ for each } \hat{a} \in \hat{\Sigma}_{N \setminus \kappa(w)}
\mbox{ with } \hat{\Theta}^N(\hat{a}) = \hat{b}.
\end{equation}
Now \eqref{eq:cp^-surj-0}, \eqref{eq:cp^-surj-2}, \eqref{eq:cp^-surj-3} and \eqref{eq:cp^-surj-5}
prove the first part of Lemma~\ref{lemma:cp^-surj}.
\ \\
\ \\
Let $\hat{b} \in \mathring{\Sigma} \dotcup \bar{\Sigma}$, and let
$(f,\hat{b},g) \in \mathring{\shuffle}_\bbP \dotcup \bar{\shuffle}_\bbP$ 
be represented by 
$(f,\hat{b},g) = (f'+1_q,\hat{b},f' + \mathring{\bar{\delta}}_\bbP(q,\hat{b}))$ 
with $q \in Q$.
On account of $f'+1_q = \hat{n}_\bbP(w)$, there exists $s \in \kappa(w)$ such that
$f' = \hat{n}_\bbP (\hat{\Pi}_{N \setminus \{s\}}^N(w))$, 
$\hat{\tau}_s^{N}(w) \notin \langle P \rangle \cup \{\varepsilon\}$, and 
$\delta(q_0,\wedge(\hat{\tau}_s^{N}(w)))=q$. 
Therefore by $(f'+1_q,\hat{b},f' + \mathring{\bar{\delta}}_\bbP(q,\hat{b})) 
\in \mathring{\shuffle}_\bbP \dotcup \bar{\shuffle}_\bbP$ holds 
$\hat{\tau}_s^{N}(w) \hat{b} \in \pre(\langle P \rangle)$, which implies
\begin{equation}\label{eq:cp^-surj-6}
w\hat{a} \in  \bigcap\limits_{t \in N} (\hat{\tau}_t^{N})^{-1}(\pre(\langle P \rangle)) 
\mbox{ for } \hat{a} \in \hat{\Sigma}_{\{s\}}
\mbox{ with } \hat{\Theta}^N(\hat{a}) = \hat{b}.
\end{equation}
Now \eqref{eq:cp^-surj-1}, \eqref{eq:cp^-surj-2}, \eqref{eq:cp^-surj-4} and \eqref{eq:cp^-surj-6}
prove the second part of Lemma~\ref{lemma:cp^-surj}.
\end{proof}
\ \\
Generally, for $L \subset \Sigma^*$ and $x \in \Sigma^*$ the \emph{left quotient} $x^{-1}(L)$ 
is defined by 
\begin{equation}\label{eq:left quotient}
x^{-1}(L) := \{y \in \Sigma^* | xy \in L \}.
\end{equation}
By induction on the length of $c \in \hat{A}_\bbP$  Lemma~\ref{lemma:cp^-surj} implies
\begin{theorem}\label{thm:cp^-surj}
\begin{subequations}\label{subeq:cp^-surj}
\begin{align}
& \hat{c}_\bbP [\bigcap\limits_{t \in N} (\hat{\tau}_t^{N})^{-1}(\pre(\langle P \rangle))] = 
       \hat{A}_\bbP . \label{subeq:cp^-surj-1}\ \ \  \mbox{  Moreover } \\
& \hat{c}_\bbP [x(x^{-1}[\bigcap\limits_{t \in N} (\hat{\tau}_t^{N})^{-1}(\pre(\langle P \rangle))])] = 
        \hat{c}_\bbP (x)[(\hat{c}_\bbP (x))^{-1}(\hat{A}_\bbP)] \label{subeq:cp^-surj-2} 
\end{align}
\end{subequations}
for each $x \in  \bigcap\limits_{t \in N} (\hat{\tau}_t^{N})^{-1}(\pre(\langle P \rangle))$.
\end{theorem}
On account of \eqref{subeq:def-cp^-2} now from \eqref{subeq:cp^-surj-1} it follows
\begin{corollary}\label{cor:cp^-surj}\ \\
$(\pre(\langle P \rangle))^\shuffle = \hat{\alpha}_\bbP(\hat{A}_\bbP)$, 
which states that the semiautomaton 
$\hat{\bbP}_\shuffle$ recognizes the prefix closed language $(\pre(\langle P \rangle))^\shuffle$.
\end{corollary}
Because of
\begin{align*}
&  \bigcap\limits_{t \in N} (\hat{\tau}_t^N)^{-1}(\langle P \rangle \cup \{\varepsilon\}) =\\ 
&  \{x\in \bigcap\limits_{t \in N} (\hat{\tau}_t^{N})^{-1}(\pre(\langle P \rangle)) \ |
   \ \hat{\tau}_t^{N}(x) \in \langle P \rangle \cup \{\varepsilon\} \mbox{ for each } t \in N\} =\\ 
&  \{x\in \bigcap\limits_{t \in N} (\hat{\tau}_t^{N})^{-1}(\pre(\langle P \rangle)) \ |
   \ \hat{n}_\bbP(x) = 0\},
\end{align*}
it holds
\begin{equation}\label{eq:np^0}
\bigcap\limits_{t \in N} (\hat{\tau}_t^N)^{-1}(\langle P \rangle \cup \{\varepsilon\}) = 
\hat{n}_\bbP^{-1}(0).
\end{equation}
Therefore \eqref{subeq:def-cp^-1} and \eqref{subeq:cp^-surj-1} imply
\begin{equation}\label{eq:cp^-surj+}
\hat{c}_\bbP [\bigcap\limits_{t \in N} (\hat{\tau}_t^N)^{-1}(\langle P \rangle \cup \{\varepsilon\})] = 
\hat{Z}_\bbP^{-1}(0).
\end{equation}
Now, from \eqref{eq:cp^-surj+} and \eqref{subeq:def-cp^-2} it follows
\begin{corollary}\label{cor:cp^-surj+}\ \\
$\langle P \rangle^\shuffle = \hat{\alpha}_\bbP(\hat{Z}_\bbP^{-1}(0))$, 
which states that the semiautomaton 
$\hat{\bbP}_\shuffle$ enriched by the final state $0 \in \dsN_0^Q$ recognizes $\langle P \rangle^\shuffle$.
\end{corollary}
\ \\
Let $\bbA$ be an automaton recognizing $L \subset \Phi^*$ and let $\varphi : \Phi^* \to \Gamma^*$  
be a \emph{strictly} alphabetic homomorphism, where strictly is defined by 
$|\varphi(w)| = |w|$ for each $w \in \Phi^*$.
Then it is easy and well known to construct an automaton $\bbA'$ recognizing $\varphi(L) \subset \Gamma^*$. 
Now this construction will be realized for the semiautomaton $\hat{\bbP}_\shuffle$ and the 
strictly alphabetic homomorphism $\wedge : \hat{\Sigma}^* \to \Sigma^*$. 
Additionally this construction will be extended to a modification of the function $\hat{c}_\bbP$.
\begin{definition} [$S$-automaton $\bbP_\shuffle$]\label{def:shuffle-automaton}
\ \\
$\bbP_\shuffle = (\Sigma,\dsN_0^Q,\shuffle_\bbP,0)$ w.r.t. $\bbP$ is a semiautomaton with
an infinite state set $\dsN_0^Q$, the initial state $0$ and a state transition relation 
$ \shuffle_\bbP \subset \dsN_0^Q \times \Sigma \times \dsN_0^Q $ defined by
$\shuffle_\bbP := \{(f,\wedge(\hat{a}),g) \in \dsN_0^Q \times \Sigma \times \dsN_0^Q \ | 
\ (f,\hat{a},g) \in \hat{\shuffle}_\bbP \}$.
\end{definition}
Adopting the notions of $\hat{S}$-automata, 
$A_\bbP \subset \shuffle_\bbP^*$ denotes the set of all paths in $\bbP_\shuffle$ 
starting with the initial state $0$ and including the empty path $\varepsilon$. 
For $w\in A_\bbP$, $Z_\bbP(w)$ denotes the final state of the path $w$ and 
$Z_\bbP(\varepsilon):=0$. 
Formally the prefix closed language $A_\bbP$ and the function
$Z_\bbP : A_\bbP \to \dsN_0^Q$ are defined inductively by 
\begin{equation} \label{eq:def A,Z}
\varepsilon \in A_\bbP, \ Z_\bbP(\varepsilon):=0, 
\ w(f,a,g)\in A_\bbP \mbox{ and } 
Z_\bbP(w(f,a,g)):=g 
\end{equation}
for $w\in A_\bbP$, $Z_\bbP(w)=f$ and $(f,a,g)\in \shuffle_\bbP$.
\ \\
\ \\
Let the function $\alpha_\bbP : A_\bbP \to \Sigma^*$ be inductively defined by
\begin{equation} \label{eq:def alpha}
\alpha_\bbP(\varepsilon):=\varepsilon \mbox{ and }
\alpha_\bbP(w(f,a,g)):=\alpha_\bbP(w)a
\end{equation}
for $w(f,a,g)\in A_\bbP$ and $(f,a,g)\in \shuffle_\bbP$.
$\alpha_\bbP(u)$ is called the \emph{label of a path u}.
\ \\
\ \\
To formally capture the relation between $\hat{\bbP}_\shuffle$ and $\bbP_\shuffle$, 
we consider the homomorphism
\begin{equation} \label{eq:def-wedge-P}
\wedge_\bbP : \hat{\shuffle}_\bbP^* \to \shuffle_\bbP^* 
\mbox{ with } \wedge_\bbP((f,\hat{a},g)) := (f,\wedge(\hat{a}),g) 
\mbox{ for } (f,\hat{a},g) \in \hat{\shuffle}_\bbP.
\end{equation}
This definition implies  
\begin{subequations}\label{subeq:def-wedge-P}
\begin{align}
& \wedge_\bbP \mbox{ is strictly alphabetic and surjective.} \label{subeq:def-wedge-P-1} \\
& \wedge_\bbP(y) \in A_\bbP \mbox{ iff } y \in \hat{A}_\bbP \mbox{ for } y \in \hat{\shuffle}_\bbP^*.  \label{subeq:def-wedge-P-2} \\
& \hat{Z}_\bbP(x) = Z_\bbP(\wedge_\bbP(x)) \mbox{ for } x \in \hat{A}_\bbP. \label{subeq:def-wedge-P-3} \\
& \wedge(\hat{\alpha}_\bbP(x)) = \alpha_\bbP(\wedge_\bbP(x)) \mbox{ for } x \in \hat{A}_\bbP. \label{subeq:def-wedge-P-4}
\end{align}
\end{subequations}
Now the composition of $\hat{c}_\bbP$ with $\wedge_\bbP$ attunes $\hat{c}_\bbP$ to $\bbP_\shuffle$.  
\begin{definition}\label{def:def-cp}\ \\
Let the function
$c_\bbP : \bigcap\limits_{t \in N} (\hat{\tau}_t^N)^{-1}(\pre(\langle P \rangle)) \to A_\bbP$ 
be defined by $c_\bbP := \wedge_\bbP \circ \hat{c}_\bbP$.
\end{definition}
By Corollary~\ref{cor:SP(hom-shuffle)} and \eqref{subeq:def-wedge-P-1} - \eqref{subeq:def-wedge-P-4}, 
Corollary~\ref{cor:cp^-surj} and Corollary~\ref{cor:cp^-surj+} imply the following automata representations:
\begin{corollary}\label{cor:cp-surj}\ \ 
$(\pre(P))^\shuffle = \alpha_\bbP(A_\bbP)$ and $P^\shuffle = \alpha_\bbP(Z_\bbP^{-1}(0))$. 
\end{corollary}
For use in the next section the following theorem assembles the properties of the function $c_\bbP$, which follow 
from \eqref{subeq:def-wedge-P-1} - \eqref{subeq:def-wedge-P-4}, 
Theorem~\ref{thm:def-cp^} and Theorem~\ref{thm:cp^-surj}:
\begin{theorem}\label{thm:def-cp}
Let $x \in  \bigcap\limits_{t \in N} (\hat{\tau}_t^{N})^{-1}(\pre(\langle P \rangle))$ then
\begin{subequations}\label{subeq:def-cp}
\begin{align}
&  Z_\bbP(c_\bbP(x)) = \hat{n}_\bbP(x) , \label{subeq:def-cp-1}\\
&  \alpha_\bbP(c_\bbP(x)) = \wedge(\hat{\Theta}^N(x)) , \label{subeq:def-cp-2}\\ 
&  |c_\bbP(x)| = |x| , \label{subeq:def-cp-3}\\ 
&  \pre(c_\bbP(x)) = c_\bbP(\pre(x)) , \mbox{ and} \label{subeq:def-cp-4}\\
&  c_\bbP [x(x^{-1}[\bigcap\limits_{t \in N} (\hat{\tau}_t^{N})^{-1}(\pre(\langle P \rangle))])] = 
        c_\bbP(x)[(c_\bbP (x))^{-1}(A_\bbP)], \label{subeq:def-cp-5}\\
&  \mbox{which implies} \nonumber \\
&  c_\bbP [\bigcap\limits_{t \in N} (\hat{\tau}_t^{N})^{-1}(\pre(\langle P \rangle))] = 
       A_\bbP . \label{subeq:def-cp-6}
\end{align}
\end{subequations}
\end{theorem}
\section{Shuffle Projection in Terms of $S$-Automata}\label{sec:automata-shuffle projection}
To express shuffle projection in terms of $S$-automata we first consider shuffle projection w.r.t. 
prefix closed languages. 
Let therefore $P, V \subset \Sigma^*$, $P \neq \emptyset$ and let $\bbP$ be an automaton for $P$ 
as in Section~\ref{sec:automata}. 
By Corollary~\ref{cor:sp-4'} together with Theorem~\ref{thm:mod-shuffle-projection} holds 
$\SP(\pre(P),V)$ iff there exists an infinite countable set $N$ such that
\begin{equation}\label{eq:sat-0}
\hat{\Pi}_{N \setminus \{r\}}^N [(\bigcap\limits_{t\in N}(\hat{\tau}_t^{N})^{-1}(\pre(\langle P \rangle))) 
\cap (\hat{\Theta}^N)^{-1}(\wedge^{-1}(V))] \subset (\hat{\Theta}^N)^{-1}(\wedge^{-1}(V))
\end{equation}
for each $r \in N$.
\ \\
The same argument as to prove \eqref{eq:mod-shuffle-projection-130} shows that \eqref{eq:sat-0} 
is equivalent to
\begin{align}\label{eq:sat-0'}
&\hat{\Pi}_{N \setminus \{r\}}^N [(\bigcap\limits_{t\in N}(\hat{\tau}_t^{N})^{-1}(\pre(\langle P \rangle))) 
\cap (\hat{\Theta}^N)^{-1}(\wedge^{-1}(V))] \subset \nonumber \\ 
&(\bigcap\limits_{t\in N}(\hat{\tau}_t^{N})^{-1}(\pre(\langle P \rangle))) 
\cap (\hat{\Theta}^N)^{-1}(\wedge^{-1}(V))
\end{align}
for each $r \in N$.
\ \\
Condition \eqref{eq:sat-0'} is a saturation property of  
$(\wedge\circ\hat{\Theta}_{|\bigcap\limits_{t \in N}(\hat{\tau}_t^{N})^{-1}(\pre(\langle P \rangle))}^N)^{-1}(V)$ 
wrt. a binary relation 
on $\bigcap\limits_{t \in N} (\hat{\tau}_t^{N})^{-1}(\pre(\langle P \rangle))$ defined by the 
homomorphisms $\hat{\Pi}_{N \setminus \{r\}}^N$ for $r \in N$. More precisely:
\ \\
\ \\
Let $R \subset F \times F$ be a binary relation on a set $F$ and let $W \subset F$. The 
\emph{saturation property} $\S(W,R)$ let be defined by 
\begin{equation} \label{eq:def-sat}
\S(W,R) \mbox{ iff } x \in W \mbox{ and } (x,y) \in R \mbox{ imply } y \in W.
\end{equation} 
Let $f : F \to G$, $g : G \to H$ and $V \subset H$, then \eqref{eq:def-sat} 
immediately implies
\begin{equation} \label{eq:sat-1}
\S((g \circ f)^{-1}(V),R) \mbox{ iff } \S((g^{-1}(V),(f \otimes f)(R)),
\end{equation}
where $f \otimes f : F \times F \to G \times G $ is defined by
$(f \otimes f)((x,y)) := (f(x),f(y))$ for $(x,y) \in F \times F$.
\begin{definition}\label{def:def-RP}
\begin{align}
&\mbox{Let }\mcR_P := \{(x,y)\in \bigcap\limits_{t \in N} (\hat{\tau}_t^{N})^{-1}(\pre(\langle P \rangle)) 
\times \bigcap\limits_{t \in N} (\hat{\tau}_t^{N})^{-1}(\pre(\langle P \rangle)) | \nonumber \\
&\ \ \ \ \ \ \ \ \ \ \ \ \ \ \ \ \ \ \ \ \ \ \ \ \ \ \ \ \mbox{there exists } 
r \in N \mbox{ with } y = \hat{\Pi}_{N \setminus \{r\}}^N (x)\}. \nonumber
\end{align}
\end{definition}
Now by \eqref{eq:sat-0} and \eqref{eq:sat-0'}
\begin{equation} \label{eq:sat-2}
\SP(\pre(P),V) \mbox{ iff } \S((\wedge\circ\hat{\Theta}_{|
\bigcap\limits_{t \in N}(\hat{\tau}_t^{N})^{-1}(\pre(\langle P \rangle))}^N)^{-1}(V),\mcR_P).
\end{equation}
On account of \eqref{subeq:def-cp-2} holds $\wedge\circ\hat{\Theta}_{|
\bigcap\limits_{t \in N}(\hat{\tau}_t^{N})^{-1}(\pre(\langle P \rangle))}^N = 
\alpha_\bbP \circ c_\bbP$. 
Therefore \eqref{eq:sat-1} and \eqref{eq:sat-2} imply
\begin{equation} \label{eq:sat-3}
\SP(\pre(P),V) \mbox{ iff } \S(\alpha_\bbP^{-1}(V),(c_\bbP \otimes c_\bbP)(\mcR_P)).
\end{equation}
\ \\
In Section~\ref{sec:automata} the idea to define $\bbP_\shuffle$ was the following:
Each computation in $\bbP_\shuffle$ ``correspond'' to a ``shuffled run'' of ``elementary computations''. 
Now we will show that $(u,v) \in (c_\bbP \otimes c_\bbP)(\mcR_P) \subset A_\bbP \times A_\bbP$ 
iff the ``shuffled run'' $v'$ of ``elementary computations'' is generated from the ``shuffled run'' 
$u'$ of ``elementary computations'' by ``deleting'' one of the ``elementary computations'' in $u'$, where 
$u$ ``correspond'' to $u'$ and $v$ ``correspond'' to $v'$. The formalization of this idea will result 
in a characterization of $(c_\bbP \otimes c_\bbP)(\mcR_P) \subset A_\bbP \times A_\bbP$ without explicit use 
of $\mcR_P$.
\ \\
\ \\
First we have to formalize ``elementary computations'': For each $r \in N$ holds 
$(\hat{\tau}_r^{\{r\}})^{-1}(\pre(\langle P \rangle)) \subset 
\bigcap\limits_{t \in N} (\hat{\tau}_t^{N})^{-1}(\pre(\langle P \rangle))$ 
and 
$c_\bbP((\hat{\tau}_r^{\{r\}})^{-1}(\pre(\langle P \rangle))) = 
c_\bbP((\hat{\tau}_s^{\{s\}})^{-1}(\pre(\langle P \rangle)))$ 
for each $s \in N$.
Therefore the following definition does not depend on $r \in N$.
\begin{definition}\label{def:def-E_P}
\ \\
Let $r \in N$. The prefix closed set 
$E_\bbP := c_\bbP((\hat{\tau}_r^{\{r\}})^{-1}(\pre(\langle P \rangle))) \subset A_\bbP$ 
is called the set of \emph{elementary computations} in $\bbP_\shuffle$.
\end{definition}
%
\begin{example}\label{ex:def-E_P}
\begin{figure}[h]
\centering
\begin{tikzpicture}[->,>=stealth',shorten >=1pt,auto,node distance=2cm,semithick,initial text=]
   \node[state,initial,inner sep=1pt,minimum size=6mm] (k1)  {\small I};
   \node[state,inner sep=1pt,minimum size=6mm] (k2) at (1.5,0) {\small II};
   \node[state,accepting,inner sep=1pt,minimum size=6mm] (k3) at (3,0) {\small III};
   \path 
         (k1) edge node {\small $a$} (k2)
         (k2) edge node {\small $b$} (k3)
		;
\end{tikzpicture}
\caption{Automaton $\bbP$ recognizing $P = \{ab\}$}\label{fig:def-E_P}
\end{figure}
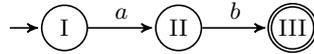
\ \\
Let $P$ and $\bbP$ be defined as in Fig.~\ref{fig:def-E_P}, then $E_\bbP = \pre(\{(0,a,1_{\small\II})(1_{\small\II},b,0)\})$.
\end{example}
$E_\bbP$ can also be characterized without referring to $c_\bbP$:
\begin{align} \label{eq:char-E_P}
& E_\bbP = \pre(\{c\in Z_\bbP^{-1}(0)\cap\alpha_\bbP^{-1}(P) | Z_\bbP(c')=1_{\delta(q_0,\alpha_\bbP(c'))} 
\mbox{ for each } \nonumber\\ 
& c' \in \pre(c) \mbox{ with } 0 < |c'| < |c|\}), \mbox{ which implies } \nonumber\\
& \alpha_\bbP(E_\bbP) = \pre(P).
\end{align}
To formally define shuffled runs and corresponding representations, 
let $\check{\Sigma}$ be a disjoint copy of $\Sigma$ 
and $\check{\iota} : \check{\Sigma}^* \to \Sigma^*$ the corresponding isomorphism. 
This isomorphism defines a deterministic automaton $\check{\bbP}$ isomorphic to $\bbP$ 
with the same state set as $\bbP$ and recognizing $\check{\iota}^{-1}(P)$. More precisely: 
Let $\check{\bbP} :=(\check{\Sigma},Q,\check{\delta},q_0,F)$, where 
$\bbP=(\Sigma,Q,\delta,q_0,F)$ and 
$\check{\delta}(p,\check{a}) := \delta(p,\check{\iota}(\check{a}))$ 
for $\check{a} \in \check{\Sigma}$ and $p \in Q$. This definition implies
\begin{equation} \label{eq:iota-shuff-0}
(f,\check{a},g) \in \shuffle_{\check{\bbP}} \mbox{ iff } 
(f,\check{\iota}(\check{a}),g) \in \shuffle_{\bbP} 
\mbox{ for } f,g \in \dsN_0^Q \mbox{ and } \check{a} \in \check{\Sigma}.
\end{equation}
Therefore
\begin{equation} \label{eq:iota-shuff-1}
\check{\iota}_{\shuffle_{\check{\bbP}}}((f,\check{a},g)) := (f,\check{\iota}(\check{a}),g)
\mbox{ for } (f,\check{a},g) \in \shuffle_{\check{\bbP}}
\end{equation}
defines an isomorphism
$\check{\iota}_{\shuffle_{\check{\bbP}}} : \shuffle_{\check{\bbP}}^* \to \shuffle_{\bbP}^*$ with
\begin{subequations}\label{subeq:iota-shuff-2}
\begin{align}
&  \check{\iota}_{\shuffle_{\check{\bbP}}}(A_{\check{\bbP}}) = A_{\bbP} , \label{subeq:iota-shuff-2-a}\\
&  \check{\iota}_{\shuffle_{\check{\bbP}}}(E_{\check{\bbP}}) = E_{\bbP} , \label{subeq:iota-shuff-2-b}\\ 
&  Z_{\check{\bbP}} = Z_{\bbP} \circ 
   \check{\iota}_{\shuffle_{\check{\bbP}} | A_{\check{\bbP}}} , \mbox{ and} \label{subeq:iota-shuff-2-c}\\
&  \check{\iota} \circ \alpha_{\check{\bbP}} = \alpha_{\bbP} \circ 
   \check{\iota}_{\shuffle_{\check{\bbP}} | A_{\check{\bbP}}}. \label{subeq:iota-shuff-2-d}
\end{align}
\end{subequations}
Because of $\check{\Sigma} \cap \Sigma = \emptyset$, it also holds 
$\shuffle_{\check{\bbP}} \cap \shuffle_{\bbP} = \emptyset$. 
\[\mbox{Let therefore } \pi_{\shuffle_{\bbP}}: 
(\shuffle_{\bbP} \dotcup \shuffle_{\check{\bbP}})^* \to \shuffle_{\bbP}^* \mbox{ be defined by }\] 
\[\pi_{\shuffle_{\bbP}}(y):= y \mbox{ for } y \in {\shuffle_{\bbP}} \mbox{ and }
\pi_{\shuffle_{\bbP}}(y):= \varepsilon \mbox{ for } y \in \shuffle_{\check{\bbP}}.\]
\[\mbox{In the same way let } \pi_{\shuffle_{\check{\bbP}}}: 
(\shuffle_{\bbP} \dotcup \shuffle_{\check{\bbP}})^* \to \shuffle_{\check{\bbP}}^* \mbox{ be defined by }\] 
\[\pi_{\shuffle_{\check{\bbP}}}(y):= \varepsilon \mbox{ for } y \in {\shuffle_{\bbP}} \mbox{ and }
\pi_{\shuffle_{\check{\bbP}}}(y):= y \mbox{ for } y \in \shuffle_{\check{\bbP}}.\]
As $A_{\bbP} \subset \shuffle_{\bbP}^*$ and $E_{\check{\bbP}} \subset \shuffle_{\check{\bbP}}^*$
are prefix closed languages, 
$\pi_{\shuffle_{\bbP}}^{-1}(A_{\bbP}) \cap \pi_{\shuffle_{\check{\bbP}}}^{-1}(E_{\check{\bbP}}) 
\subset (\shuffle_{\bbP} \dotcup \shuffle_{\check{\bbP}})^*$ is also a prefix closed language. 
Its elements are called \emph{shuffled runs} of a computation in $\bbP$ and an elementary computation in $\check{\bbP}$.
Let now $\beta_{\bbP}: 
(\shuffle_{\bbP} \dotcup \shuffle_{\check{\bbP}})^* \to {\Sigma}^*$ be defined by
\begin{align}\label{eq:beta_P} 
& \beta_\bbP((f,x,g)):= x \mbox{ for } (f,x,g)\in \shuffle_{\bbP} \mbox{ and } \nonumber\\
& \beta_\bbP((f,x,g)):= \check{\iota}(x) \mbox{ for } (f,x,g)\in \shuffle_{\check{\bbP}}.
\end{align}
A shuffled run $b \in \pi_{\shuffle_{\bbP}}^{-1}(A_{\bbP}) \cap \pi_{\shuffle_{\check{\bbP}}}^{-1}(E_{\check{\bbP}})$ 
is called a \emph{shuffled representation} of $c\in A_\bbP$ by $d\in A_\bbP$ and $e\in E_{\check{\bbP}}$ iff
\begin{subequations}\label{subeq:shuff-rep}
\begin{align}
&  \alpha_\bbP(c)=\beta_\bbP(b) , \label{subeq:shuff-rep-a}\\
&  \pi_{\shuffle_{\check{\bbP}}}(b)=e , \label{subeq:shuff-rep-b}\\ 
&  \pi_{\shuffle_{\bbP}}(b)= d , \mbox{ and} \label{subeq:shuff-rep-c}\\
&  Z_\bbP(c')= Z_{\check{\bbP}}(\pi_{\shuffle_{\check{\bbP}}}(b'))+ 
   Z_{{\bbP}}(\pi_{\shuffle_{\bbP}}(b')) \nonumber\\
&  \mbox{for each } c'\in \pre(c), \mbox{ where } b'\in \pre(b) \mbox{ with } |b'|=|c'| . \label{subeq:shuff-rep-d}
\end{align}
\end{subequations}
\begin{example}\label{ex:shuff-rep}
\ \\
Let $P$ and $\bbP$ be defined as in Fig.~\ref{fig:def-E_P}, and 
\[d = (0,a,1_{\small\II})(1_{\small\II},b,0)(0,a,1_{\small\II})(1_{\small\II},b,0) \in A_{\bbP}, \]
\[e = (0,\check{a},1_{\small\II})(1_{\small\II},\check{b},0) \in E_{\check{\bbP}}, \]
\[b = (0,a,1_{\small\II})(0,\check{a},1_{\small\II})(1_{\small\II},b,0)
(0,a,1_{\small\II})(1_{\small\II},\check{b},0)(1_{\small\II},b,0) \in 
\pi_{\shuffle_{\bbP}}^{-1}(A_{\bbP}) \cap \pi_{\shuffle_{\check{\bbP}}}^{-1}(E_{\check{\bbP}}) \]
and
\[c = (0,a,1_{\small\II})(1_{\small\II},a,2_{\small\II})(2_{\small\II},b,1_{\small\II})
(1_{\small\II},a,2_{\small\II})(2_{\small\II},b,1_{\small\II})(1_{\small\II},b,0) \in A_{\bbP}, \]
then $b$ is a shuffled representation of $c$ by $d$ and $e$.
\end{example}
The shuffled representations define a relation
$\mcR_\bbP \subset A_\bbP \times A_\bbP$:

\begin{definition}\label{def:rel-RP}
\ \\
$\mcR_\bbP := \{(c,d)\in A_\bbP \times A_\bbP | \mbox{ there exists } e\in E_{\check{\bbP}} 
\mbox{ and a shuffled representation } b  \in 
\pi_{\shuffle_{\bbP}}^{-1}(A_{\bbP}) \cap \pi_{\shuffle_{\check{\bbP}}}^{-1}(E_{\check{\bbP}}) 
\mbox{ of } c \mbox{ by } d \mbox{ and } e\}$.
\end{definition}
Now we will prove $\mcR_\bbP = (c_\bbP \otimes c_\bbP)(\mcR_P)$. 
For this purpose we define an appropriate function 
$b_\bbP : N \times \bigcap\limits_{t \in N} (\hat{\tau}_t^N)^{-1}(\pre(\langle P \rangle)) \to 
\pi_{\shuffle_{\bbP}}^{-1}(A_{\bbP}) \cap \pi_{\shuffle_{\check{\bbP}}}^{-1}(E_{\check{\bbP}})$. 
For it we first need a unique factorization property of the elements of 
$\bigcap\limits_{t \in N} (\hat{\tau}_t^N)^{-1}(\pre(\langle P \rangle))$:\\ 
\ \\
\ \\
Let $w \in \bigcap\limits_{t \in N} (\hat{\tau}_t^N)^{-1}(\pre(\langle P \rangle))$, 
$r \in N$, $x = \hat{\Pi}_{\{r\}}^N(w)$ 
and $y = \hat{\Pi}_{N \setminus \{r\}}^N(w)$. 
Then there exists exactly one $y_0 \in \hat{\Sigma}_{N \setminus \{r\}}^*$, 
and for each $i \in \{i \in \dsN | 1 \leq i \leq |x| \}$ exactly one 
$x_i \in \hat{\Sigma}_{\{r\}}$ as well as exactly one
$y_i \in \hat{\Sigma}_{N \setminus \{r\}}^*$ such that
\begin{align}\label{eq:unique-fac}
&  w = y = y_0 \mbox{ for } x = \varepsilon \mbox{, \ \ \  and }  \nonumber\\
&  w = y_0x_1y_1...x_{|x|}y_{|x|},\ x = x_1...x_{|x|} 
\mbox{ as well as } y = y_0y_1...y_{|x|} \mbox{ for } x \neq \varepsilon. 
\end{align}
Because of $|c_\bbP(x)| = |x|$, $|c_\bbP(y)| = |y|$ and
$\hat{\Pi}_{\{r\}}^N(w) = (\hat{\tau}_r^{\{r\}})^{-1}(\hat{\tau}_r^N(w))$, 
which implies $c_\bbP(x) \in E_\bbP$, the following definition is sound:
\begin{definition}\label{def:b_P}
\ \\
Let $r,\ w,\  x,\ y$, and the factorizations of $w$, $x$ and $y$ as in \ref{eq:unique-fac}, then
\[b_\bbP : N \times \bigcap\limits_{t \in N} (\hat{\tau}_t^N)^{-1}(\pre(\langle P \rangle)) \to 
\pi_{\shuffle_{\bbP}}^{-1}(A_{\bbP}) \cap \pi_{\shuffle_{\check{\bbP}}}^{-1}(E_{\check{\bbP}}) 
\mbox{ is defined by }\] 
\[b_\bbP((r,w)) := c_\bbP(y) \mbox{ for } x = \varepsilon \mbox{ and } 
b_\bbP((r,w)) := v_0u_1v_1...u_{|x|}v_{|x|} \mbox{ for } x \neq \varepsilon, \] 
where $u_1...u_{|x|} = \check{\iota}_{\shuffle_{\check{\bbP}}}^{-1}(c_\bbP(x))$, 
$v_0v_1...v_{|x|} = c_\bbP(y)$, $|u_i| = |x_i|$ and 
$|v_k| = |y_k|$ for $1 \leq i \leq |x|$ and $0 \leq k \leq |x|$.
\end{definition}
By this definition $c_\bbP(y)$ and $\check{\iota}_{\shuffle_{\check{\bbP}}}^{-1}(c_\bbP(x))$ 
are shuffled in $b_\bbP((r,w))$ in the same manner as $y$ and $x$ are shuffled in $w$, 
which implies
\begin{equation}\label{eq:b_P-1}
|b_\bbP((r,w))| = |w|,
\end{equation}
and moreover
\begin{align}\label{eq:unique-fac-1}
&  b_\bbP((r,w)) \in \pi_{\shuffle_{\bbP}}^{-1}(c_\bbP(\hat{\Pi}_{N \setminus \{r\}}^N(w))) \cap 
   \pi_{\shuffle_{\check{\bbP}}}^{-1}(\check{\iota}_{\shuffle_{\check{\bbP}}}^{-1}
   (c_\bbP(\hat{\Pi}_{\{r\}}^N(w)))), \nonumber\\
&  |\hat{\Pi}_{N \setminus \{r\}}^N(w')| = |\pi_{\shuffle_{\bbP}}(b')| \mbox{ and }
   |\hat{\Pi}_{\{r\}}^N(w')| = |\pi_{\shuffle_{\check{\bbP}}}(b')| \nonumber\\
&  \mbox{for each } w'\in \pre(w) \mbox{ and } b'\in \pre(b_\bbP((r,w))) \mbox{ with } |w'|=|b'| . 
\end{align}
It is easy to see that  \ref{eq:unique-fac-1} characterizes $b_\bbP((r,w))$. More precisely:
\begin{align}\label{eq:characterize-b_P}
\{b_\bbP((r,w))\} = \{b \in &\pi_{\shuffle_{\bbP}}^{-1}(c_\bbP(\hat{\Pi}_{N \setminus \{r\}}^N(w))) \cap 
   \pi_{\shuffle_{\check{\bbP}}}^{-1}(\check{\iota}_{\shuffle_{\check{\bbP}}}^{-1}
   (c_\bbP(\hat{\Pi}_{\{r\}}^N(w)))) \ | \nonumber\\
& |\hat{\Pi}_{N \setminus \{r\}}^N(w')| = |\pi_{\shuffle_{\bbP}}(b')| \mbox{ and }
   |\hat{\Pi}_{\{r\}}^N(w')| = |\pi_{\shuffle_{\check{\bbP}}}(b')| \nonumber\\
&  \mbox{for each } w'\in \pre(w) \mbox{ and } b'\in \pre(b) \mbox{ with } |w'|=|b'| \} . 
\end{align}
Now \eqref{eq:characterize-b_P} and Theorem~\ref{thm:def-cp} 
together with \eqref{eq:def-cp^2},
\eqref{subeq:iota-shuff-2-c}, \eqref{subeq:iota-shuff-2-d}
and \eqref{eq:beta_P} imply
\begin{align}
&  \pre(b_\bbP((r,w))) = b_\bbP((r,\pre(w))) , \label{eq:b_P-2}\\ 
&  \wedge(\hat{\Theta}^N(w)) = \beta_\bbP(b_\bbP((r,w))) \mbox{ and} \label{eq:b_P-3}\\
&  \hat{n}_\bbP(w) = Z_{\check{\bbP}}(\pi_{\shuffle_{\check{\bbP}}}(b_\bbP((r,w))))+ 
   Z_{{\bbP}}(\pi_{\shuffle_{\bbP}}(b_\bbP((r,w)))) \label{eq:b_P-4}
\end{align}
To complete the list of properties of $b_\bbP$ we will show
\begin{equation}\label{eq:b_P-5}
b_\bbP(N \times \bigcap\limits_{t \in N} (\hat{\tau}_t^N)^{-1}(\pre(\langle P \rangle))) = 
\pi_{\shuffle_{\bbP}}^{-1}(A_{\bbP}) \cap \pi_{\shuffle_{\check{\bbP}}}^{-1}(E_{\check{\bbP}}).
\end{equation}
\begin{proof} Proof of equation~\eqref{eq:b_P-5}:\\
\ \\
Let $b \in \pi_{\shuffle_{\bbP}}^{-1}(A_{\bbP}) \cap \pi_{\shuffle_{\check{\bbP}}}^{-1}(E_{\check{\bbP}}).$
Because of \eqref{subeq:def-cp-6}, Definition~\ref{def:def-E_P} and \eqref{subeq:iota-shuff-2-b} there exist 
$y \in \bigcap\limits_{t \in N} (\hat{\tau}_t^N)^{-1}(\pre(\langle P \rangle))$ and 
$\hat{x} \in \pre(\langle P \rangle)$ such that $c_\bbP(y) = \pi_{\shuffle_{\bbP}}(b)$ and 
$\check{\iota}_{\shuffle_{\check{\bbP}}}^{-1}(c_\bbP((\hat{\tau}_s^{\{s\}})^{-1}(\hat{x}))) = 
\pi_{\shuffle_{\check{\bbP}}}(b)$ for each $s \in N$.\\
\ \\
Let now $r \in N \setminus \{ \kappa(y) \}$, then by the same argument as in 
\eqref{eq:unique-fac} and in \eqref{eq:unique-fac-1} $y$ and $(\hat{\tau}_r^{\{r\}})^{-1}(\hat{x})$ 
can be shuffled in the same manner as $\pi_{\shuffle_{\bbP}}(b)$ and $\pi_{\shuffle_{\check{\bbP}}}(b)$ 
are shuffled in $b$. 
This result in $w \in \bigcap\limits_{t \in N} (\hat{\tau}_t^N)^{-1}(\pre(\langle P \rangle))$ 
with $(\hat{\tau}_r^{\{r\}})^{-1}(\hat{x}) = \hat{\Pi}_{\{r\}}^N(w)$, $y = \hat{\Pi}_{N \setminus \{r\}}^N(w)$, 
$|\hat{\Pi}_{N \setminus \{r\}}^N(w')| = |\pi_{\shuffle_{\bbP}}(b')|$ and 
$|\hat{\Pi}_{\{r\}}^N(w')| = |\pi_{\shuffle_{\check{\bbP}}}(b')|$  
for each $w'\in \pre(w)$ and $b'\in \pre(b)$ with $|w'|=|b'|$. 
Now by \eqref{eq:characterize-b_P}
$b_\bbP((r,w)) = b$, which completes the proof of equation~\eqref{eq:b_P-5}.
\end{proof}
To prove the main theorem of this section, additionally to \eqref{eq:b_P-1} - \eqref{eq:b_P-5} 
the following characterization of equality in $A_\bbP$ is needed, which is an immediate consequence 
of the definitions in \eqref{eq:def A,Z} and \eqref{eq:def alpha}:
\begin{align} \label{eq:char-eq-A_P}
& \mbox{Let } u,v \in A_\bbP, \mbox{ then } u = v \mbox{ iff } \alpha_\bbP(u) = \alpha_\bbP(v) 
  \mbox{ and } Z_\bbP(u') = Z_\bbP(v') \nonumber\\ 
& \mbox{for each } u'\in \pre(u) \mbox{ and } v'\in \pre(v) \mbox{ with } |u'|=|v'|.
\end{align}
\begin{theorem}\label{thm:eq-rel-RP}
\ \ \ \ $\mcR_\bbP = (c_\bbP \otimes c_\bbP)(\mcR_P)$
\end{theorem}
\begin{proof}\ \\
Let $(w,y) \in \mcR_P$, then $w \in \bigcap\limits_{t \in N} (\hat{\tau}_t^N)^{-1}(\pre(\langle P \rangle))$, 
and there exists  $r \in N$ such that $y = \hat{\Pi}_{N \setminus \{r\}}^N(w)$. 
By \eqref{eq:unique-fac-1} - \eqref{eq:b_P-4} and 
Theorem~\ref{thm:def-cp} together with \eqref{eq:def-cp^2} $b_\bbP((r,w))$ is a shuffled representation 
of $c_\bbP(w)$ by $c_\bbP(y)$ and 
$\check{\iota}_{\shuffle_{\check{\bbP}}}^{-1}(c_\bbP((\hat{\tau}_s^{\{s\}})^{-1}(\hat{x})))$. 
Therefore $(c_\bbP(w),c_\bbP(y)) \in \mcR_\bbP$, which proves 
$(c_\bbP \otimes c_\bbP)(\mcR_P) \subset \mcR_\bbP$.\\
\ \\
To show the contrary inclusion let $(c,d) \in \mcR_\bbP$. Then 
there exists $e\in E_{\check{\bbP}}$ and a shuffled representation $b \in 
\pi_{\shuffle_{\bbP}}^{-1}(A_{\bbP}) \cap \pi_{\shuffle_{\check{\bbP}}}^{-1}(E_{\check{\bbP}})$ 
of $c$ by $d$ and $e$. By \eqref{eq:b_P-5} there exists 
$w \in \bigcap\limits_{t \in N} (\hat{\tau}_t^N)^{-1}(\pre(\langle P \rangle))$, 
and $r \in N$ such that $b = b_\bbP((r,w))$. Now \eqref{eq:characterize-b_P} - \eqref{eq:b_P-4} and 
Theorem~\ref{thm:def-cp} together with \eqref{eq:char-eq-A_P} imply 
$(c,d) = (c_\bbP(w),c_\bbP(\hat{\Pi}_{N \setminus \{r\}}^N(w))) = 
(c_\bbP \otimes c_\bbP)(w,\hat{\Pi}_{N \setminus \{r\}}^N(w)) \in (c_\bbP \otimes c_\bbP)(\mcR_P)$. 
Therefore $\mcR_\bbP \subset (c_\bbP \otimes c_\bbP)(\mcR_P)$, which completes the proof of     
Theorem~\ref{thm:eq-rel-RP}.
\end{proof}
\ \\
Now we consider shuffle projections w.r.t. arbitrary languages. Therefore in 
Definition~\ref{def:def-RP} $\pre(\langle P \rangle)$ has to be replaced by 
$\langle P \rangle \cup \{\varepsilon\}$. So on account of \eqref{eq:np^0} we define:  
\begin{definition}\label{def:def-R0P}
\ \\
Let $\mathring{\mcR}_P := \{(x,y)\in \hat{n}_\bbP^{-1}(0) \times \hat{n}_\bbP^{-1}(0) |
\mbox{there exists } r \in N \mbox{ with } y = \hat{\Pi}_{N \setminus \{r\}}^N (x)\}.$
\end{definition}
Because of $\hat{\Pi}_{N \setminus \{r\}}^N (\hat{n}_\bbP^{-1}(0)) \subset \hat{n}_\bbP^{-1}(0)$ 
it holds 
\begin{equation}\label{eq:def-R0P}
\mathring{\mcR}_P = \mcR_P \cap (\hat{n}_\bbP^{-1}(0) \times \hat{n}_\bbP^{-1}(0)).
\end{equation}
Now by the same argument as in  \eqref{eq:sat-2} 
\begin{equation} \label{eq:sat-R0P-2}
\SP(P \cup \{\varepsilon\},V) \mbox{ iff } 
\S((\wedge\circ\hat{\Theta}_{|\hat{n}_\bbP^{-1}(0)}^N)^{-1}(V),\mathring{\mcR}_P).
\end{equation}
On account of \eqref{subeq:def-cp-2} holds $\wedge\circ\hat{\Theta}_{|\hat{n}_\bbP^{-1}(0)}^N = 
\alpha_\bbP \circ c_{\bbP |\hat{n}_\bbP^{-1}(0)}$. 
Therefore \eqref{eq:sat-1} and \eqref{eq:sat-R0P-2} imply
\[\SP(P \cup \{\varepsilon\},V) \mbox{ iff } 
\S(\alpha_\bbP^{-1}(V),(c_{\bbP |\hat{n}_\bbP^{-1}(0)} \otimes c_{\bbP |\hat{n}_\bbP^{-1}(0)})(\mathring{\mcR}_P)),\]
and because of $(c_{\bbP |\hat{n}_\bbP^{-1}(0)} \otimes c_{\bbP |\hat{n}_\bbP^{-1}(0)})(\mathring{\mcR}_P) = 
(c_\bbP \otimes c_\bbP)(\mathring{\mcR}_P)$ 
\begin{equation} \label{eq:sat-R0P-3}
\SP(P \cup \{\varepsilon\},V) \mbox{ iff } 
\S(\alpha_\bbP^{-1}(V),(c_\bbP \otimes c_\bbP)(\mathring{\mcR}_P)).
\end{equation}
Theorem~\ref{thm:eq-rel-RP} allows to characterize the relation 
$(c_\bbP \otimes c_\bbP)(\mathring{\mcR}_P) \subset A_\bbP \times A_\bbP$ 
without explicit use of $\mathring{\mcR}_P$:
\begin{corollary}\label{cor:eq-rel-R0P}\ \ 
$(c_\bbP \otimes c_\bbP)(\mathring{\mcR}_P) = 
\mcR_\bbP \cap (Z_\bbP^{-1}(0) \times Z_\bbP^{-1}(0)) =: \mathring{\mcR}_\bbP.$ 
\end{corollary}
\begin{proof}
\ \\
\eqref{subeq:def-cp-1} \eqref{eq:def-R0P} and Theorem~\ref{thm:eq-rel-RP} imply 
\begin{align*}
(c_\bbP \otimes c_\bbP)(\mathring{\mcR}_P) =  
& (c_\bbP \otimes c_\bbP)[\mcR_P \cap (\hat{n}_\bbP^{-1}(0) \times \hat{n}_\bbP^{-1}(0))] = \\
& (c_\bbP \otimes c_\bbP)[\mcR_P \cap (c_\bbP^{-1}(Z_\bbP^{-1}(0)) \times c_\bbP^{-1}(Z_\bbP^{-1}(0)))] = \\
& (c_\bbP \otimes c_\bbP)[\mcR_P \cap (c_\bbP^{-1}\otimes c_\bbP^{-1})(Z_\bbP^{-1}(0) \times Z_\bbP^{-1}(0))] = \\
& (c_\bbP \otimes c_\bbP)(\mcR_P) \cap (Z_\bbP^{-1}(0) \times Z_\bbP^{-1}(0)) = \ \\
& \mcR_\bbP \cap (Z_\bbP^{-1}(0) \times Z_\bbP^{-1}(0)),\\
\end{align*}
which completes the proof of  Corollary~\ref{cor:eq-rel-R0P}.
\end{proof}
\ \\
Considering the powerset $2^F$, a binary relation $R \subset F \times F$ defines 
a function $R' : 2^F \to 2^F$ by 
\begin{equation}\label{eq:def-R'}
R'(U) := \{y \in F | \mbox{there exists } x \in U \mbox{ with } (x,y) \in R\} 
\mbox{ for each } U \in 2^F.
\end{equation}
It is an immediate consequence that
\begin{equation}\label{eq:fkt-R'-1}
R'(U) = \bigcup\limits_{x \in U} R'(\{x\}) 
\mbox{ \ \ \ for each } U \in 2^F.
\end{equation}
Now, 
\begin{equation}\label{eq:fkt-R'-2}
\S(W,R) \mbox{ iff } R'(W) \subset W, \mbox{ for each } W \in 2^F. 
\end{equation}
Applying \eqref{eq:fkt-R'-2} to \eqref{eq:sat-3} and Theorem~\ref{thm:eq-rel-RP} 
result in
\begin{corollary}\label{cor:eq-rel-RP'} 
\[\SP(\pre(P),V) \mbox{ iff } \mcR'_\bbP(\alpha_\bbP^{-1}(V)) \subset \alpha_\bbP^{-1}(V).\]
\end{corollary}
Corollary~\ref{cor:eq-rel-R0P} implies
\begin{equation}\label{eq:fkt-R0P'-1}
\mathring{\mcR'}_\bbP(U) = \mcR'_\bbP(U \cap Z_\bbP^{-1}(0)) \cap Z_\bbP^{-1}(0) 
\mbox{ for each } U \subset A_\bbP. 
\end{equation}
On account of \eqref{subeq:shuff-rep-d} holds
\begin{equation}\label{eq:fkt-RP'}
\mcR'_\bbP(Z_\bbP^{-1}(0)) \subset Z_\bbP^{-1}(0),
\end{equation}
and therefore by \eqref{eq:fkt-R0P'-1}
\begin{equation}\label{eq:fkt-R0P'-2}
\mathring{\mcR'}_\bbP(U) = \mcR'_\bbP(U \cap Z_\bbP^{-1}(0))  
\mbox{ for each } U \subset A_\bbP. 
\end{equation}
Now from \eqref{eq:sat-R0P-3}, \eqref{eq:fkt-R'-2}, \eqref{eq:fkt-RP'}, and \eqref{eq:fkt-R0P'-2} it follows
\begin{corollary}\label{cor:eq-rel-R0P'}
\begin{align*}
\SP(P \cup \{\varepsilon\},V) & \mbox{ iff } 
\mcR'_\bbP(\alpha_\bbP^{-1}(V) \cap Z_\bbP^{-1}(0)) \subset \alpha_\bbP^{-1}(V) \\ 
& \mbox{ iff } 
\mcR'_\bbP(\alpha_\bbP^{-1}(V) \cap Z_\bbP^{-1}(0)) \subset \alpha_\bbP^{-1}(V) \cap Z_\bbP^{-1}(0).
\end{align*}
\end{corollary}
\section{Construction Principles}
Under certain conditions for a fixed language $P$ Corollary~\ref{cor:eq-rel-RP'} allows 
to construct a variety of languages $V$ such that $\SP(\pre(P),V)$. The key 
to such constructions is the following implication of \eqref{subeq:shuff-rep-d}:
\begin{equation}\label{eq:fkt-RP'-1}
Z_\bbP(\pre(\mcR'_\bbP(\{c\}))) \subset \bigcup\limits_{x \in \pre(c)} \{f\in\dsN_0^Q | f\le Z_\bbP(x)\} 
\mbox{ for each } c \in A_\bbP,
\end{equation}
where Q is the state set of $\bbP$.    
\begin{definition}[{initial segment}]\label{def:initial segment} \ \\
$\emptyset \neq I \subset \dsN_0^Q$ is called \emph{initial segment} iff $r \le s \in I$ implies $r \in I$. 
For each initial segment $I$, let $A_{(I,\bbP)}:= \{c\in A_\bbP | Z_\bbP(\pre(c)) \subset I \}$.
\end{definition}
It holds $\emptyset \neq  A_{(I,\bbP)} = \pre(A_{(I,\bbP)})$. 
\begin{definition}\label{def:i-s-compatible} \ \\
An initial segment $I$ is called \emph{compatible with $\bbP$} iff $A_{(I,\bbP)}$ is saturated by the partition of $A_\bbP$ 
induced by $\alpha_\bbP$. I.e. $c, c' \in A_{(I,\bbP)}$ and $\alpha_\bbP(c') = \alpha_\bbP(c)$
implies $c' \in A_{(I,\bbP)}$. For an initial segment $I$ compatible with $\bbP$, let $L_{(I,\bbP)}:= \alpha_\bbP(A_{(I,\bbP)})$.
\end{definition}
By this definition $\emptyset \neq L_{(I,\bbP)} \subset (\pre(P))^\shuffle$ and $L_{(I,\bbP)}=\pre(L_{(I,\bbP)})$.
\begin{theorem}\label{thm:i-s-compatible}
Let $\emptyset \neq P \subset \Sigma^*$ and $I$ an initial segment compatible with $\bbP$, then $\SP(\pre(P),L_{(I,\bbP)})$.  
\end{theorem}
\begin{proof}
On account of Corollary~\ref{cor:eq-rel-RP'} and \eqref{eq:fkt-R'-1} 
it is sufficient to show
\begin{equation}\label{eq:i-s-1}
\mcR'_\bbP(\{c\}) \subset \alpha_\bbP^{-1}(L_{(I,\bbP)}) 
\mbox{ for each } c \in \alpha_\bbP^{-1}(L_{(I,\bbP)}).
\end{equation}
Since the initial segment $I$ is compatible with $\bbP$ it holds
\begin{equation}\label{eq:i-s-2}
\alpha_\bbP^{-1}(L_{(I,\bbP)}) = \{x \in A_\bbP | Z_\bbP(\pre(x)) \subset I \}. 
\end{equation}
Now \eqref{eq:fkt-RP'-1} and \eqref{eq:i-s-2} imply \eqref{eq:i-s-1}, 
which completes the proof of Theorem~\ref{thm:i-s-compatible}.
\end{proof}
An immediate consequence of Definition~\ref{def:i-s-compatible} is
\begin{lemma}\label{lemma:i-s-compatible}
An initial segment $I$ is compatible with $\bbP$ iff for each $c, \check{c} \in A_{(I,\bbP)}$ 
with $\alpha(c) = \alpha(\check{c})$, and for each $(Z_\bbP(c),a,f) \in \shuffle_\bbP$ 
and $(Z_\bbP(\check{c}),a,\check{f}) \in \shuffle_\bbP$ holds 
$f \in I$ iff $\check{f} \in I$.
\end{lemma}
The condition of Lemma~\ref{lemma:i-s-compatible} can be checked by a partial 
powerset construction on $\bbP_\shuffle$. 
For this purpose let the partial function 
$\mcD_{(I,\bbP)} : 2^I \times \Sigma \to 2^I$ be defined by
\begin{align}\label{eq:def-D}
& \mcD_{(I,\bbP)}(M,a) := \{f \in \dsN_0^Q | \mbox{ there exist } g \in M \mbox{ and } (g,a,f) \in \shuffle_\bbP \} \nonumber \\
& \mbox{for each } (M,a) \in 2^I \times \Sigma \mbox{ with } \nonumber \\
& \emptyset \neq \{f \in \dsN_0^Q | \mbox{ there exist } g \in M \mbox{ and } (g,a,f) \in \shuffle_\bbP \} \subset I.
\end{align}
The partial function $\mcD_{(I,\bbP)}$ defines a deterministic semiautomaton
\begin{equation}\label{eq:def-A}
\mcP_{(I,\bbP)} := (\Sigma, 2^I,\mcD_{(I,\bbP)}, \{0\}).
\end{equation}
Now Lemma~\ref{lemma:i-s-compatible} implies
\begin{theorem}\label{thm:i-s-compatible2} 
An initial segment $I \subset \dsN_0^Q$ is compatible with $\bbP$, iff \\
for each $a \in \Sigma$ and $M \in 2^I$ reachable in $\mcP_{(I,\bbP)}$ either 
$\mcD_{(I,\bbP)}(M,a)$ is defined, or 
$\{f \in \dsN_0^Q | \mbox{ there exist } g \in M \mbox{ and } (g,a,f) \in \shuffle_\bbP \} \subset \dsN_0^Q \setminus I$. \\
In that case $\mcP_{(I,\bbP)}$ recognizes $L_{(I,\bbP)}$.  
\end{theorem}
\begin{example}\label{ex:i-s-compatible} \ \\
Let $\tilde{P}=\{abc\}$, $\tilde\bbP$ as defined in  Fig.~\ref{fig:i-s-compatible-1},
and $\tilde{I}=\{0, 1_{\small II}, 1_{\small III}, 1_{\small II}+1_{\small III} \}$.
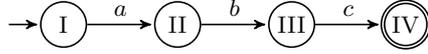
\begin{figure}[h]
\centering
\begin{tikzpicture}
[->,>=stealth',shorten >=1pt,auto,node distance=2cm,semithick,initial text=]
   \node[state,initial,inner sep=1pt,minimum size=6mm] (k1)  {\small I};
   \node[state,inner sep=1pt,minimum size=6mm] (k2) at (1.5,0) {\small II};
   \node[state,inner sep=1pt,minimum size=6mm] (k3) at (3,0) {\small III};
   \node[state,accepting,inner sep=1pt,minimum size=6mm] (k4) at (4.5,0) {\small IV};
   \path 
         (k1) edge node {\small $a$} (k2)
         (k2) edge node {\small $b$} (k3)
         (k3) edge node {\small $c$} (k4)
		;
\end{tikzpicture}
\caption{Automaton $\tilde\bbP$ recognizing $\tilde P$}\label{fig:i-s-compatible-1}
\end{figure}
The partial powerset construction result in the semiautomaton 
$\mcP_{(\tilde{I},\tilde\bbP)}$ of Fig.~\ref{fig:i-s-compatible-2}, which 
fulfills the conditions of Theorem~\ref{thm:i-s-compatible2}. 
Therefore $\tilde{I}$ is compatible with $\tilde\bbP$, which
implies $\SP(\pre(\tilde{P}),L_{(\tilde{I},\tilde\bbP)})$.
\begin{figure}[h]
\centering
\begin{tikzpicture}[->,>=stealth',shorten >=1pt,auto,node distance=1.5cm,semithick,initial text=]
   \node[state,rectangle, rounded corners=8pt,inner sep=4pt,initial] (k1) at (0,0) {\small $\{0\}$};
   \node[state,rectangle, rounded corners=8pt,inner sep=4pt] (k2) at (-1.5,-1.5) {\small $\{1_{\mathrm{II}}\}$};
   \node[state,rectangle, rounded corners=8pt,inner sep=4pt] (k3) at (1.5,-1.5) {\small $\{1_{\mathrm{III}}\}$};
   \node[state,rectangle, rounded corners=8pt,inner sep=4pt] (k4) at (0,-3) {\small $\{1_{\mathrm{II}}+1_{\mathrm{III}}\}$};
   \path 
         (k1) edge  node [swap] {\small $a$} (k2)
         (k2) edge  node {\small $b$} (k3)
         (k3) edge  node [swap] {\small $c$} (k1)
         (k3) edge  node {\small $a$} (k4)
         (k4) edge  node {\small $c$} (k2)
		;
\end{tikzpicture}
\caption{Semiautomaton $\mcP_{(\tilde{I},\tilde\bbP)}$ recognizing $L_{(\tilde{I},\tilde\bbP)}$}\label{fig:i-s-compatible-2}
\end{figure}
\end{example}
It is an immediate consequence of Definition~\ref{def:shuffle-automaton} that
\begin{equation}\label{eq:finite-support}
Z_\bbP(A_\bbP)\subset T(Q) := \{f\in\dsN_0^Q | \{q\in Q | f(q)\neq 0\} \mbox{ is a finite set.} \}
\end{equation}
for each deterministic automaton $\bbP$ with state set $Q$ (not necessarily finite).
\ \\
\ \\
There are special initial sections $I \subset T(Q)$ and automata $\bbP$ with state set $Q$, 
such that compatibility of $I$ with $\bbP$ can be verified easily:
\begin{equation}\label{eq:easy-compatible-1}
\mbox{For } f \in T(Q) \mbox{ let } \lVert{f}\rVert:=\sum\limits_{q\in Q} f(q)\in\dsN_0.
\end{equation}
\begin{equation}\label{eq:easy-compatible-2}
\mbox{For } n\in \dsN_0 \mbox{ let } K(n,Q):= \{f \in T(Q) | \lVert{f}\rVert \le n\},
\end{equation}
which is an initial segment.
\begin{theorem}\label{thm:easy-compatible} \ \\
Let $\Phi$, $\Gamma$, and $\Omega$ be pairwise disjoint sets,
$\emptyset\neq P\subset \Gamma \cup \Phi\Gamma^*\Omega$ and $\bbP$
be a deterministic automaton with state set $Q$ recognizing $P$.
Then $K(n,Q)$ is compatible with $\bbP$ for each $n\in \dsN_0$, and 
therefore  $\SP(\pre(P),L_{(K(n,Q),\bbP)})$.
\end{theorem}
\begin{proof} \ \\
From Definition~\ref{def:shuffle-automaton} it follows for each $(f,a,g) \in \shuffle_\bbP$
\begin{align}\label{eq:easy-compatible-3}
& a \in \Phi \mbox{ implies } \lVert{g}\rVert = \lVert{f}\rVert +1, \nonumber \\
& a \in \Gamma \mbox{ implies } \lVert{g}\rVert = \lVert{f}\rVert, \mbox{ and} \nonumber \\
& a \in \Omega \mbox{ implies } \lVert{g}\rVert = \lVert{f}\rVert -1.
\end{align}
Therefore
\begin{equation}\label{eq:easy-compatible-4}
f,f' \in M \mbox{ implies } \lVert{f}\rVert = \lVert{f'}\rVert 
\mbox{ for each state } M \mbox{ reachable in } \mcP_{(K(n,Q),\bbP)}.
\end{equation}
Now \eqref{eq:easy-compatible-3} and \eqref{eq:easy-compatible-4} together 
with Theorem~\ref{thm:i-s-compatible} completes the proof.
\end{proof}
\begin{example}\label{ex:easy-compatible} \ \\
Let $\bar\bbP$ and $\bar{P}$ as defined in Figure~\ref{fig:easy-compatible-1}. Then by Theorem~\ref{thm:easy-compatible}
$K(n,\bar{Q})$ is compatible with $\bar\bbP$ for each $n\in \dsN_0$, where $\bar{Q}$ is the state set of $\bar\bbP$,
and it holds $\SP(\pre(\bar{P}),L_{(K(n,\bar{Q}),\bar\bbP)})$ for each $n\in \dsN_0$.
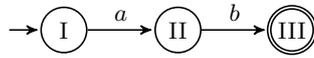
\begin{figure}[h]
\centering
\begin{tikzpicture}[->,>=stealth',shorten >=1pt,auto,node distance=2cm,semithick,initial text=]
   \node[state,initial,inner sep=1pt,minimum size=6mm] (k1)  {\small I};
   \node[state,inner sep=1pt,minimum size=6mm] (k2) at (1.5,0) {\small II};
   \node[state,accepting,inner sep=1pt,minimum size=6mm] (k3) at (3,0) {\small III};
   \node[init] (unsichtbar) [right of =k3,node distance=1.5cm] {};
   \path 
         (k1) edge node {\small $a$} (k2)
         (k2) edge node {\small $b$} (k3)
		;
\end{tikzpicture}
\caption{Automaton $\bar\bbP$ recognizing $\bar{P}:=\{ab\}$}\label{fig:easy-compatible-1}
\end{figure}
\ \\
Figure~\ref{fig:easy-compatible-2} shows the semiautomaton $\mcP_{(K(n,\bar{Q}),\bar\bbP)}$.
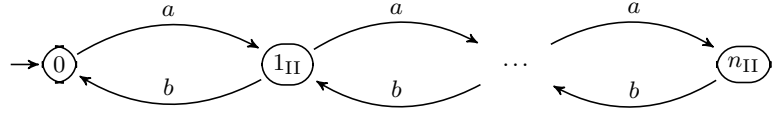
\begin{figure}[h]
\centering
\begin{tikzpicture}[->,>=stealth',shorten >=1pt,auto,node distance=2cm,semithick,initial text=]
   \node[state,initial,rectangle, rounded corners=8pt,inner sep=4pt] (k1) 
   {\small $0$};
   \node[state,rectangle, rounded corners=8pt,inner sep=4pt] (k2) at (3,0) 
   {\small $1_{\mathrm{II}}$};
   \node[state,rectangle, rounded corners=8pt,inner sep=8pt,draw opacity=0] (inv) at (6,0) 
   {\small $\ldots$};
   \node[state,rectangle, rounded corners=8pt,inner sep=4pt] (k3) at (9,0) 
   {\small $n_{\mathrm{II}}$};
   \path 
         (k1) edge  [bend left] node  {\small $a$} (k2)
         (k2) edge  [bend left] node [swap] {\small $b$} (k1)
         (k2) edge  [bend left] node  {\small $a$} (inv)
         (inv) edge  [bend left] node [swap] {\small $b$} (k2)
         (inv) edge  [bend left] node  {\small $a$} (k3)
         (k3) edge  [bend left] node [swap] {\small $b$} (inv)
		;
\end{tikzpicture} 
\caption{Semiautomaton recognizing $L_{(K(n,\bar{Q}),\bar\bbP)}$ for each $n\in \dsN_0$}\label{fig:easy-compatible-2}
\end{figure}
\end{example} 
The following example is a bridge to the next section.
\begin{example}\label{ex:not-i-s-compatible} \ \\
Let $\mathring{P}$ and $\mathring{\bbP}$ as defined in Fig.~\ref{fig:not-i-s-compatible-1}.
\begin{figure}[h]
\centering
\begin{tikzpicture}[->,>=stealth',shorten >=1pt,auto,node distance=2cm,semithick,initial text=]
   \node[state,initial,inner sep=4pt] (k1) 
   {\small $\mathrm{I}$};
   \node[state,inner sep=4pt] (k2) at (3,0) 
   {\small $\mathrm{II}$};
   \node[state,accepting,inner sep=4pt] (k3) at (6,0) 
   {\small $\mathrm{III}$};
   \path 
         (k1) edge  [bend left] node  {\small $a$} (k2)
         (k1) edge  [bend right] node [swap] {\small $b$} (k2)
         (k2) edge  node  {\small $c$} (k3)
		;
\end{tikzpicture} 
\caption{Automaton $\mathring{\bbP}$ recognizing $\mathring{P}$}\label{fig:not-i-s-compatible-1}
\end{figure}
It holds $Z_{\mathring{\bbP}}(A_{\mathring{\bbP}}) = \{0\} \cup \{n_{\mathrm{II}} | n\in \dsN \}$.
Therefore, $ab\in L_{(\mathring{I},\mathring{\bbP})}$ implies $ba\in L_{(\mathring{I},\mathring{\bbP})}$
for each initial segment $\mathring{I}$ compatible with $\mathring{\bbP}$.
\ \\
\begin{figure}[h]
\centering
\begin{tikzpicture}[->,>=stealth',shorten >=1pt,auto,node distance=2cm,semithick,initial text=]
   \node[state,initial,inner sep=4pt] (k1) 
   {\small $1$};
   \node[state,inner sep=4pt] (k2) at (0,-2) 
   {\small $2$};
   \node[state,inner sep=4pt] (k3) at (3,0) 
   {\small $3$};
   \node[state,inner sep=4pt] (k4) at (3,-2) 
   {\small $4$};
   \path 
         (k1) edge  node  {\small $a$} (k2)
         (k1) edge  node [swap] {\small $b$} (k3)
         (k2) edge  node  {\small $b$} (k4)
         (k2) edge  [bend left] node  {\small $c$} (k1)
         (k3) edge  [bend right] node [swap] {\small $c$} (k1)
         (k3) edge  [bend left] node  {\small $b$} (k4)
         (k4) edge  node  {\small $c$} (k3)
		;
\end{tikzpicture} 
\caption{Semiautomaton ${\mathring{\bbV}}$ recognizing ${\mathring{V}}$}\label{fig:not-i-s-compatible-2}
\end{figure}
\ \\
Let the prefix closed language $\mathring{V}$ be defined by the semiautomaton in Fig.~\ref{fig:not-i-s-compatible-2}. 
Because of $ab \in \mathring{V}$ but $ba \notin \mathring{V}$,
$\mathring{V}$ cannot be represented by 
$\mathring{V}= L_{(\mathring{I}_{\mathring{V}},\mathring{\bbP})}$ 
with an initial segment $\mathring{I}_{\mathring{V}}$ compatible
with $\mathring{\bbP}$.
So $\SP(\pre(\mathring{P}),\mathring{V})$ cannot be shown by theorem~\ref{thm:i-s-compatible}. 
But in the next section a method will be developed to prove $\SP(\pre(\mathring{P}),\mathring{V})$.
\end{example}
\section{Representation Theorem}\label{sec:representation theorem}
In this section a representation of $\mcR'_\bbP$ will be developed, which shows certain 
restrictions of $\mcR'_\bbP$ to be rational transductions \cite{berstel79}. More precisely: 
Depending on a subset $\Delta \subset \shuffle_\bbP$, an alphabet $\Delta^{()}$ and a 
prefix closed language $W_\Delta \subset \Delta^{()*}$ will be constructed, which represents 
the function $\mcR'_{\bbP|2^{A_\bbP \cap \Delta^*}}$ in the following manner: \\
\ \\
There exist two alphabetic homomorphisms $\mu_\Delta : \Delta^{()*} \to \Delta^*$ 
and $\nu_\Delta : \Delta^{()*} \to \shuffle_\bbP^*$ such that for each  
$c \in A_\bbP \cap \Delta^*$ it holds,
\begin{equation*}
d \in \mcR'_\bbP(\{c\}) \mbox{ iff there exists } x \in W_\Delta  \mbox{ with } 
c = \mu_\Delta(x) \mbox{ and } d = \nu_\Delta(x),
\end{equation*}
which is equivalent to
\begin{equation}\label{eq:RT-179}
\mcR'_\bbP(B) = \nu_\Delta(\mu_\Delta^{-1}(B) \cap W_\Delta) 
\mbox{ for each } B \subset A_\bbP \cap \Delta^*.
\end{equation}
Additionally, it will be shown that
\begin{equation}\label{eq:RT-180}
W_\Delta \mbox{ is regular if } \Delta \mbox{ is finite.} 
\end{equation}
In that case $\mcR'_{\bbP|2^{A_\bbP \cap \Delta^*}}$ is a rational transduction \cite{berstel79}. \\
\ \\
On account of \eqref{eq:finite-support} it can be assumed
\begin{equation}\label{eq:RT-180.1}
\Delta \subset \shuffle_\bbP \cap (T(Q) \times \Sigma \times T(Q)). 
\end{equation}
The construction of $W_\Delta$ is based on the following idea: Each $x \in W_\Delta$ 
uniquely describes a shuffled representation $b$ of $c \in A_\bbP \cap \Delta^*$ 
by $d\in A_\bbP$ and $e\in E_{\check{\bbP}}$ as defined in \eqref{subeq:shuff-rep}. 
This description is structured into three tracks, respectively one for $c$, $d$, and $e$. 
Additionally the second and third track describe the position of $d$ and $e$ in $b$ 
such that both tracks together represent $b$. 
These three tracks will be formalized by three components of the elements of $\Delta^{()}$. \\
\ \\
By an appropriate definition of $\Delta^{()}$, $W_\Delta$ can be defined as a \emph{local} 
prefix closed language \cite{berstel79}. So $W_\Delta$ will be defined by $\Delta^{()}$, the 
set of initial letters of its words and the set of forbidden adjacencies of letters in its words. 
Generally, local languages with a finite alphabet are regular languages \cite{berstel79}.  
Starting basis for this are the definitions of $A_\bbP \cap \Delta^*$ and $E_{\check{\bbP}}$ 
as local prefix closed languages: \eqref{eq:def A,Z} imply 
\begin{align}\label{eq:local-def-A}
& A_\bbP \cap \Delta^* = \nonumber \\
& (\{ \varepsilon \} \cup \{(f,a,g) \in \Delta | f = 0 \} \Delta^*) 
\setminus \Delta^* \{(f,a,g)(f',a',g') \in \Delta\Delta | g \neq f' \} \Delta^*. 
\end{align}
With $\shuffle_{\check{\bbP}}^E := \{(f,a,g) \in \shuffle_{\check{\bbP}} | f,g \in 
\{ 0 \} \cup \{ 1_q \in \dsN_0^Q | q \in Q \} \} $ \eqref{eq:char-E_P} imply 
\begin{align}\label{eq:local-def-E}
& E_{\check{\bbP}} = 
(\{ \varepsilon \} \cup \{(f,a,g) \in \shuffle_{\check{\bbP}}^E | f = 0 \} \shuffle_{\check{\bbP}}^{E*}) \setminus \nonumber \\
& \shuffle_{\check{\bbP}}^{E*} \{(f,a,g)(f',a',g') \in \shuffle_{\check{\bbP}}^E\shuffle_{\check{\bbP}}^E
| g \neq f' \mbox{ or } g = 0 \} \shuffle_{\check{\bbP}}^{E*}. 
\end{align}
To achieve \eqref{eq:RT-179} and \eqref{eq:RT-180}, $W_\Delta$ has to be defined in such a way, 
that for each $x \in W_\Delta$ the corresponding $c$ and $d$ can be extracted from $x$ by 
alphabetic homomorphisms, and that finiteness of $\Delta$ implies regularity of $W_\Delta$. 
For formalization let
\begin{equation}\label{eq:RT-183}
\Delta^{()} \subset \Delta^{()'} := \Delta^{(1)} \times \Delta^{(2)} \times \Delta^{(3)},
\end{equation}
where $\Delta^{(1)}$ is the alphabet for representing $c \in A_\bbP \cap \Delta^*$, and 
$\Delta^{(2)} \times \Delta^{(3)}$ is the alphabet for representing 
$b \in \pi_{\shuffle_{\bbP}}^{-1}(A_{\bbP}) \cap \pi_{\shuffle_{\check{\bbP}}}^{-1}(E_{\check{\bbP}})$, 
the shuffled representation of $c$. This is possible, since $|c| = |b|$ because of 
\eqref{subeq:shuff-rep-a}.
\ \\
\ \\
In that representation $\Delta^{(2)}$ is the alphabet for representing 
$d = \pi_{\shuffle_{\bbP}}(b) \in A_\bbP$ as well as for describing the 
positioning of $d$ inside $b$. 
$\Delta^{(3)}$ is the alphabet for representing 
$e = \pi_{\shuffle_{\check{\bbP}}}(b) \in E_{\check{\bbP}}$ as well as for describing the 
positioning of $e$ inside $b$.
Additionally it should be noticed that each 
$b \in (\shuffle_{\bbP} \dotcup \shuffle_{\check{\bbP}})^*$ is uniquely determined 
by $\pi_{\shuffle_{\bbP}}(b)$, $\pi_{\shuffle_{\check{\bbP}}}(b)$ and by the information,
which positions of $b$ contain elements of $\shuffle_{\bbP}$ and which positions 
contain elements of $\shuffle_{\check{\bbP}}$. \\
\ \\
As $\Delta^{(1)}$ is the alphabet for representing $c \in A_\bbP \cap \Delta^*$, let  
\begin{equation}\label{eq:RT-184}
\Delta^{(1)} := \Delta,
\end{equation}
and let $\varphi_\Delta^{(1)} : \Delta^{()'*} \to \Delta^{(1)*}$ be the homomorphism defined by
\begin{equation}\label{eq:RT-185}
\varphi_\Delta^{(1)}((x_1,x_2,x_3)) := x_1 \mbox{ for } (x_1,x_2,x_3) \in \Delta^{()'}.
\end{equation}
Then $\varphi_\Delta^{(1)}(x) \in A_\bbP \cap \Delta^*$ should hold for each $x \in W_\Delta$. \\
\ \\
Let now the mappings $\varphi_\Delta^{(1.1)}$, $\varphi_\Delta^{(1.2)}$ and $\varphi_\Delta^{(1.3)}$ be defined by
\begin{equation}\label{eq:RT-186}
\varphi_\Delta^{(1.1)} : \Delta^{(1)} \to \dsN_0^Q \mbox{ with } 
\varphi_\Delta^{(1.1)}((f,a,g)) := f,
\end{equation}
\begin{equation}\label{eq:RT-186'}
\varphi_\Delta^{(1.2)} : \Delta^{(1)} \to \Sigma \mbox{ with } 
\varphi_\Delta^{(1.2)}((f,a,g)) := a,
\end{equation}
and
\begin{equation}\label{eq:RT-187}
\varphi_\Delta^{(1.3)} : \Delta^{(1)} \to \dsN_0^Q \mbox{ with } 
\varphi_\Delta^{(1.3)}((f,a,g)) := g
\end{equation}
for each $(f,a,g) \in \Delta^{(1)}$. \\
\ \\
Then \eqref{eq:local-def-A} becomes
\begin{equation*}
A_\bbP \cap \Delta^* = 
(\{ \varepsilon \} \cup (\varphi_\Delta^{(1.1)})^{-1}(0) \Delta^{(1)*}) 
\setminus \Delta^{(1)*} F^{(1)} \Delta^{(1)*} 
\end{equation*}
with
\begin{equation}\label{eq:RT-187'}
F^{(1)} := \{xy \in \Delta^{(1)}\Delta^{(1)} | \varphi_\Delta^{(1.3)}(x) \neq \varphi_\Delta^{(1.1)}(y) \}.
\end{equation}
Therefore $\varphi_\Delta^{(1)}(W_\Delta) \subset A_\bbP \cap \Delta^*$ if
\begin{equation}\label{eq:RT-188}
\varphi_\Delta^{(1)}(W_\Delta) \subset
(\{ \varepsilon \} \cup (\varphi_\Delta^{(1.1)})^{-1}(0) \Delta^{(1)*}) 
\setminus \Delta^{(1)*} F^{(1)} \Delta^{(1)*}. 
\end{equation}
With two further conditions similar to \eqref{eq:RT-188} and additional 
restrictions of the alphabet $\Delta^{()'}$ the language $W_\Delta$ will be defined. 
But first the sets $\Delta^{(2)}$ and $\Delta^{(3)}$ have to be defined. \\
\ \\
Since the elements of $W_\Delta$ particularly have to represent condition 
\eqref{subeq:shuff-rep-d} let
\begin{equation}\label{eq:RT-190}
S_\Delta^{(1)} := Z_\bbP(A_\bbP \cap \Delta^*),
\end{equation}
\begin{equation}\label{eq:RT-191}
S_\Delta^{(3)} := \{f \in Z_{\check{\bbP}}(E_{\check{\bbP}}) | 
\mbox{ there exists } g \in S_\Delta^{(1)} \mbox{ with } f \le g \},
\end{equation}
and
\begin{equation}\label{eq:RT-192}
S_\Delta^{(2)} := \{f \in Z_\bbP(A_\bbP) | 
\mbox{ there exists } g \in S_\Delta^{(1)} \mbox{ and } h \in S_\Delta^{(3)}  
\mbox{ with } g = f + h \}.
\end{equation}
Now, on account of \eqref{eq:RT-180.1} 
\begin{align}\label{eq:RT-193} 
& \mbox{Finiteness of } \Delta  \mbox{ implies finiteness of } 
  S_\Delta^{(1)}, S_\Delta^{(2)} \mbox{ and of } S_\Delta^{(3)},\nonumber \\ 
& \mbox{which can be effectively determined. }
\end{align}
With
$\Sigma_\Delta := \varphi_\Delta^{(1.2)}(\Delta^{(1)}) = \varphi_\Delta^{(1.2)}(\Delta) \subset \Sigma$ 
\begin{equation}\label{eq:RT-194'} 
\mbox{finiteness of } \Delta  \mbox{ implies finiteness of } \Sigma_\Delta.
\end{equation}
By the definition  
\begin{equation}\label{eq:RT-195} 
\Delta^{(2)'} := \shuffle_\bbP \cap (S_\Delta^{(2)} \times \Sigma_\Delta \times S_\Delta^{(2)})
\end{equation}
holds $d \in A_\bbP \cap \Delta^{(2)'*}$ for $d = \pi_{\shuffle_{\bbP}}(b)$ 
because of \eqref{subeq:shuff-rep-d}. \\
\ \\
Now $S_\Delta^{(2)}$ is used to describe the positioning of $d$ inside $b$. 
Let therefore
\begin{align}\label{eq:RT-196} 
& \Delta^{(2)} := \Delta^{(2)'} \dotcup S_\Delta^{(2)}, \mbox{ which is finite if } \Delta \mbox{ is finite,}\nonumber \\ 
& \mbox{and can be effectively determined. }
\end{align}
Let the homomorphisms $\varphi_\Delta^{(2)} : \Delta^{()'*} \to \Delta^{(2)*}$ and 
$\gamma_\Delta^{(2)} : \Delta^{(2)*} \to \Delta^{(2)'*}$ be defined by
\begin{align}\label{eq:RT-197} 
& \varphi_\Delta^{(2)}((x_1,x_2,x_3)) := x_2 \mbox{ for } (x_1,x_2,x_3) \in \Delta^{()'}, \nonumber \\
& \gamma_\Delta^{(2)}(y) := y \mbox{ for } y \in \Delta^{(2)'} \mbox{ and } \nonumber \\
& \gamma_\Delta^{(2)}(y) := \varepsilon \mbox{ for } y \in S_\Delta^{(2)}.
\end{align}
Now, on account of \eqref{subeq:shuff-rep-c}
$\gamma_\Delta^{(2)}(\varphi_\Delta^{(2)}(x)) \in A_\bbP \cap \Delta^{(2)'*}$ 
should hold for each $x \in W_\Delta$. \\
\ \\
With the mappings $\varphi_\Delta^{(2.1)} : \Delta^{(2)} \to S_\Delta^{(2)}$ and 
$\varphi_\Delta^{(2.3)} : \Delta^{(2)} \to S_\Delta^{(2)}$ defined by
\begin{align}\label{eq:RT-198}
& \varphi_\Delta^{(2.1)}((f,a,g)) := f \mbox{ and } \varphi_\Delta^{(2.3)}((f,a,g)) := g 
\mbox{ for } (f,a,g) \in \Delta^{(2)'} \mbox{ and } \nonumber \\ 
& \varphi_\Delta^{(2.1)}(f) := \varphi_\Delta^{(2.3)}(f) := f \mbox{ for } f \in S_\Delta^{(2)}
\end{align}
it holds $\gamma_\Delta^{(2)}(\varphi_\Delta^{(2)}(W_\Delta)) \subset A_\bbP \cap \Delta^{(2)'*}$ if
\begin{align}\label{eq:RT-199}
& \varphi_\Delta^{(2)}(W_\Delta) \subset 
(\{ \varepsilon \} \cup (\varphi_\Delta^{(2.1)})^{-1}(0) \Delta^{(2)*}) 
\setminus \Delta^{(2)*} F^{(2)} \Delta^{(2)*} \mbox{ where } \nonumber \\ 
& F^{(2)} := \{xy \in \Delta^{(2)}\Delta^{(2)} | \varphi_\Delta^{(2.3)}(x) \neq \varphi_\Delta^{(2.1)}(y) \}.
\end{align}
Let the mapping $Z_\Delta^{(2)} : \Delta^{(2)*} \to S_\Delta^{(2)}$ be defined by
\begin{equation}\label{eq:RT-199'}
Z_\Delta^{(2)}(\varepsilon) := 0, \mbox{ and } Z_\Delta^{(2)}(uv) := \varphi_\Delta^{(2.3)}(v)
\mbox{ for } u \in \Delta^{(2)*} \mbox{ and } v \in \Delta^{(2)}.
\end{equation}
Then \eqref{eq:RT-199} implies
\begin{equation}\label{eq:RT-199''}
Z_\Delta^{(2)}(\varphi_\Delta^{(2)}(x)) = Z_\bbP(\gamma_\Delta^{(2)}(\varphi_\Delta^{(2)}(x)))
\mbox{ for each } x \in W_\Delta.
\end{equation}
Now, the definitions concerning $\Delta^{(3)}$ are similar to those concerning $\Delta^{(2)}$. 
But additionally it must be pointed out that \[E_{\check{\bbP}} \subset \shuffle_{\check{\bbP}}^{E*} \setminus
(\shuffle_{\check{\bbP}}^{E*} \{(f,a,g)(f',a',g') \in \shuffle_{\check{\bbP}}^E\shuffle_{\check{\bbP}}^E
| g = 0 \} \shuffle_{\check{\bbP}}^{E*}).\] 
Therefore we use an additional letter 
$\check{0} \notin \shuffle_{\check{\bbP}}^E \dotcup S_\Delta^{(3)}$ 
to define the content of the third track by a prefix closed local language 
such that \[\varphi_\Delta^{(3)}(W_\Delta) \subset 
\pre((S_\Delta^{(3)} \cup \{(f,a,g) \in \shuffle_{\check{\bbP}}^E | g \neq 0 \})^*
\{(f,a,g) \in \shuffle_{\check{\bbP}}^E | g = 0 \} \{ \check{0} \}^*).\] 
So let
\begin{align}\label{eq:RT-200} 
& \Delta^{(3)'} := \shuffle_{\check{\bbP}}^E \cap (S_\Delta^{(3)} \times \check{\iota}^{-1}(\Sigma_\Delta) \times S_\Delta^{(3)})
  \mbox{ and } 
  \Delta^{(3)} := \Delta^{(3)'} \dotcup S_\Delta^{(3)} \dotcup \{ \check{0} \},\nonumber \\ 
& \mbox{which are finite and can be effectively determined, if } \Delta \mbox{ is finite.}
\end{align}
By this definition of $\Delta^{(3)'}$ holds $e \in E_{\check{\bbP}} \cap \Delta^{(3)'*}$ 
for $e = \pi_{\shuffle_{\check{\bbP}}}(b)$ 
because of \eqref{subeq:shuff-rep-d}. 
$S_\Delta^{(3)} \dotcup \{ \check{0} \}$ is used to describe the positioning of $e$ inside $b$. \\
\ \\
Let the homomorphisms $\varphi_\Delta^{(3)} : \Delta^{()'*} \to \Delta^{(3)*}$ and 
$\gamma_\Delta^{(3)} : \Delta^{()*} \to \Delta^{(3)'*}$ be defined by
\begin{align}\label{eq:RT-202} 
& \varphi_\Delta^{(3)}((x_1,x_2,x_3)) := x_3 \mbox{ for } (x_1,x_2,x_3) \in \Delta^{()'}, \nonumber \\
& \gamma_\Delta^{(3)}(y) := y \mbox{ for } y \in \Delta^{(3)'} \mbox{ and } \nonumber \\
& \gamma_\Delta^{(3)}(y) := \varepsilon \mbox{ for } y \in S_\Delta^{(3)} \dotcup \{ \check{0} \}.
\end{align}
Now, on account of \eqref{subeq:shuff-rep-b}
$\gamma_\Delta^{(3)}(\varphi_\Delta^{(3)}(x)) \in E_{\check{\bbP}} \cap \Delta^{(3)'*}$ 
should hold for each $x \in W_\Delta$. \\
\ \\
With the mappings $\varphi_\Delta^{(3.1)} : \Delta^{(3)} \to S_\Delta^{(3)}$ and 
$\varphi_\Delta^{(3.3)} : \Delta^{(3)} \to S_\Delta^{(3)}$ defined by
\begin{align}\label{eq:RT-205}
& \varphi_\Delta^{(3.1)}((f,a,g)) := f \mbox{ and } \varphi_\Delta^{(3.3)}((f,a,g)) := g 
\mbox{ for } (f,a,g) \in \Delta^{(3)'}, \nonumber \\ 
& \varphi_\Delta^{(3.1)}(f) := \varphi_\Delta^{(3.3)}(f) := f \mbox{ for } f \in S_\Delta^{(3)} \mbox{ and } \nonumber \\
& \varphi_\Delta^{(3.1)}(\check{0}) := \varphi_\Delta^{(3.3)}(\check{0}) := 0 .
\end{align}
it holds $\gamma_\Delta^{(3)}(\varphi_\Delta^{(3)}(W_\Delta)) \subset E_{\check{\bbP}} \cap \Delta^{(3)'*}$ if
\begin{align}\label{eq:RT-206}
& \varphi_\Delta^{(3)}(W_\Delta) \subset 
(\{ \varepsilon \} \cup ((\varphi_\Delta^{(3.1)})^{-1}(0) \setminus \{ \check{0} \}) \Delta^{(3)*}) \  
\setminus \ \Delta^{(3)*} F^{(3)} \Delta^{(3)*} \mbox{ where } \nonumber \\ 
& F^{(3)} := \{xy \in \Delta^{(3)}\Delta^{(3)} | \varphi_\Delta^{(3.3)}(x) \neq \varphi_\Delta^{(3.1)}(y) \} \  \cup \nonumber \\
& ((\Delta^{(3)'} \cap (\varphi_\Delta^{(3.3)})^{-1}(0))\cup \{ \check{0} \})
  (\Delta^{(3)} \setminus \{ \check{0} \}) \  \cup \nonumber \\
& (\Delta^{(3)} \setminus ((\Delta^{(3)'} \cap (\varphi_\Delta^{(3.3)})^{-1}(0))\cup \{ \check{0} \})) \{ \check{0} \}.
\end{align}
Let the mapping $Z_\Delta^{(3)} : \Delta^{(3)*} \to S_\Delta^{(3)}$ be defined by
\begin{equation}\label{eq:RT-206'}
Z_\Delta^{(3)}(\varepsilon) := 0, \mbox{ and } Z_\Delta^{(3)}(uv) := \varphi_\Delta^{(3.3)}(v)
\mbox{ for } u \in \Delta^{(3)*} \mbox{ and } v \in \Delta^{(3)}.
\end{equation}
Then \eqref{eq:RT-206} implies
\begin{equation}\label{eq:RT-206''}
Z_\Delta^{(3)}(\varphi_\Delta^{(3)}(x)) = Z_\bbP(\gamma_\Delta^{(3)}(\varphi_\Delta^{(3)}(x)))
\mbox{ for each } x \in W_\Delta.
\end{equation}
\ \\
Now the conditions \eqref{subeq:shuff-rep-a} and \eqref{subeq:shuff-rep-d} 
imply restrictions of the set $\Delta^{()'}$, which finally define the alphabet 
\begin{equation*}
\Delta^{()} \subset \Delta^{()'} = \Delta^{(1)} \times \Delta^{(2)} \times \Delta^{(3)} 
= \Delta \times (\Delta^{(2)'} \dotcup S_\Delta^{(2)}) \times 
(\Delta^{(3)'} \dotcup S_\Delta^{(3)} \dotcup \{ \check{0} \}).
\end{equation*}	
For that purpose let the mappings $\varphi_\Delta^{(2.2)} : \Delta^{(2)'} \to \Sigma$ 
and $\varphi_\Delta^{(3.2)} : \Delta^{(3)'} \to \check{\Sigma}$ be defined by
\begin{equation}\label{eq:RT-206'''}
\varphi_\Delta^{(i.2)}((f,a,g)) := a \mbox{ for } (f,a,g) \in \Delta^{(i)'} 
\mbox{ with } i \in \{ 2,3 \}.
\end{equation}
As the second and third track together represent a shuffled representation, 
\eqref{subeq:shuff-rep-a} requires
\begin{align}\label{eq:RT-207}
& \mbox{either } x_2 \in \Delta^{(2)'}, \ x_3 \in S_\Delta^{(3)} \dotcup \{ \check{0} \} 
  \mbox{ and } \varphi_\Delta^{(1.2)}(x_1) = \varphi_\Delta^{(2.2)}(x_2) \nonumber \\
& \mbox{or } x_2 \in S_\Delta^{(2)}, \ x_3 \in \Delta^{(3)'} 
  \mbox{ and } \varphi_\Delta^{(1.2)}(x_1) = \check{\iota}(\varphi_\Delta^{(3.2)}(x_3)) \nonumber \\
& \mbox{for each } (x_1,x_2,x_3) \in \Delta^{()}.
\end{align}
Additionally \eqref{subeq:shuff-rep-d} requires
\begin{align}\label{eq:RT-208}
& \varphi_\Delta^{(1.1)}(x_1) = \varphi_\Delta^{(2.1)}(x_2) + \varphi_\Delta^{(3.1)}(x_3) 
  \mbox{ and } \nonumber \\
& \varphi_\Delta^{(1.3)}(x_1) = \varphi_\Delta^{(2.3)}(x_2) + \varphi_\Delta^{(3.3)}(x_3) 
  \mbox{ for each } (x_1,x_2,x_3) \in \Delta^{()}.
\end{align}
Let therefore
\begin{equation}\label{eq:RT-209}
\Delta^{()} := \{(x_1,x_2,x_3) \in \Delta^{()'} | \mbox{ it holds } \eqref{eq:RT-207} 
\mbox{ and } \eqref{eq:RT-208} \},
\end{equation}
which is finite and can be effectively determined, if $\Delta$ is finite. \\
\ \\
Combining \eqref{eq:RT-209} with \eqref{eq:RT-188}, \eqref{eq:RT-199} and \eqref{eq:RT-206} 
result in
\begin{definition}\label{def:RT-210} 
Let $\Delta \subset \shuffle_\bbP$, then
\begin{align*}
W_\Delta := & \Delta^{()*} \ \cap \\
& (\varphi_\Delta^{(1)})^{-1}[ 
(\{ \varepsilon \} \cup (\varphi_\Delta^{(1.1)})^{-1}(0) \Delta^{(1)*}) \  
\setminus \ \Delta^{(1)*} F^{(1)} \Delta^{(1)*}] \ \cap \\
& (\varphi_\Delta^{(2)})^{-1}[ 
(\{ \varepsilon \} \cup (\varphi_\Delta^{(2.1)})^{-1}(0) \Delta^{(2)*}) \  
\setminus \ \Delta^{(2)*} F^{(2)} \Delta^{(2)*}] \ \cap \\
& (\varphi_\Delta^{(3)})^{-1}[ 
(\{ \varepsilon \} \cup ((\varphi_\Delta^{(3.1)})^{-1}(0) \setminus \{ \check{0} \}) \Delta^{(3)*}) \  
\setminus \ \Delta^{(3)*} F^{(3)} \Delta^{(3)*}]. 
\end{align*}
\end{definition}
By the well known closure properties of the class of regular languages \cite{berstel79} 
this representation shows that $W_\Delta$ is regular, if $\Delta$ is finite, 
and it is a prefix closed local language, because of 
$\varphi_\Delta^{(i)}( \Delta^{()'}) \subset \Delta^{(i)}$ for each $i \in \{ 1,2,3 \}$. \\
\ \\
To show that $W_\Delta$ represents 
the function $\mcR'_{\bbP|2^{A_\bbP \cap \Delta^*}}$, we need an additional homomorphism 
$\eta_\Delta : \Delta^{()*} \to (\shuffle_{\bbP} \dotcup \shuffle_{\check{\bbP}})^*$, 
defined by
\begin{align}\label{eq:RT-210.1}
& \eta_\Delta((x_1,x_2,x_3)) := x_2 \mbox{ for } (x_1,x_2,x_3) \in \Delta^{()} 
  \mbox{ with } x_2 \in \shuffle_{\bbP} \nonumber \\
& \mbox{and } \nonumber \\
& \eta_\Delta((x_1,x_2,x_3)) := x_3 \mbox{ for } (x_1,x_2,x_3) \in \Delta^{()} 
  \mbox{ with } x_3 \in \shuffle_{\check{\bbP}}.
\end{align}
By \eqref{eq:RT-210.1} $\eta_\Delta$ is well defined, because 
\[\Delta^{()} = \{(x_1,x_2,x_3) \in \Delta^{()} | x_2 \in \shuffle_{\bbP} \} 
\dotcup \{(x_1,x_2,x_3) \in \Delta^{()} | x_3 \in \shuffle_{\check{\bbP}} \} \]
on account of \eqref{eq:RT-207}. \\
\ \\
\eqref{eq:RT-199} and \eqref{eq:RT-206} imply $\eta_\Delta(W_\Delta) \subset 
\pi_{\shuffle_{\bbP}}^{-1}(A_{\bbP}) \cap \pi_{\shuffle_{\check{\bbP}}}^{-1}(E_{\check{\bbP}})$. 
With a standard induction technique for prefix closed local languages it follows
\begin{lemma} \label{lemma:RT-1}\ \\ 
Let $x \in W_\Delta$, then $\eta_\Delta(x) \in 
\pi_{\shuffle_{\bbP}}^{-1}(A_{\bbP}) \cap \pi_{\shuffle_{\check{\bbP}}}^{-1}(E_{\check{\bbP}})$ 
is a shuffled representation of $\varphi_\Delta^{(1)}(x) \in A_\bbP$ by 
$\gamma_\Delta^{(2)}(\varphi_\Delta^{(2)}(x)) \in A_\bbP$ and .
$\gamma_\Delta^{(3)}(\varphi_\Delta^{(3)}(x)) \in E_{\check{\bbP}}$.
\end{lemma}
To show the reverse of Lemma~\ref{lemma:RT-1}, the following observation is helpful:
\begin{lemma} \label{lemma:RT-2}\ \\ 
Let $b',x \in (\shuffle_{\bbP} \dotcup \shuffle_{\check{\bbP}})^*$ and $b = b'x \in 
\pi_{\shuffle_{\bbP}}^{-1}(A_{\bbP}) \cap \pi_{\shuffle_{\check{\bbP}}}^{-1}(E_{\check{\bbP}})$ 
be a shuffled representation of  
$c \in A_\bbP$ by $d = \pi_{\shuffle_{\bbP}}(b) \in A_\bbP$ and 
$e = \pi_{\shuffle_{\check{\bbP}}}(b) \in E_{\check{\bbP}}$, then 
$b'$ is a shuffled representation of $c' \in \pre(c)$ with $|c'| = |b'|$ 
by $d' \in \pre(d)$ and $e' \in \pre(e)$ with $|d'| = |\pi_{\shuffle_{\bbP}}(b')|$ 
and $|e'| = |\pi_{\shuffle_{\check{\bbP}}}(b')|$.
\end{lemma}
Using Lemma~\ref{lemma:RT-2} with $|x| = 1$, standard induction technique shows  
\begin{lemma} \label{lemma:RT-3}\ \\ 
Let $b \in 
\pi_{\shuffle_{\bbP}}^{-1}(A_{\bbP}) \cap \pi_{\shuffle_{\check{\bbP}}}^{-1}(E_{\check{\bbP}})$ 
be a shuffled representation of  
$c \in A_\bbP$ by $d = \pi_{\shuffle_{\bbP}}(b) \in A_\bbP$ and 
$e = \pi_{\shuffle_{\check{\bbP}}}(b) \in E_{\check{\bbP}}$, then 
there exists $x \in W_\Delta$ such that $b = \eta_\Delta(x)$, 
$c = \varphi_\Delta^{(1)}(x)$, $d = \gamma_\Delta^{(2)}(\varphi_\Delta^{(2)}(x))$ 
and $e = \gamma_\Delta^{(3)}(\varphi_\Delta^{(3)}(x))$.
\end{lemma}
Lemma~\ref{lemma:RT-1} and Lemma~\ref{lemma:RT-3} imply 
that for each  $c \in A_\bbP \cap \Delta^*$ it holds,
\begin{equation*}
d \in \mcR'_\bbP(\{c\}) \mbox{ iff there exists } x \in W_\Delta  \mbox{ with } 
c = \varphi_\Delta^{(1)}(x) \mbox{ and } d = \gamma_\Delta^{(2)}(\varphi_\Delta^{(2)}(x)).
\end{equation*}  
Now the results of this section can be summarized:
\begin{definition}\label{def:mu-nu} 
\ \\
Let the alphabetic homomorphisms $\mu_\Delta : \Delta^{()*} \to \Delta^*$ 
and $\nu_\Delta : \Delta^{()*} \to \shuffle_\bbP^*$ be defined by 
\begin{align*}
& \mu_\Delta(x) := \varphi_\Delta^{(1)}(x) \in \Delta^{(1)*} = \Delta^* \mbox{ and } 
\nu_\Delta(x) := \gamma_\Delta^{(2)}(\varphi_\Delta^{(2)}(x)) \in \Delta^{(2)'*} \subset \shuffle_\bbP^* \\
& \mbox{for } x \in \Delta^{()*} \subset \Delta^{()'*}. 
\end{align*}
\end{definition}
\begin{theorem}[{Representation Theorem}]\label{thm:RT} \ \\
Let $\Delta \subset \shuffle_\bbP$, then $\mcR'_\bbP(B) = \nu_\Delta(\mu_\Delta^{-1}(B) \cap W_\Delta)$ 
for each $B \subset A_\bbP \cap \Delta^*$.
\ \\
Additionally $W_\Delta$ is regular, if $\Delta$ is finite.
\end{theorem}
\begin{example}\label{ex:RT}
\ \\
Theorem~\ref{thm:RT} can be applied to Example~\ref{ex:not-i-s-compatible} 
to prove $\SP(\pre(\mathring{P}),\mathring{V})$. For that purpose a finite subset 
$\Delta \subset \shuffle_{\mathring{\bbP}}$ has to be found such that 
$\alpha_{\mathring{\bbP}}^{-1}(\mathring{V}) \subset A_{\mathring{\bbP}} \cap \Delta^*$. 
This can be achieved considering the product automaton of $\mathring{\bbV}$ and 
${\mathring{\bbP}_\shuffle}$, if this automaton is finite. 
Reachability analysis for this product construction 
result in the product automaton of Fig.~\ref{fig:productVP}.
\begin{figure}[h]
\centering
\begin{tikzpicture}[->,>=stealth',shorten >=1pt,auto,node distance=2cm,semithick,initial text=]
   \node[state,initial,rectangle,rounded corners=8pt,inner sep=4pt] (k1)
   {\small $(1,0)$};
   \node[state,rectangle,rounded corners=8pt,inner sep=4pt] (k2) at (0,-2)
   {\small $(2,1_{\mathrm{II}})$};
   \node[state,rectangle,rounded corners=8pt,inner sep=4pt] (k3) at (4,0)
   {\small $(3,1_{\mathrm{II}})$};
   \node[state,rectangle,rounded corners=8pt,inner sep=4pt] (k4) at (4,-2)
   {\small $(4,2_{\mathrm{II}})$};
   \path 
         (k1) edge  node  {\small $(0,a,1_{\mathrm{II}})$} (k2)
         (k1) edge  node [swap] {\small $(0,b,1_{\mathrm{II}})$} (k3)
         (k2) edge  node  {\small $(1_{\mathrm{II}},b,2_{\mathrm{II}})$} (k4)
         (k2) edge  [bend left] node  {\small $(1_{\mathrm{II}},c,0)$} (k1)
         (k3) edge  [bend right] node [swap] {\small $(1_{\mathrm{II}},c,0)$} (k1)
         (k3) edge  [bend left] node  {\small $(1_{\mathrm{II}},b,2_{\mathrm{II}})$} (k4)
         (k4) edge  node  {\small $(2_{\mathrm{II}},c,1_{\mathrm{II}})$} (k3)
        ;
\end{tikzpicture}
\caption{Product automaton of $\mathring{\bbV}$ and ${\mathring{\bbP}_\shuffle}$}\label{fig:productVP}
\end{figure}
Fig.~\ref{fig:productVP} shows that 
\begin{equation}\label{eq:ex-RT-1}
\alpha_{\mathring{\bbP}}^{-1}(\mathring{V}) \subset A_{\mathring{\bbP}} \cap 
\{(0,a,1_{\mathrm{II}}),(0,b,1_{\mathrm{II}}),(1_{\mathrm{II}},c,0),
(1_{\mathrm{II}},b,2_{\mathrm{II}}),(2_{\mathrm{II}},c,1_{\mathrm{II}})\}^*.
\end{equation}
For this example let therefore \[\Delta := \Delta^{(1)} := \{(0,a,1_{\mathrm{II}}),(0,b,1_{\mathrm{II}}),(1_{\mathrm{II}},c,0),
(1_{\mathrm{II}},b,2_{\mathrm{II}}),(2_{\mathrm{II}},c,1_{\mathrm{II}})\}.\] 
This implies
\[S_\Delta^{(1)} = \{0,1_{\mathrm{II}},2_{\mathrm{II}}\}, 
\ S_\Delta^{(3)} = \{0,1_{\mathrm{II}}\}, 
\ S_\Delta^{(2)} = \{0,1_{\mathrm{II}},2_{\mathrm{II}}\}, 
\ \Sigma_\Delta = \{a,b,c\},\]
\[\Delta^{(2)'} = \{(0,a,1_{\mathrm{II}}),(0,b,1_{\mathrm{II}}),(1_{\mathrm{II}},c,0),
(1_{\mathrm{II}},a,2_{\mathrm{II}}),(1_{\mathrm{II}},b,2_{\mathrm{II}}),
(2_{\mathrm{II}},c,1_{\mathrm{II}})\}, \mbox{ and }\]
\[\Delta^{(3)'} = \{(0,a,1_{\mathrm{II}}),(0,b,1_{\mathrm{II}}),(1_{\mathrm{II}},c,0)\}.\]
Now $\Delta^{()}$ is given by \eqref{eq:RT-209}. To illustrate the three tracks, we use a 
column notation to represent the elements of \
\ \\
\ \\ 
\noindent
$\Delta^{()} = \{
\left[ \begin{array}{c}
(0,a,1_\II)\\
(0,a,1_\II)\\
0
\end{array} \right]
%
\ ,\ 
\left[ \begin{array}{c}
(0,a,1_\II)\\
0\\
(0,\check{a},1_\II)
\end{array} \right]
%
\ ,\ 
\left[ \begin{array}{c}
(0,b,1_\II)\\
(0,b,1_\II)\\
0
\end{array} \right]
%
\ ,\ 
\left[ \begin{array}{c}
(0,b,1_\II)\\
0\\
(0,\check{b},1_\II)
\end{array} \right]
%
\ ,\ 
\left[ \begin{array}{c}
(1_\II,c,0)\\
(1_\II,c,0)\\
0
\end{array} \right]
%
\ ,\\
\ \\
\ \\ 
\left[ \begin{array}{c}
(1_\II,c,0)\\
0\\
(1_\II,\check{c},0)
\end{array} \right]
%
\ ,\ 
\left[ \begin{array}{c}
(1_\II,b,2_\II)\\
(1_\II,b,2_\II)\\
0
\end{array} \right]
%
\ ,\ 
\left[ \begin{array}{c}
(1_\II,b,2_\II)\\
(0,b,1_\II)\\
1_\II
\end{array} \right]
%
\ ,\ 
\left[ \begin{array}{c}
(1_\II,b,2_\II)\\
1_\II\\
(0,\check{b},1_\II)
\end{array} \right]
%
\ ,\ 
\left[ \begin{array}{c}
(2_\II,c,1_\II)\\
(2_\II,c,1_\II)\\
0
\end{array} \right]
%
\ ,\\
\ \\
\ \\ 
\left[ \begin{array}{c}
(2_\II,c,1_\II)\\
(1_\II,c,0)\\
1_\II
\end{array} \right]
%
\ ,\ 
\left[ \begin{array}{c}
(2_\II,c,1_\II)\\
1_\II\\
(1_\II,\check{c},0)
\end{array} \right]
%
\ , \ 
\left[ \begin{array}{c}
(0,a,1_\II)\\
(0,a,1_\II)\\
\check{0}
\end{array} \right]
%
\ ,\ 
\left[ \begin{array}{c}
(0,b,1_\II)\\
(0,b,1_\II)\\
\check{0}
\end{array} \right]
%
\ ,\ 
\left[ \begin{array}{c}
(1_\II,c,0)\\
(1_\II,c,0)\\
\check{0}
\end{array} \right]
%
\ ,\\
\ \\
\ \\
\left[ \begin{array}{c}
(1_\II,b,2_\II)\\
(1_\II,b,2_\II)\\
\check{0}
\end{array} \right]
%
\ ,\ 
\left[ \begin{array}{c}
(2_\II,c,1_\II)\\
(2_\II,c,1_\II)\\
\check{0}
\end{array} \right]
%
\}.$ \\
\ \\
\ \\
The definition of $W_\Delta$ can be translated into a semiautomaton 
\[\bbW_\Delta := 
(\Delta^{()} \ , \ S_\Delta^{()} \ , \ \Lambda_\Delta \ , \ (0,0,0)) \]
recognizing $W_\Delta$, where 
$S_\Delta^{()} := S_\Delta^{(1)} \times S_\Delta^{(2)} \times (S_\Delta^{(3)} \dotcup \{ \check{0} \})$. 
Its state transition relation
\[\Lambda_\Delta \subset 
S_\Delta^{()} \times \Delta^{()} \times S_\Delta^{()}\]
can be constructed step by step in compliance with the restrictions of Definition~\ref{def:RT-210}. 
For its representation we use a column notation for the states just as for the elements
of $\Delta^{()}$. So we get 
\ \\
\ \\
\noindent
$\Lambda_\Delta = \{
\left[ \begin{array}{c}
0\\
0\\
0
\end{array} \right]
\left[ \begin{array}{c}
(0,a,1_\II)\\
(0,a,1_\II)\\
0
\end{array} \right]
\left[ \begin{array}{c}
1_\II\\
1_\II\\
0
\end{array} \right]
\ ,\
%
\left[ \begin{array}{c}
0\\
0\\
0
\end{array} \right]\ 
\left[ \begin{array}{c}
(0,a,1_\II)\\
0\\
(0,\check{a},1_\II)
\end{array} \right]
\left[ \begin{array}{c}
1_\II\\
0\\
1_\II
\end{array} \right]
\ ,\
%
\left[ \begin{array}{c}
0\\
0\\
0
\end{array} \right]
\left[ \begin{array}{c}
(0,b,1_\II)\\
(0,b,1_\II)\\
0
\end{array} \right]
\left[ \begin{array}{c}
1_\II\\
1_\II\\
0
\end{array} \right]
\ ,\\
%
\ \\
\ \\
\left[ \begin{array}{c}
0\\
0\\
0
\end{array} \right]\ 
\left[ \begin{array}{c}
(0,b,1_\II)\\
0\\
(0,\check{b},1_\II)
\end{array} \right]
\left[ \begin{array}{c}
1_\II\\
0\\
1_\II
\end{array} \right]
\ ,\
%
\left[ \begin{array}{c}
1_\II\\
1_\II\\
0
\end{array} \right] 
\left[ \begin{array}{c}
(1_\II,c,0)\\
(1_\II,c,0)\\
0
\end{array} \right]
\left[ \begin{array}{c}
0\\
0\\
0
\end{array} \right]
\ ,\
%
\left[ \begin{array}{c}
1_\II\\
1_\II\\
0
\end{array} \right]
\left[ \begin{array}{c}
(1_\II,b,2_\II)\\
(1_\II,b,2_\II)\\
0
\end{array} \right]
\left[ \begin{array}{c}
2_\II\\
2_\II\\
0
\end{array} \right]
\ ,\\
%
\ \\
\ \\
\left[ \begin{array}{c}
1_\II\\
1_\II\\
0
\end{array} \right]
\left[ \begin{array}{c}
(1_\II,b,2_\II)\\
1_\II\\
(0,\check{b},1_\II)
\end{array} \right]
\left[ \begin{array}{c}
2_\II\\
1_\II\\
1_\II
\end{array} \right]
\ ,\
%
\left[ \begin{array}{c}
1_\II\\
0\\
1_\II
\end{array} \right]\ 
\left[ \begin{array}{c}
(1_\II,c,0)\\
0\\
(1_\II,\check{c},0)
\end{array} \right]
\left[ \begin{array}{c}
0\\
0\\
\check{0}
\end{array} \right]
\ ,\
%
\left[ \begin{array}{c}
1_\II\\
0\\
1_\II
\end{array} \right]
\left[ \begin{array}{c}
(1_\II,b,2_\II)\\
(0,b,1_\II)\\
1_\II
\end{array} \right]
\left[ \begin{array}{c}
2_\II\\
1_\II\\
1_\II
\end{array} \right]
\ ,\\
%
\ \\
\ \\
\left[ \begin{array}{c}
2_\II\\
2_\II\\
0
\end{array} \right] 
\left[ \begin{array}{c}
(2_\II,c,1_\II)\\
(2_\II,c,1_\II)\\
0
\end{array} \right]
\left[ \begin{array}{c}
1_\II\\
1_\II\\
0
\end{array} \right]
\ ,\
%
\left[ \begin{array}{c}
2_\II\\
1_\II\\
1_\II
\end{array} \right] 
\left[ \begin{array}{c}
(2_\II,c,1_\II)\\
(1_\II,c,0)\\
1_\II
\end{array} \right]
\left[ \begin{array}{c}
1_\II\\
0\\
1_\II
\end{array} \right]
\ ,\
%
\left[ \begin{array}{c}
2_\II\\
1_\II\\
1_\II
\end{array} \right]
\left[ \begin{array}{c}
(2_\II,c,1_\II)\\
1_\II\\
(1_\II,\check{c},0)
\end{array} \right]
\left[ \begin{array}{c}
1_\II\\
1_\II\\
\check{0}
\end{array} \right]
\ ,\\
%
\ \\
\ \\
\left[ \begin{array}{c}
0\\
0\\
\check{0}
\end{array} \right]\ 
\left[ \begin{array}{c}
(0,a,1_\II)\\
(0,a,1_\II)\\
\check{0}
\end{array} \right]
\left[ \begin{array}{c}
1_\II\\
1_\II\\
\check{0}
\end{array} \right]
\ ,\
%
\left[ \begin{array}{c}
0\\
0\\
\check{0}
\end{array} \right]\ 
\left[ \begin{array}{c}
(0,b,1_\II)\\
(0,b,1_\II)\\
\check{0}
\end{array} \right]
\left[ \begin{array}{c}
1_\II\\
1_\II\\
\check{0}
\end{array} \right]
\ ,\
%
\left[ \begin{array}{c}
1_\II\\
1_\II\\
\check{0}
\end{array} \right]\ 
\left[ \begin{array}{c}
(1_\II,c,0)\\
(1_\II,c,0)\\
\check{0}
\end{array} \right]
\left[ \begin{array}{c}
0\\
0\\
\check{0}
\end{array} \right]
\ ,\\
%
\ \\
\ \\
\left[ \begin{array}{c}
1_\II\\
1_\II\\
\check{0}
\end{array} \right]\ 
\left[ \begin{array}{c}
(1_\II,b,2_\II)\\
(1_\II,b,2_\II)\\
\check{0}
\end{array} \right]
\left[ \begin{array}{c}
2_\II\\
2_\II\\
\check{0}
\end{array} \right]
\ ,\
%
\left[ \begin{array}{c}
2_\II\\
2_\II\\
\check{0}
\end{array} \right]\ 
\left[ \begin{array}{c}
(2_\II,c,1_\II)\\
(2_\II,c,1_\II)\\
\check{0}
\end{array} \right]
\left[ \begin{array}{c}
1_\II\\
1_\II\\
\check{0}
\end{array} \right]
%
\}.$ \\
\ \\
\ \\
Applying standard automata algorithms \cite{berstel79} to this semiautomaton, shows 
$\nu_\Delta(W_\Delta) \subset \alpha_{\mathring{\bbP}}^{-1}(\mathring{V})$, 
which by Theorem~\ref{thm:RT} and \eqref{eq:ex-RT-1} implies
\begin{equation}\label{eq:ex-RT-2}
\mcR'_{\mathring{\bbP}}(\alpha_{\mathring{\bbP}}^{-1}(\mathring{V})) = 
\nu_\Delta(\mu_\Delta^{-1}(\alpha_{\mathring{\bbP}}^{-1}(\mathring{V})) \cap W_\Delta) 
\subset \alpha_{\mathring{\bbP}}^{-1}(\mathring{V}).
\end{equation}
Now \eqref{eq:ex-RT-2} together with Corollary~\ref{cor:eq-rel-RP'} proves 
$\SP(\pre(\mathring{P}),\mathring{V})$.
\end{example}
Using Corollary~\ref{cor:eq-rel-RP'} and Theorem~\ref{thm:RT}, Example~\ref{ex:RT} 
demonstrates how to decide $\SP(\pre(P),V)$, if there exists a finite subset 
$\Delta \subset \shuffle_\bbP$, such that $\alpha_\bbP^{-1}(V) \subset \Delta^*$. 
Since we assume $\emptyset \neq P \subset \Sigma^*$ and $\delta(q_0,\pre(P))=Q$, $\pre(P)$ is
recognized by the automaton $\grave{\bbP} := (\Sigma,Q,\delta,q_0,Q)$.
So using Corollary~\ref{cor:eq-rel-R0P'} instead of Corollary~\ref{cor:eq-rel-RP'}, 
we also can decide $\SP(\pre(P),V)$, if there exists a finite subset 
$\grave{\Delta} \subset \shuffle_{\grave{\bbP}}$ such that 
$(\alpha_{\grave{\bbP}}^{-1}(V) \cap Z_{\grave{\bbP}}^{-1}(0)) \subset \grave{\Delta}^*$. \\
\ \\
Now the question arises: Is there any relation between $\Delta$ and $\grave{\Delta}$? 
The only difference between $\bbP$ and $\grave{\bbP}$ is the set of their final states: 
$F \subset Q$ versus $Q$. Therefore Definition~\ref{def:shuffle-automaton} implies
\begin{align*}
\shuffle_{\grave{\bbP}} = & \shuffle_\bbP  
\ \cup \ \{(f,a,f) \in \dsN_0^Q \times \Sigma \times \dsN_0^Q \ | \ \delta(q_0,a) \mbox{ is defined} \} \ \cup  \\
& \{(f,a,f-1_q) \in \dsN_0^Q \times \Sigma \times \dsN_0^Q \ | 
\ f \geqslant 1_q \mbox{ and } \delta(q,a) \mbox{ is defined} \}.
\end{align*}
Now on account of
\begin{align*}
& \{(f,a,f-1_q+1_{\delta(q,a)}) \in \dsN_0^Q \times \Sigma \times \dsN_0^Q \ | 
\ f \geqslant 1_q \mbox{ and } \delta(q,a) \mbox{ is defined} \} \ \cup \\ 
&\{(f,a,f+1_{\delta(q_0,a)}) \in \dsN_0^Q \times \Sigma \times \dsN_0^Q \ | \ \delta(q_0,a) \mbox{ is defined} \} 
\subset \shuffle_\bbP,
\end{align*}
\begin{equation}\label{eq:grave-P-1}
\mbox{for each } (f,a,g) \in \shuffle_{\grave{\bbP}} 
\mbox{ there exists } (f,a,g') \in \shuffle_\bbP 
\mbox{ such that } g' \geqslant g . 
\end{equation}
\eqref{eq:def-cp^3} implies
\begin{equation}\label{eq:grave-P-2}
(f+h,a,g'+h) \in \shuffle_\bbP \mbox{ for each } 
(f,a,g') \in \shuffle_\bbP \mbox{ and } h \in \dsN_0^Q. 
\end{equation}
By \eqref{eq:grave-P-1} and \eqref{eq:grave-P-2} an 
induction proof shows:
\begin{align*}
& \mbox{For each } x \in A_{\grave{\bbP}} 
\mbox{ there exists } y \in A_\bbP \mbox{ with } 
\alpha_{\grave{\bbP}}(x) = \alpha_\bbP(y)
\mbox{ and } \\ 
& Z_\bbP(y') \geqslant Z_{\grave{\bbP}}(x') 
\mbox{ for each } x' \in \pre(x) \mbox{ and } y' \in \pre(y) 
\mbox{ with } |x'| = |y'|,
\end{align*}
which implies
\begin{align}\label{eq:grave-P-3}
& \mbox{For each } x \in \alpha_{\grave{\bbP}}^{-1}(V) 
\mbox{ there exists } y \in \alpha_\bbP^{-1}(V) \mbox{ with } 
\alpha_{\grave{\bbP}}(x) = \alpha_\bbP(y)
\mbox{ and } \nonumber \\ 
& Z_\bbP(y') \geqslant Z_{\grave{\bbP}}(x') 
\mbox{ for each } x' \in \pre(x) \mbox{ and } y' \in \pre(y) 
\mbox{ with } |x'| = |y'|.
\end{align}
Let now $\Delta \subset \shuffle_\bbP \cap (T(Q) \times \Sigma \times T(Q))$ 
such that $\alpha_\bbP^{-1}(V) \subset \Delta^*$, and let $\Sigma_\Delta$ be 
defined as in \eqref{eq:RT-194'}. Let 
\begin{align}\label{eq:grave-P-4}
& S_\Delta := Z_\bbP(\pre(\alpha_\bbP^{-1}(V))) \mbox{ and } \nonumber \\ 
& \grave{S}_\Delta := \{f \in \dsN_0^Q |\mbox{ there exists }  
g \in S_\Delta \mbox{ with } g \geqslant f \}. 
\end{align}
Then finiteness of $\Delta$ implies finiteness of $\Sigma_\Delta$, 
$S_\Delta$ and $\grave{S}_\Delta$, and by \eqref{eq:grave-P-3} 
holds $\alpha_{\grave{\bbP}}^{-1}(V) \subset \grave{\Delta}^*$ with 
$\grave{\Delta} := \shuffle_{\grave{\bbP}} \cap (\grave{S}_\Delta \times \Sigma_\Delta \times \grave{S}_\Delta)$. 
This implies:
\begin{align}\label{eq:grave-P-5}
& \mbox{If } \alpha_\bbP^{-1}(V) \subset \Delta^* \mbox{ for a finite subset }
  \Delta \subset \shuffle_\bbP, \mbox{ then there exists }  \nonumber \\ 
& \mbox{a finite subset } \grave{\Delta} \subset \shuffle_{\grave{\bbP}}  
  \mbox{ with } (\alpha_{\grave{\bbP}}^{-1}(V) \cap Z_{\grave{\bbP}}^{-1}(0)) \subset \grave{\Delta}^*. 
\end{align}
The following example shows that the converse of \eqref{eq:grave-P-5} does not hold.
\begin{example}\label{ex:grave-P}
\ \\
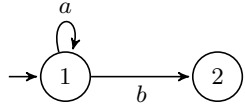
\begin{figure}[h]
\centering
\begin{tikzpicture}[->,>=stealth',shorten >=1pt,auto,node distance=2cm,semithick,initial text=]
   \node[state,initial,inner sep=4pt] (k1)
   {\small $1$};
   \node[state,inner sep=4pt] (k2) at (2,0)
   {\small $2$};
   \path
         (k1) edge [loop above]  node  {\small $a$} (k1)
         (k1) edge  node [swap] {\small $b$} (k2)
        ;
\end{tikzpicture}
\caption{Semiautomaton $\bar{\bbV}$ recognizing $\bar{V}$}\label{fig:bar-V}
\end{figure}
Let $\bar{\bbP}$ and $\bar{P}$ as defined in Figure~\ref{fig:easy-compatible-1}, and
let $\bar{\bbV}$ and $\bar{V}$ as defined in Figure~\ref{fig:bar-V}.
Then $Z_{\bar{\bbP}}(\alpha_{\bar{\bbP}}^{-1}(\bar{V})) = \{0\} \cup \{n_{\mathrm{II}} | n\in \dsN \}$. 
Therefore each $\Delta \subset \shuffle_{\bar{\bbP}}$ with 
$\alpha_{\bar{\bbP}}^{-1}(\bar{V}) \subset \Delta^*$ is an infinite set. \\
\ \\
But $(\alpha_{\grave{\bar{\bbP}}}^{-1}(\bar{V}) \cap Z_{\grave{\bar{\bbP}}}^{-1}(0)) \subset 
\{ (0,a,0) , (0,a,1_\II) , (1_\II,a,1_\II) , (1_\II,b,0) \}^*$, because of 
$Z_{\grave{\bar{\bbP}}}(\alpha_{\grave{\bar{\bbP}}}^{-1}(\bar{V})) = 
\{0\} \cup \{n_{\mathrm{II}} | n\in \dsN \}$, and 
$c^{-1}(\alpha_{\grave{\bar{\bbP}}}^{-1}(\bar{V}) \cap Z_{\grave{\bar{\bbP}}}^{-1}(0)) = \emptyset$ 
for each $c \in \alpha_{\grave{\bar{\bbP}}}^{-1}(\bar{V})$ with 
$Z_{\grave{\bar{\bbP}}}(c) \geqslant 2_\II$. 
\end{example}
\section{Decidability Questions}
%
%
%
%




In Section~\ref{sec:representation theorem} it was demonstrated by way of an example,
how a finite set 
$\Delta \subset \shuffle_\bbP$ can be found
that fulfils the condition $\alpha^{-1}_\bbP(\pre(V)) \subset \Delta^*$.
Given that $P$ and $V$ are regular languages the approach is now considered in
general and it is shown how the existence of such a finite set can be decided. \\
\ \\
For an arbitrary alphabet $\Gamma$ let the mapping $\alf : 2^{\Gamma^*} \to 2^\Gamma$ 
be defined by $\Gamma(\emptyset) := \Gamma(\{ \varepsilon \}) := \emptyset$, 
$\Gamma(\{ wa \}) := \Gamma(\{ w \}) \cup \{ a \}$ for $w \in \Gamma^*$ and $a \in \Gamma$, 
and $\Gamma(L) := \bigcup\limits_{w\in L}\Gamma(\{ w \})$.
Then the minimal set $\Delta$ with the above property is $\alf(\alpha^{-1}_\bbP(\pre(V)))$.
So, the problem is to find $\alf(\alpha^{-1}_\bbP(\pre(V)))$ and to prove that
$\alf(\alpha^{-1}_\bbP(\pre(V)))$ is finite. In \eqref{eq:grave-P-5} it is shown that the 
more general problem is to investigate $\alf(\alpha^{-1}_\bbP(V) \cap Z^{-1}_\bbP(0))$. 
But we first examine the problem concerning $\alf(\alpha^{-1}_\bbP(\pre(V)))$, because there 
is a much easier decision procedure than for the general problem. \\
\ \\
Let now $\emptyset\neq P\subset\Sigma^*$, $\emptyset\neq V\subset\Sigma^*$, $\bbP=(\Sigma,Q,\delta,q_0,F)$
a deterministic automaton that recognizes $P$ with $\delta(q_0,\pre(P))=Q$,
and $\bbV=(\Sigma,Q_\bbV,\delta_\bbV,q_{\bbV_0})$
a deterministic semiautomaton that recognizes $\pre(V)$ with $Q \cap Q_\bbV = \emptyset$. Then 
$\delta_\bbV(q_{\bbV_0},\alpha_\bbP(x))$ is defined for each $x\in\alpha^{-1}_\bbP(\pre(V))$.  
\begin{align}\label{eq:A1}
\mbox{The set } 
x^{-1}(\alpha^{-1}_\bbP(\pre(V))) \cap \shuffle_\bbP
\mbox{ is finite for each } x\in\alpha^{-1}_\bbP(\pre(V)) \nonumber \\
\mbox{ and depends only on }
(Z_\bbP(x),\delta_\bbV(q_{\bbV_0},\alpha_\bbP(x))) .
\end{align}
\begin{align}\label{eq:A2}
\mbox{For each }
y\in x^{-1}(\alpha^{-1}_\bbP(\pre(V))) \cap \shuffle_\bbP
\mbox{ is } (Z_\bbP(xy),\delta_\bbV(q_{\bbV_0},\alpha_\bbP(xy))) \nonumber \\
\mbox{ uniquely determined by }
(Z_\bbP(x),\delta_\bbV(q_{\bbV_0},\alpha_\bbP(x))) \mbox{ and } y . 
\end{align}
Let $Q_{\bbP\bbV}:=\{(Z_\bbP(x),\delta_\bbV(q_{\bbV_0},\alpha_\bbP(x))) | x\in \alpha^{-1}_\bbP(\pre(V))\}$.
Then $Q_{\bbP\bbV}$ can be considered as the state set of a deterministic semiautomaton 
$\bbS_{\bbP\bbV}$ that recognizes
$\alpha^{-1}_\bbP(\pre(V))$. 
Its initial state is $(0,q_{\bbV_0})$, its alphabet is $\shuffle_\bbP$, and its
state transition function is given by \eqref{eq:A2}. More precisely:
\begin{align}\label{eq:A2.1}
& \bbS_{\bbP\bbV} = (\shuffle_\bbP,Q_{\bbP\bbV},\delta_{\bbP\bbV},(0,q_{\bbV_0})) 
\mbox{ where } \delta_{\bbP\bbV} : Q_{\bbP\bbV} \times \shuffle_\bbP \to Q_{\bbP\bbV} \nonumber \\ 
& \mbox{is a partial function with } \nonumber \\  
& \delta_{\bbP\bbV}((Z_\bbP(x),\delta_\bbV(q_{\bbV_0},\alpha_\bbP(x))),y) := 
(Z_\bbP(xy),\delta_\bbV(q_{\bbV_0},\alpha_\bbP(xy))) \nonumber \\
& \mbox{for } x\in\alpha^{-1}_\bbP(\pre(V)) \mbox{ and }
y\in x^{-1}(\alpha^{-1}_\bbP(\pre(V))) \cap \shuffle_\bbP. 
\end{align}
In example~\ref{ex:RT}, $\bbS_{\bbP\bbV}$ corresponds to the product automaton of Figure~\ref{fig:productVP}. \\
\ \\
Let now $Z_{\bbP\bbV}:\alpha^{-1}_\bbP(\pre(V))\to Q_{\bbP\bbV}$ with
\begin{equation}\label{eq:A2.2}
Z_{\bbP\bbV}(x) := (Z_\bbP(x),\delta_\bbV(q_{\bbV_0},\alpha_\bbP(x))) 
\mbox{ for each } x\in\alpha^{-1}_\bbP(\pre(V)).
\end{equation}
Then 
\begin{align}\label{eq:A2.3}
& Q_{\bbP\bbV}=Z_{\bbP\bbV}(\alpha^{-1}_\bbP(\pre(V))) \mbox{ and } \nonumber \\
& Z_{\bbP\bbV}(x) = \delta_{\bbP\bbV}((0,q_{\bbV_0}),x) 
  \mbox{ for each } x\in\alpha^{-1}_\bbP(\pre(V)). 
\end{align}
\begin{align}\label{eq:A3}
\mbox{For each }
n\in\dsN_0 \mbox{ let } 
A^{(n)}_\bbP:=\{w\in A_\bbP | \abs{w} \leq n\}
\mbox{ and } \nonumber \\
Q^{(n)}_{\bbP\bbV}:= Z_{\bbP\bbV}(\alpha^{-1}_\bbP(\pre(V))\cap A^{(n)}_\bbP) . 
\end{align}
\begin{align}\label{eq:A4}
\mbox{From \eqref{eq:A1} follows that }
\alpha^{-1}_\bbP(\pre(V))\cap A^{(n)}_\bbP
\mbox{ and thus } \nonumber \\
Q^{(n)}_{\bbP\bbV}
\mbox{ for each } 
n\in \dsN_0
\mbox{ are finite sets} . 
\end{align}
If $Q^{(k)}_{\bbP\bbV}=Q^{(k+1)}_{\bbP\bbV}$ for a $k\in \dsN_0$, 
then follows from \eqref{eq:A1} and \eqref{eq:A2}
$Q^{(i)}_{\bbP\bbV}=Q^{(k)}_{\bbP\bbV}$ and
\begin{equation}\label{eq:A5}
\alf(\alpha^{-1}_\bbP(\pre(V))\cap A^{(i+1)}_\bbP)=
\alf(\alpha^{-1}_\bbP(\pre(V))\cap A^{(k+1)}_\bbP) 
\mbox{ for each } i\geq k.
\end{equation}
Because $A_\bbP=\bigcup\limits_{n\in\dsN_0}A^{(n)}_\bbP$ and
$A^{(n)}_\bbP\subset A^{(n+1)}_\bbP$ for each $n\in \dsN_0$ holds
\begin{equation}\label{eq:A6}
Q_{\bbP\bbV} = \bigcup\limits_{n\in\dsN_0}Q_{\bbP\bbV}^{(n)} \mbox{ and }
Q_{\bbP\bbV}^{(n)}\subset Q_{\bbP\bbV}^{(n+1)} \mbox{ for each } n\in \dsN_0. 
\end{equation}

From \eqref{eq:A5}-\eqref{eq:A6} follows
\begin{align}\label{eq:A7}
&\alf(\alpha^{-1}_\bbP(\pre(V)))=\alf(\alpha^{-1}_\bbP(\pre(V))\cap A^{(k+1)}_\bbP),
\mbox{ as well as } 
Q_{\bbP\bbV}= Q_{\bbP\bbV}^{(k)} \mbox{ if } \nonumber \\
& Q_{\bbP\bbV}^{(k)}= Q_{\bbP\bbV}^{(k+1)}, \mbox{ and } 
\alf(\alpha^{-1}_\bbP(\pre(V))) \mbox{ and } Q_{\bbP\bbV}
\mbox{ are finite sets } 
\end{align}
because of \eqref{eq:A4}.\\
\ \\
Because $\alpha^{-1}_\bbP(\pre(V))$ is prefix closed 
\[\alf(\alpha^{-1}_\bbP(\pre(V))) \subset Z_\bbP(\alpha^{-1}_\bbP(\pre(V)))\times \Sigma
\times Z_\bbP(\alpha^{-1}_\bbP(\pre(V))), \mbox{ and}\]
$Z_\bbP(\alpha^{-1}_\bbP(\pre(V)))\subset p_3(\alf(\alpha^{-1}_\bbP(\pre(V))))\cup\{0\}$,
where $p_3((f,a,g)):=g$ for $(f,a,g)\in \shuffle_\bbP$.
Because $\Sigma$ is finite, it follows
\begin{equation}\label{eq:A8}
\alf(\alpha^{-1}_\bbP(\pre(V))) \mbox{ is finite iff } 
Z_\bbP(\alpha^{-1}_\bbP(\pre(V))) \mbox{ is finite.} 
\end{equation}
Accordingly, from the finiteness of $Q_\bbV$ follows
\begin{equation}\label{eq:A9}
Z_\bbP(\alpha^{-1}_\bbP(\pre(V))) \mbox{ is finite iff } 
Q_{\bbP\bbV} \mbox{ is finite.} 
\end{equation}
If $Q_{\bbP\bbV}$
is finite, then because of \eqref{eq:A6} 
\begin{equation}\label{eq:A10}
\mbox{it exists a } k\in\dsN_0
\mbox{ with } Q_{\bbP\bbV}^{(i)}= Q_{\bbP\bbV}^{(k)}
\mbox{ for all } i\geq k . 
\end{equation}
Because of \eqref{eq:A7}-\eqref{eq:A10} the stepwise computation of
\begin{align}\label{eq:A11}
&Q_{\bbP\bbV}^{(i)} \mbox{ for } i \in \dsN_0 \mbox{ and the test }
Q_{\bbP\bbV}^{(i)}=Q_{\bbP\bbV}^{(i+1)} \nonumber \\
&\mbox{provides a semi-algorithm for the finiteness of }
\alf(\alpha^{-1}_\bbP(\pre(V))). 
\end{align}
In case of a positive result, 
$\alf(\alpha^{-1}_\bbP(\pre(V)))$
can be computed using \eqref{eq:A7}.
\ \\
\ \\
In preparation for the decision on finiteness of $Q_{\bbP\bbV}$ we need a closer look 
on the structure of $\shuffle_\bbP$.
By Definition~\ref{def:shuffle-automaton^}, Definition~\ref{def:shuffle-automaton} and 
\eqref{eq:def-wedge-P} it holds
\begin{align}\label{eq:A11.1}
\shuffle_\bbP =& \wedge_\bbP(\hat{\shuffle}_\bbP) = \wedge_\bbP(\tilde{\shuffle}_\bbP) \ \cup 
\ \wedge_\bbP(\mathring{\shuffle}_\bbP) \ \cup 
\ \wedge_\bbP(\bar{\shuffle}_\bbP) \ \cup \ \wedge_\bbP(\tilde{\bar{\shuffle}}_\bbP) \mbox{ and } \nonumber \\ 
\wedge_\bbP(\tilde{\shuffle}_\bbP) =& \{(f,a,f+1_p) \in \dsN_0^Q \times \Sigma \times \dsN_0^Q \ | 
\ \delta(q_0,a)=p \mbox{ and it exists } b\in \Sigma \mbox{ such } \nonumber \\
&\mbox{ that } \delta(p,b) \mbox{ is defined} \},\nonumber \\
\wedge_\bbP(\mathring{\shuffle}_\bbP) =& \{(f,a,f+1_p-1_q) \in \dsN_0^Q \times \Sigma \times \dsN_0^Q \ | 
\ f \geqslant 1_q, \delta(q,a)=p \mbox{ and it exists } \nonumber \\
&b\in \Sigma \mbox{ such that } \delta(p,b) \mbox{ is defined} \},\nonumber \\
\wedge_\bbP(\bar{\shuffle}_\bbP) =& \{(f,a,f-1_q) \in \dsN_0^Q \times \Sigma \times \dsN_0^Q \ | 
\ f \geqslant 1_q \mbox{ and } \delta(q,a) \in F \} \mbox{ and }\nonumber \\
\wedge_\bbP(\tilde{\bar{\shuffle}}_\bbP) =& \{(f,a,f) \in \dsN_0^Q \times \Sigma \times \dsN_0^Q \ |\ \delta(q_0,a) \in F \}.
\end{align}
On account of \eqref{eq:grave-P-2} a proper subset $\shuffle_\bbP^\sigma \subset \shuffle_\bbP$ 
together with $\dsN_0^Q$ suffices 
to completely characterize $\shuffle_\bbP$. Let therefore
\begin{align}\label{eq:A11.2}
\shuffle_\bbP^\sigma :=& \tilde{\shuffle}_\bbP^\sigma \ \cup 
\ \mathring{\shuffle}_\bbP^\sigma \ \cup 
\ \bar{\shuffle}_\bbP^\sigma \ \cup \ \tilde{\bar{\shuffle}}_\bbP^\sigma \mbox{ with }\nonumber \\ 
\tilde{\shuffle}_\bbP^\sigma :=& \{(0,a,1_p) \in \dsN_0^Q \times \Sigma \times \dsN_0^Q \ |
 \ \delta(q_0,a)=p \mbox{ and it exists } b\in \Sigma \mbox{ such } \nonumber \\
&\mbox{ that } \delta(p,b) \mbox{ is defined} \},\nonumber \\
\mathring{\shuffle}_\bbP^\sigma :=& \{(1_q,a,1_p) \in \dsN_0^Q \times \Sigma \times \dsN_0^Q \ | 
\ \delta(q,a)=p \mbox{ and it exists } \nonumber \\
&b\in \Sigma \mbox{ such that } \delta(p,b) \mbox{ is defined} \},\nonumber \\
\bar{\shuffle}_\bbP^\sigma :=& \{(1_q,a,0) \in \dsN_0^Q \times \Sigma \times \dsN_0^Q \ | \ \delta(q,a) \in F \} \mbox{ and }\nonumber \\
\tilde{\bar{\shuffle}}_\bbP^\sigma :=& \{(0,a,0) \in \dsN_0^Q \times \Sigma \times \dsN_0^Q \ |\ \delta(q_0,a) \in F \}.
\end{align}
Then by \eqref{eq:grave-P-2}
\begin{equation}\label{eq:A11.3}
\shuffle_\bbP = \{(f+h,a,g+h) \in \dsN_0^Q \times \Sigma \times \dsN_0^Q \ |\ (f,a,g) \in\shuffle_\bbP^\sigma 
\mbox{ and } h \in \dsN_0^Q \}.
\end{equation}
The following should be noticed:
\begin{align}\label{eq:A11.4}
& \shuffle_\bbP^\sigma = \tilde{\shuffle}_\bbP^\sigma \ \dotcup 
\ \mathring{\shuffle}_\bbP^\sigma \ \dotcup 
\ \bar{\shuffle}_\bbP^\sigma \ \dotcup \ \tilde{\bar{\shuffle}}_\bbP^\sigma. \nonumber \\ 
& \mbox{Generally, for } (f',a,g') \in \shuffle_\bbP \mbox{ the representation } 
 (f',a,g') = (f+h,a,g+h) \nonumber \\ 
& \mbox{with } (f,a,g) \in\shuffle_\bbP^\sigma 
\mbox{ and } h \in \dsN_0^Q \mbox{ is not unique.} \nonumber \\
& \shuffle_\bbP^\sigma \mbox{ is finite for finite automata } \bbP.
\end{align}
Let the mapping $\sigma_\bbP : \shuffle_\bbP \to 2^{\shuffle_\bbP^\sigma} \setminus \{ \emptyset \}$ 
be defined by
\begin{equation}\label{eq:A11.5}
\sigma_\bbP ((f',a,g')) := \{(f,a,g) \in\shuffle_\bbP^\sigma \ |\ (f',a,g') = (f+h,a,g+h) 
\mbox{ with } h \in \dsN_0^Q \} 
\end{equation}
for $(f',a,g') \in \shuffle_\bbP$.
\ \\
\ \\
For the decision on finiteness of $Q_{\bbP\bbV}$ we now utilize
results from Petri nets \cite{reutenauer90}, \cite{Wimmel08}.
A Petri net $N=(S,T,K)$ consists of a finite set $S$ of \emph{places},
a finite set $T$ of \emph{transitions}, and a set $K\subset (S\times T) \cup (T\times S)$
of \emph{edges}.
A \emph{marking} of such a Petri net is a mapping $M:S\to \dsN_0$. 
Dynamic behavior of Petri nets is formalized in terms of \emph{occurrence steps} and
\emph{occurrence sequences}. The set $\Omega$ of occurrence steps is defined by 
\begin{align}\label{eq:A11.6}
& \Omega := \{(M,t,M') \in \dsN_0^S \times T \times \dsN_0^S \ |\ 
M \geqslant \sum\limits_{x \in S , (x,t) \in K} 1_x \mbox{ and } \nonumber \\
& M' = M - \sum\limits_{x \in S , (x,t) \in K} 1_x + 
\sum\limits_{y \in S , (t,y) \in K} 1_y \}.
\end{align}
The set $\mcO$ of occurrence steps with $\mcO \subset \Omega^+$ and the functions 
$\mcI : \mcO \to \dsN_0^S$ and $\mcF : \mcO \to \dsN_0^S$ are defined inductively by
\begin{align}\label{eq:A11.7}
& \mbox{For each } o = (M,t,M') \in \Omega \mbox{ let } o \in \mcO ,\  
\mcI(o) := M \mbox{ and } \mcF(o) := M'. \nonumber \\
& \mbox{For each } w \in \mcO \mbox{ and } o \in \Omega \mbox{ with } 
\mcF(w) = \mcI(o) \mbox{ let } \nonumber \\ 
& wo \in \mcO ,\  \mcI(wo) := \mcI(w) \mbox{ and } \mcF(wo) := \mcF(o).
\end{align}
$\mcI(w)$ is called the \emph{initial marking} and $\mcF(w)$ the \emph{final marking} of $w$.
For $M \in \dsN_0^S$ the \emph{reachability set} $\mcE(M)$ is defined by
\begin{equation}\label{eq:A11.7'}
\mcE(M) := \{ M \} \cup \mcF(\mcI^{-1}(M)).
\end{equation}
The semiautomaton $\bbS_{\bbP\bbV}$ can be simulated by a Petri net
$N_{\bbP\bbV}$ such that there exists an injective mapping $\iota$ from 
$Q_{\bbP\bbV}$ into the set of markings of $N_{\bbP\bbV}$ with
\begin{equation}\label{eq:A14}
\iota(Q_{\bbP\bbV})=\mcE(\iota((0,q_{\bbV_0}))). 
\end{equation}
To define $N_{\bbP\bbV}$ let its set of places $S:=Q\dotcup Q_\bbV$.  
Let therefore the injective mapping \[\iota : Q_{\bbP\bbV}\to \dsN_0^{Q\dotcup Q_\bbV}\]
be defined by
\begin{align}\label{eq:A13}
&\iota((f,q))(x):=f(x) \mbox{ for } x\in Q, \nonumber \\
&\iota((f,q))(x):=0 \mbox{ for } x\in Q_\bbV \setminus \{ q \} \mbox{ and } \nonumber \\
&\iota((f,q))(x):=1 \mbox{ for } x\in Q_\bbV \cap \{ q \}, \nonumber \\
&\mbox{for each } (f,q)\in Q_{\bbP\bbV}\subset \dsN_0^Q\times Q_\bbV. 
\end{align} 
The set $T$ of transitions of $N_{\bbP\bbV}$ will be defined such that there exists 
a bijective mapping $\chi : \shuffle_\bbP^\sigma \times Q_\bbV \to T$. 
For this purpose let 
$T:=\tilde{T} \ \dotcup \ \mathring{T} \ \dotcup \ \bar{T} \ \dotcup \ \tilde{\bar{T}}$,
where
\begin{align}\label{eq:A13.1}
\tilde{T} :=& \{(r,a,(p,s)) \in Q_\bbV \times {\Sigma} \times (Q\times Q_\bbV) \ | 
\ (0,a,1_p) \in \tilde{\shuffle}_\bbP^\sigma \mbox{ and } \delta_\bbV(r,a)=s\}, \nonumber \\
\mathring{T} :=& \{((q,r),a,(p,s)) \in (Q\times Q_\bbV) \times {\Sigma} \times (Q\times Q_\bbV) \ | 
\ (1_q,a,1_p) \in \mathring{\shuffle}_\bbP^\sigma \mbox{ and } \delta_\bbV(r,a)=s\}, \nonumber \\
\bar{T} :=& \{((q,r),a,s) \in (Q\times Q_\bbV) \times {\Sigma} \times Q_\bbV \ | 
\ (1_q,a,0) \in \bar{\shuffle}_\bbP^\sigma \mbox{ and } \delta_\bbV(r,a)=s\}, \mbox{ and }\nonumber \\
\tilde{\bar{T}} :=& \{(r,a,s) \in Q_\bbV \times {\Sigma} \times Q_\bbV \ |
\ (0,a,0) \in \tilde{\bar{\shuffle}}_\bbP^\sigma \mbox{ and } \delta_\bbV(r,a)=s\}.
\end{align}
Now let the bijective mapping $\chi : \shuffle_\bbP^\sigma \times Q_\bbV \to T$ 
be defined by
\begin{align}\label{eq:A13.2}
\chi(((0,a,1_p),r)) :=& (r,a,(p,\delta_\bbV(r,a))) \mbox{ for } 
((0,a,1_p),r) \in \tilde{\shuffle}_\bbP^\sigma \times Q_\bbV, \nonumber \\
\chi(((1_q,a,1_p),r)) :=& ((q,r),a,(p,\delta_\bbV(r,a))) \mbox{ for }
\ ((1_q,a,1_p),r) \in \mathring{\shuffle}_\bbP^\sigma \times Q_\bbV, \nonumber \\
\chi(((1_q,a,0),r)) :=& ((q,r),a,\delta_\bbV(r,a)) \mbox{ for }
\ ((1_q,a,0),r) \in \bar{\shuffle}_\bbP^\sigma \times Q_\bbV, \mbox{ and }\nonumber \\
\chi(((0,a,0),r)) :=& (r,a,\delta_\bbV(r,a)) \mbox{ for }
\ ((0,a,0),r) \in \tilde{\bar{\shuffle}}_\bbP^\sigma \times Q_\bbV.
\end{align}
The set K of edges of $N_{\bbP\bbV}$ let be defined by 
$K:=\tilde{K} \ \dotcup \ \mathring{K} \ \dotcup \ \bar{K} \ \dotcup \ \tilde{\bar{K}}$,
where 
\begin{align}\label{eq:A12}
\tilde{K} :=& \bigcup\limits_{(r,a,(p,s)) \in \tilde{T}}
\{(r,(r,a,(p,s))),((r,a,(p,s)),p),((r,a,(p,s)),s)\} \nonumber \\
& \ \ \ \ \ \ \ \ \ \ \ \ \subset 
(Q_\bbV\times \tilde{T})\cup(\tilde{T}\times(Q\cup Q_\bbV)), \nonumber \\
\mathring{K} :=& \bigcup\limits_{((q,r),a,(p,s))  \in \mathring{T}} 
\{(q,((q,r),a,(p,s))),(r,((q,r),a,(p,s))),(((q,r),a,(p,s)),p),\nonumber \\
& \ \ \ \ \ \ \ \ \ \ \ \ (((q,r),a,(p,s)),s)\} \subset
((Q\cup Q_\bbV)\times \mathring{T})\cup(\mathring{T}\times(Q\cup Q_\bbV)), \nonumber \\
\bar{K} :=& \bigcup\limits_{((q,r),a,s) \in \bar{T}}
\{(q,((q,r),a,s)),(r,((q,r),a,s)),(((q,r),a,s),s)\} \nonumber \\
& \ \ \ \ \ \ \ \ \ \ \ \ \subset
((Q\cup Q_\bbV)\times \bar{T})\cup(\bar{T}\times Q_\bbV), \mbox{ and}\nonumber \\
\tilde{\bar{K}} :=& \bigcup\limits_{(r,a,s) \in \tilde{\bar{T}}}
\{(r,(r,a,s)),((r,a,s),s)\} \subset
(Q_\bbV\times \tilde{\bar{T}})\cup(\tilde{\bar{T}}\times Q_\bbV). 
\end{align}
With these definitions of $N_{\bbP\bbV}$, $\iota$ and $\chi$ the following 
can be shown by induction: 
\begin{align}\label{eq:A12.1}
& \mbox{For each }o = o_1...o_{|o|} \in (\dsN_0^{Q\dotcup Q_\bbV} \times T \times \dsN_0^{Q\dotcup Q_\bbV})^+ \nonumber \\
& \mbox{with } o_i \in \dsN_0^{Q\dotcup Q_\bbV} \times T \times \dsN_0^{Q\dotcup Q_\bbV} 
\mbox{ for } 1 \leq i \leq |o| \mbox{ holds } o \in \mcI^{-1}(\iota((0,q_{\bbV_0}))),\nonumber \\
& \mbox{iff there exists } x \in \alpha^{-1}_\bbP(\pre(V)) 
\mbox{ with } |x| = |o| \mbox{ such that for } 1\leq i \leq \abs{o} \mbox{ holds:}\nonumber \\
& o_i=(\iota(Z_{\bbP\bbV}(x'_{i-1})),t_i,\iota(Z_{\bbP\bbV}(x'_i))) \mbox{ with } 
x'_j\in \pre(x) \mbox{ and } \abs{x'_j}=j \mbox{ for } 0\leq j \leq \abs{o} \nonumber \\
& \mbox{and } \nonumber \\
& t_i \in \chi((y_i,\delta_\bbV(q_{\bbV 0}, \alpha_\bbP(x'_{i-1})))) \mbox{ with } y_i \in \sigma_\bbP(x_i), \nonumber \\
& \mbox{where } x=x_1\ldots{}x_{\abs{o}} \mbox{ and } x_i\in \shuffle_\bbP \mbox{ for } 1\leq i \leq \abs{o}.
\end{align}
This proves \eqref{eq:A14}.
Because $\iota$ is injective, $Q_{\bbP\bbV}$ is finite iff $\mcE(\iota((0,q_{\bbV_0})))$ is finite.
The finiteness of $\mcE(M)$ is decidable for each each Petri net and each marking $M$ of the net 
\cite{reutenauer90} \cite{Wimmel08}. 
Therefore, with \eqref{eq:A11} and \eqref{eq:A7}, the following theorem holds.
\begin{theorem}\label{thm:A15}
If $\bbP$ and $\bbV$ are finite automata, then it is decidable if 
$\alf(\alpha^{-1}_\bbP(\pre(V)))$ is finite.
In the positive case $\alf(\alpha^{-1}_\bbP(\pre(V)))$ is computable.
\end{theorem}
The key to decide finiteness of $\mcE(M)$ is Dickson's lemma \cite{lothaire83}, \cite{reutenauer90}. Therefore 
Theorem~\ref{thm:A15} can also be proven by directly applying Dickson's lemma. We used 
the simulation by a Petri net, because we also need this simulation to tackle the more 
general problem to decide the finiteness of $\alf(Z_\bbP^{-1}(0)\cap \alpha^{-1}_\bbP(V))$.
\ \\
\ \\
For this we make the same assumptions as in the respective problem regarding
$\alf(\alpha^{-1}_\bbP(\pre(V)))$ and additionally postulate the existence of
$F_\bbV\subset Q_\bbV$ with
\begin{align}\label{eq:A16}
V=\{w\in \pre(V) | \delta_\bbV(q_{\bbV_0},w)\in F_\bbV\}. 
\end{align}
Let therefore,
\begin{align}\label{eq:A17}
\breve{A}_{\bbP\bbV} :=& \pre(Z_\bbP^{-1}(0)\cap \alpha^{-1}_\bbP(V))= \nonumber \\
&\{u\in \alpha^{-1}_\bbP(\pre(V))|  
Z_{\bbP\bbV}(u(u^{-1}(\alpha^{-1}_\bbP(\pre(V)))))\cap \{0\}\times F_\bbV \neq \emptyset\} . 
\end{align}
From \eqref{eq:A12.1} it follows:
\begin{align}\label{eq:A14'}
&\mbox{For each } u\in \alpha^{-1}_\bbP(\pre(V)) \mbox{ holds }\nonumber \\
&\iota(Z_{\bbP\bbV}(u(u^{-1}(\alpha^{-1}_\bbP(\pre(V))))))=\mcE(\iota(Z_{\bbP\bbV}(u))).
\end{align}
For each Petri net and each two markings $M$ and $M'$ it is decidable if
$M'\in \mcE(M)$ \cite{reutenauer90}, \cite{Wimmel08}.
From this it follows on account of \eqref{eq:A14'}:
\begin{align}\label{eq:A18}
\mbox{For each } u\in \alpha^{-1}_\bbP(\pre(V)) \mbox{ it is decidable, if } u\in \breve{A}_{\bbP\bbV}. 
\end{align}
On account of \eqref{eq:A17}:
\begin{align}\label{eq:A19}
\alf(Z_\bbP^{-1}(0)\cap \alpha^{-1}_\bbP(V))=\alf(\breve{A}_{\bbP\bbV}) .
\end{align}
\begin{align}\label{eq:A20}
\mbox{Let now } \breve{Q}_{\bbP\bbV}:=Z_{\bbP\bbV}(\breve{A}_{\bbP\bbV})\subset Q_{\bbP\bbV} .
\end{align}
Analog to \eqref{eq:A8} and \eqref{eq:A9},
\begin{align}\label{eq:A21}
\alf(\breve{A}_{\bbP\bbV}) \mbox{ is finite if }
\breve{Q}_{\bbP\bbV} \mbox{ is finite.}
\end{align}
\begin{align}\label{eq:A22}
\mbox{For each } n\in \dsN_0 \mbox{ let } 
\breve{Q}_{\bbP\bbV}^{(n)}:=\breve{Q}_{\bbP\bbV}\cap {Q}_{\bbP\bbV}^{(n)}. 
\end{align}
Therewith,
\begin{align}\label{eq:A23}
\breve{Q}_{\bbP\bbV}^{(n)} \mbox{ are finite sets that are computable on account of \eqref{eq:A18}.} 
\end{align}
As in  \eqref{eq:A6} - \eqref{eq:A10}, 
\begin{align}\label{eq:A24}
&\mbox{the stepwise computation of each } \breve{Q}_{\bbP\bbV}^{(i)}
\mbox{ for } i\in\dsN_0  \mbox{ and the test } \breve{Q}_{\bbP\bbV}^{(i)}=\breve{Q}_{\bbP\bbV}^{(i+1)}\nonumber \\
&\mbox{provides a semi-algorithm to decide the finiteness of } \alf(\breve{A}_{\bbP\bbV}) . 
\end{align}
Let $k\in\dsN_0$ be the smallest $i\in\dsN_0$ such that 
$\breve{Q}_{\bbP\bbV}^{(i)}=\breve{Q}_{\bbP\bbV}^{(i+1)}$, then 
\begin{align}\label{eq:A25}
 \alf(\breve{A}_{\bbP\bbV})= \alf(\breve{A}_{\bbP\bbV}\cap A_\bbP^{(i+1)}). 
\end{align}
With regard to \eqref{eq:A21} it remains to prove the decidability of finiteness of $\breve{Q}_{\bbP\bbV}$.
This can be done using the following result for Petri nets:
\begin{align}\label{eq:A26}
&\mbox{Let } M \mbox{ and } M' \mbox{ markings of a Petri net, then it is decidable if } \nonumber \\
&\{\mathring{M}\in \mcE(M) | M'\in \mcE(\mathring{M})\} \mbox{ is finite \cite{Wimmel08}.} 
\end{align}
On account of \eqref{eq:A14'}, \eqref{eq:A17}, and \eqref{eq:A20}:
\begin{align}\label{eq:A27}
&\breve{Q}_{\bbP\bbV} \mbox{ is finite iff for each } q\in F_\bbV \nonumber \\
&\iota(Z_{\bbP\bbV}(\{u\in \alpha^{-1}_\bbP(\pre(V)) | 
(0,q)\in Z_{\bbP\bbV}(u(u^{-1}(\alpha^{-1}_\bbP(\pre(V)))))\})
\mbox{ is finite .} 
\end{align}
On account of \eqref{eq:A14'} furthermore holds:
\begin{align}\label{eq:A28}
&\iota(Z_{\bbP\bbV}(\{u\in \alpha^{-1}_\bbP(\pre(V)) | 
(0,q)\in Z_{\bbP\bbV}(u(u^{-1}(\alpha^{-1}_\bbP(\pre(V)))))\}) \nonumber \\
&= \{x\in \mcE(\iota((0,q_{\bbV_0}))) | \iota(0,q)\in \mcE(x)\}. 
\end{align}
Now \eqref{eq:A26} - \eqref{eq:A28} prove the following theorem:
\begin{theorem}\label{thm:A29}
If $\bbP$ and $\bbV$ are finite automata, then it is decidable if 
$\alf(Z_\bbP^{-1}(0)\cap \alpha^{-1}_\bbP(V))$ is finite.
In the positive case $\alf(Z_\bbP^{-1}(0)\cap \alpha^{-1}_\bbP(V))$ is computable by \eqref{eq:A25}.
\end{theorem}
Now, combining the technique of Section~\ref{sec:representation theorem} with the simulation 
of S-automata by Petri nets will result in a proof of the decidability of $\SP(P \cup \{\varepsilon\},V)$ 
for regular $P$ and $V$. The idea is, to consider the counterexamples for 
\[\mcR'_\bbP(\alpha_\bbP^{-1}(V) \cap Z_\bbP^{-1}(0)) \subset \alpha_\bbP^{-1}(V).\]  
Preliminarily we notice that on account of \eqref{subeq:shuff-rep-d}
\begin{equation}\label{eq:decidability-1}
Z_\bbP(d) \leq Z_\bbP(c) \mbox{ for each } d \in \mcR'_\bbP(\{ c \}).
\end{equation}
By Corollary~\ref{cor:eq-rel-R0P'} $\SP(P \cup \{\varepsilon\},V)$ does not hold, iff 
there exists $c \in \alpha_\bbP^{-1}(V) \cap Z_\bbP^{-1}(0)$ and 
$d \in \mcR'_\bbP(\{ c \})$ with $d \notin \alpha_\bbP^{-1}(V)$. 
With Theorem~\ref{thm:RT} this is equivalent to: 
\begin{align}\label{eq:decidability-2}
& \mbox{There exists } x \in W_{\shuffle_\bbP} \mbox{ with } \nonumber \\ 
& \mu_{\shuffle_\bbP}(x) \in Z_\bbP^{-1}(0) \cap \alpha_\bbP^{-1}(V) 
  \mbox{ and } \nu_{\shuffle_\bbP}(x) \notin \alpha_\bbP^{-1}(V).
\end{align}
As $\nu_{\shuffle_\bbP}(x) \in \mcR'_\bbP(\{ \mu_{\shuffle_\bbP}(x) \})$ 
by \eqref{eq:decidability-1} \eqref{eq:decidability-2} is equivalent to
\begin{align}\label{eq:decidability-3}
& \mbox{There exists } x \in W_{\shuffle_\bbP} \mbox{ with } \nonumber \\ 
& \mu_{\shuffle_\bbP}(x) \in Z_\bbP^{-1}(0) \cap \alpha_\bbP^{-1}(V) 
  \mbox{ and } \nu_{\shuffle_\bbP}(x) \in Z_\bbP^{-1}(0) \setminus \alpha_\bbP^{-1}(V).
\end{align}
The first step to decide $\SP(P \cup \{\varepsilon\},V)$ is to ``isomorphically refine'' 
the prefix closed language $W_{\shuffle_\bbP}$ by appropriately attaching states of 
$\bbV$ to the elements of $\shuffle_\bbP^{()}$. 
For that purpose, additionally to the assumptions about $\bbP$ and $\bbV$, we assume 
that $\bbV$ is complete. This means $\delta_\bbV : Q_\bbV \times \Sigma \to Q_\bbV$ is a total 
function, and it poses no restriction on $V$, as using an additional dummy state each deterministic 
automaton can be transformed into a complete deterministic automaton recognizing the same language. \\
\ \\
According to $\shuffle_\bbP^{()}$, $\varphi_{\shuffle_\bbP}^{(1.2)}$ 
and $\varphi_{\shuffle_\bbP}^{(2.2)}$, as defined in 
\eqref{eq:RT-209}, \eqref{eq:RT-186'} and \eqref{eq:RT-206'''}, let now 
\begin{align}\label{eq:decidability-4}
\Psi_\bbP^\bbV := \{& ((y_1,y_2),x,(y'_1,y'_2)) \in 
  (Q_\bbV \times Q_\bbV) \times \shuffle_\bbP^{()} \times (Q_\bbV \times Q_\bbV) |  \nonumber \\
& y'_1 = \delta_\bbV(y_1,\varphi_{\shuffle_\bbP}^{(1.2)}(x_1)), \nonumber \\
& y'_2 = \delta_\bbV(y_2,\varphi_{\shuffle_\bbP}^{(2.2)}(x_2)) 
  \mbox{ if } x_2 \in \shuffle_\bbP^{(2)'} \mbox{ and } \nonumber \\  
& y'_2 = y_2 \mbox{ if } x_2 \in S_{\shuffle_\bbP}^{(2)}, \mbox{ with } \nonumber \\
& x = (x_1,x_2,x_3) \},
\end{align}
let the mappings $Z_\bbP^{\bbV(1)} : \Psi_\bbP^{\bbV*} \to Q_\bbV$ and 
$Z_\bbP^{\bbV(2)} : \Psi_\bbP^{\bbV*} \to Q_\bbV$ be defined by
\begin{equation}\label{eq:decidability-5}
Z_\bbP^{\bbV(1)}(\varepsilon) := Z_\bbP^{\bbV(2)}(\varepsilon) := q_{\bbV_0}, \ 
  Z_\bbP^{\bbV(1)}(uv) := y'_1 \mbox{ and } Z_\bbP^{\bbV(2)}(uv) := y'_2
\end{equation}
for  $u \in \Psi_\bbP^{\bbV*}$ and $v \in \Psi_\bbP^\bbV$ with $v = ((y_1,y_2),x,(y'_1,y'_2))$, \\
\ \\
and let the homomorphism $\psi_\bbP^\bbV : \Psi_\bbP^{\bbV*} \to \shuffle_\bbP^{()*}$ 
be defined by
\begin{equation}\label{eq:decidability-6}
\psi_\bbP^\bbV((y_1,y_2),x,(y'_1,y'_2))) := x \mbox{ for } ((y_1,y_2),x,(y'_1,y'_2)) \in \Psi_\bbP^\bbV.
\end{equation}
\begin{definition}\label{def:WPV} \ \\
Let the prefix closed language $W_\bbP^\bbV \subset \Psi_\bbP^{\bbV*}$ be defined by
\begin{align*}
W_\bbP^\bbV := \{& w \in (\psi_\bbP^\bbV)^{-1}(W_{\shuffle_\bbP}) | 
  \ y_1 = Z_\bbP^{\bbV(1)}(u) \mbox{ and } y_2 = Z_\bbP^{\bbV(2)}(u) \mbox{ for each }  \\ 
& uv \in \pre(w) \mbox{ with } u \in \Psi_\bbP^{\bbV*} \mbox{ and } 
  v = ((y_1,y_2),x,(y'_1,y'_2))\in \Psi_\bbP^\bbV \}.
\end{align*}
\end{definition}
Now Definition~\ref{def:mu-nu} implies
\begin{align}\label{eq:decidability-7}
& Z_\bbP^{\bbV(1)}(u) = 
   \delta_\bbV(q_{\bbV_0},\alpha_\bbP(\mu_{\shuffle_\bbP}(\psi_\bbP^\bbV(u)))) \mbox{ and } \nonumber \\
& Z_\bbP^{\bbV(2)}(u) = \delta_\bbV(q_{\bbV_0},\alpha_\bbP(\nu_{\shuffle_\bbP}(\psi_\bbP^\bbV(u)))) 
   \mbox{ for each } u \in W_\bbP^\bbV.
\end{align}
\begin{equation}\label{eq:decidability-8}
\mbox{As } \bbV \mbox{ is a complete deterministic automaton } (\psi_\bbP^\bbV)_{|W_\bbP^\bbV} \mbox{ is a bijection.} 
\end{equation}
On account of \eqref{eq:decidability-8}, \eqref{eq:decidability-3} is equivalent to: 
\begin{align*}
& \mbox{There exists } u \in W_\bbP^\bbV \mbox{ with } \nonumber \\ 
& \mu_{\shuffle_\bbP}(\psi_\bbP^\bbV(u)) \in Z_\bbP^{-1}(0) \cap \alpha_\bbP^{-1}(V) 
  \mbox{ and } \nu_{\shuffle_\bbP}(\psi_\bbP^\bbV(u)) \in Z_\bbP^{-1}(0) \setminus \alpha_\bbP^{-1}(V),
\end{align*}
which by \eqref{eq:decidability-7} can be equivalently restated in terms of reachable states:
\begin{align}\label{eq:decidability-9}
& \mbox{There exists } u \in W_\bbP^\bbV \mbox{ with } Z_\bbP^{\bbV(1)}(u) \in F_\bbV,
  \ Z_\bbP^{\bbV(2)}(u) \in Q_\bbV \setminus F_\bbV \mbox{ and } \nonumber \\ 
& Z_\bbP(\mu_{\shuffle_\bbP}(\psi_\bbP^\bbV(u))) = 0 = Z_\bbP(\nu_{\shuffle_\bbP}(\psi_\bbP^\bbV(u))).
\end{align}
Caused by this formulation, the second step to decide $\SP(P \cup \{\varepsilon\},V)$ is to construct 
a deterministic semiautomaton $\bbW_\bbP^\bbV$ recognizing $W_\bbP^\bbV$. Generally $\bbW_\bbP^\bbV$ will be infinite.  
It is an immediate consequence of 
Definition~\ref{def:WPV} that
\begin{align}\label{eq:decidability-10}
W_\bbP^\bbV = &(\psi_\bbP^\bbV)^{-1}(W_{\shuffle_\bbP}) \setminus 
  (X_\bbP^\bbV \Psi_\bbP^{\bbV*} \cup \Psi_\bbP^{\bbV*} Y_\bbP^\bbV \Psi_\bbP^{\bbV*}) \mbox{ with } \nonumber \\
X_\bbP^\bbV = &\{((y_1,y_2),x,(y'_1,y'_2))\in \Psi_\bbP^\bbV | y_1 \neq q_{\bbV_0} 
  \mbox{ or } y_2 \neq q_{\bbV_0} \} \mbox{ and } \nonumber \\
Y_\bbP^\bbV = &\{((y_1,y_2),x,(y'_1,y'_2))((\bar{y}_1,\bar{y}_2),\bar{x},(\bar{y}'_1,\bar{y}'_2)) 
\in \Psi_\bbP^\bbV \Psi_\bbP^\bbV   | \nonumber \\ 
&  \ y'_1 \neq \bar{y}_1 \mbox{ or } y'_2 \neq \bar{y}_2 \}.
\end{align}
Let now 
$\bbW_{\shuffle_\bbP} = 
(\shuffle_\bbP^{()} , S_{\shuffle_\bbP}^{()} , \lambda_{\shuffle_\bbP} , s_0)$
be a deterministic semiautomaton recognizing $W_{\shuffle_\bbP}$, where 
$\lambda_{\shuffle_\bbP} : 
S_{\shuffle_\bbP}^{()} \times \shuffle_\bbP^{()} \to S_{\shuffle_\bbP}^{()}$ 
is a partial function and $s_0 \in S_{\shuffle_\bbP}^{()}$. 
Generally $\bbW_{\shuffle_\bbP}$ is infinite.
\eqref{eq:decidability-10} implies that the following deterministic semiautomaton 
$\bbW_\bbP^\bbV$ recognizes $W_\bbP^\bbV$:	
\begin{align}\label{eq:decidability-11}
& \bbW_\bbP^\bbV := (\Psi_\bbP^\bbV, S_\bbP^\bbV, \lambda_\bbP^\bbV, q_{\bbP_0}^\bbV) \mbox{ where } \nonumber \\
& S_\bbP^\bbV := Q_\bbV \times Q_\bbV \times S_{\shuffle_\bbP}^{()}, \ 
q_{\bbP_0}^\bbV := (q_{\bbV_0}, q_{\bbV_0}, s_0), \mbox{ and } \nonumber \\
& \lambda_\bbP^\bbV : S_\bbP^\bbV \times \Psi_\bbP^\bbV \to S_\bbP^\bbV 
\mbox{ is a partial function with } \nonumber \\
& \lambda_\bbP^\bbV((y_1,y_2,s),a) := (y'_1,y'_2,s'), \mbox{ for }
  (y_1,y_2,s) \in Q_\bbV \times Q_\bbV \times S_{\shuffle_\bbP}^{()}, \nonumber \\
& a = ((y_1,y_2),x,(y'_1,y'_2)) \in \Psi_\bbP^\bbV \mbox{ and } \nonumber \\
& \lambda_{\shuffle_\bbP}(s,x) = s'. 
\end{align}
By \eqref{eq:decidability-11} and Definition~\ref{def:WPV} holds 
\begin{equation}\label{eq:decidability-12}
\lambda_\bbP^\bbV(q_{\bbP_0}^\bbV,u) = 
(Z_\bbP^{\bbV(1)}(u),Z_\bbP^{\bbV(2)}(u),\lambda_{\shuffle_\bbP}(s_0,\psi_\bbP^\bbV(u))) 
\mbox{ for each } u \in W_\bbP^\bbV.
\end{equation}
To completely define $\bbW_\bbP^\bbV$, a complete definition of $\bbW_{\shuffle_\bbP}$ must be given. 
For that purpose we need the mapping $\check{\varphi}_{\shuffle_\bbP}^{(3.1)} : \shuffle_\bbP^{(3)} 
\to S_{\shuffle_\bbP}^{(3)} \dotcup \{\check{0}\}$ defined by
\begin{align}\label{eq:decidability-13}
& \check{\varphi}_{\shuffle_\bbP}^{(3.1)}((f,a,g)) := f 
\mbox{ for } (f,a,g) \in \shuffle_\bbP^{(3)'}, \nonumber \\ 
& \check{\varphi}_{\shuffle_\bbP}^{(3.1)}(f) := f \mbox{ for } f \in S_{\shuffle_\bbP}^{(3)} \mbox{ and } \nonumber \\
& \check{\varphi}_{\shuffle_\bbP}^{(3.1)}(\check{0}) := \check{0}.
\end{align}
As in Example~\ref{ex:RT}, \eqref{eq:RT-209} and Definition~\ref{def:RT-210} can be directly translated 
into a deterministic semiautomaton $\bbW_{\shuffle_\bbP}$. Let therefore 
\begin{align}\label{eq:decidability-14}
& S_{\shuffle_\bbP}^{()} := 
  S_{\shuffle_\bbP}^{(1)} \times S_{\shuffle_\bbP}^{(2)} \times (S_{\shuffle_\bbP}^{(3)} \dotcup \{ \check{0} \}), 
  \ s_0 := (0,0,0), \mbox{ and let } \nonumber \\
& \lambda_{\shuffle_\bbP}((q_1,q_2,q_3),(x_1,x_2,x_3)) \mbox{ be defined for } \nonumber \\
& (q_1,q_2,q_3) \in S_{\shuffle_\bbP}^{()} \mbox{ and } 
  (x_1,x_2,x_3) \in \shuffle_\bbP^{()} \mbox{ with } \nonumber \\
& (q_1,q_2,q_3) = 
  (\varphi_{\shuffle_\bbP}^{(1.1)}(x_1),\varphi_{\shuffle_\bbP}^{(2.1)}(x_2),\check{\varphi}_{\shuffle_\bbP}^{(3.1)}(x_3)), 
  \mbox{ where } \nonumber \\
& \lambda_{\shuffle_\bbP}((q_1,q_2,q_3),(x_1,x_2,x_3)) := 
  (\varphi_{\shuffle_\bbP}^{(1.3)}(x_1),\varphi_{\shuffle_\bbP}^{(2.3)}(x_2),\check{0}) \mbox{ for } \nonumber \\
& x_3 \in (\shuffle_\bbP^{(3)'} \cap (\varphi_{\shuffle_\bbP}^{(3.3)})^{-1}(0))\cup \{ \check{0} \} \mbox{ and } \nonumber \\
& \lambda_{\shuffle_\bbP}((q_1,q_2,q_3),(x_1,x_2,x_3)) := 
  (\varphi_{\shuffle_\bbP}^{(1.3)}(x_1),\varphi_{\shuffle_\bbP}^{(2.3)}(x_2),
   \varphi_{\shuffle_\bbP}^{(3.3)}(x_3)) \mbox{ for } \nonumber \\
& x_3 \in \shuffle_\bbP^{(3)} \setminus
  ((\shuffle_\bbP^{(3)'} \cap (\varphi_{\shuffle_\bbP}^{(3.3)})^{-1}(0))\cup \{ \check{0} \}).
\end{align}
Now by induction it is easy to show that $\bbW_{\shuffle_\bbP}$ recognizes $W_{\shuffle_\bbP}$. With 
\eqref{eq:local-def-A}, \eqref{eq:RT-199''},  Definition~\ref{def:RT-210} and 
Definition~\ref{def:mu-nu}, \eqref{eq:decidability-14} implies 
\begin{align}\label{eq:decidability-15}
& Z_\bbP(\mu_{\shuffle_\bbP}(w)) = q_1 \mbox{ and } Z_\bbP(\nu_{\shuffle_\bbP}(w)) = q_2, \mbox{ with } 
  \lambda_{\shuffle_\bbP}((0,0,0),w) = (q_1,q_2,q_3), \nonumber \\
& \mbox{for each } w \in W_{\shuffle_\bbP}.
\end{align}
By \eqref{eq:RT-208} holds
\begin{align}\label{eq:decidability-16}
& \lambda_{\shuffle_\bbP}((0,0,0),w) \in \{(0,0,0),(0,0,\check{0}) \} 
  \mbox{ for each } w \in W_{\shuffle_\bbP} \mbox{ with } \nonumber \\
& \lambda_{\shuffle_\bbP}((0,0,0),w) \in 
\{ 0 \} \times S_{\shuffle_\bbP}^{(2)} \times (S_{\shuffle_\bbP}^{(3)} \dotcup \{ \check{0} \}).
\end{align}
Now, on account of \eqref{eq:decidability-3}, \eqref{eq:decidability-9}, \eqref{eq:decidability-12}, 
\eqref{eq:decidability-15}  and \eqref{eq:decidability-16}
\begin{align}\label{eq:decidability-17}
& \SP(P \cup \{\varepsilon\},V) \mbox{ iff there don't exist any } u \in W_\bbP^\bbV \mbox{ with } \nonumber \\  
& \lambda_\bbP^\bbV(q_{\bbP_0}^\bbV,u) \in F_\bbV \times (Q_\bbV \setminus F_\bbV) \times \{(0,0,0),(0,0,\check{0})\}.
\end{align}
A more detailed analysis shows that
\begin{align}
& \lambda_\bbP^\bbV(q_{\bbP_0}^\bbV,u) \in F_\bbV \times (Q_\bbV \setminus F_\bbV) \times \{(0,0,0),(0,0,\check{0})\} 
  \mbox{ iff } \nonumber \\  
& \lambda_\bbP^\bbV(q_{\bbP_0}^\bbV,u) \in F_\bbV \times (Q_\bbV \setminus F_\bbV) \times \{(0,0,\check{0})\}. \nonumber
\end{align}
The reachability question posed by \eqref{eq:decidability-17} can be decided by simulating 
$\bbW_\bbP^\bbV$ by a Petri net. Preparative to that simulation, first we need an appropriate characterization of 
$\lambda_\bbP^\bbV$, similar to the characterization of $\shuffle_\bbP$ by $\shuffle_\bbP^\sigma$ together with $\dsN_0^Q$. 
So the third step to decide $\SP(P \cup \{\varepsilon\},V)$ is to present such a characterization. \\
\ \\
By \eqref{eq:decidability-4}, \eqref{eq:decidability-5}, \eqref{eq:decidability-11} and \eqref{eq:decidability-14} 
$\lambda_\bbP^\bbV$ is uniquely determined by $\delta_\bbV$ and $\shuffle_\bbP^{()}$. Therefore we now look for an 
appropriate characterization of $\shuffle_\bbP^{()}$. For that purpose we assume 
\begin{equation}\label{eq:decidability-18}
\alf(P) = \Sigma , 
\end{equation}
which don't cause any restriction for $P$. Now we assemble the different sets 
to define $\shuffle_\bbP^{()}$. On account of \eqref{eq:RT-184}, \eqref{eq:RT-190}, \eqref{eq:RT-191}, 
\eqref{subeq:iota-shuff-2-c} and $E_\bbP \subset A_\bbP$ holds 
\begin{equation}\label{eq:decidability-19}
\shuffle_\bbP^{(1)} = \shuffle_\bbP , \ S_{\shuffle_\bbP}^{(1)} = Z_\bbP(A_\bbP) 
\mbox{ and } S_{\shuffle_\bbP}^{(3)} = Z_\bbP(E_\bbP).
\end{equation}
\begin{align}\label{eq:decidability-20}
& S_{\shuffle_\bbP}^{(3)} \mbox{ is finite and can be effectively determined, if } \nonumber \\
& \bbP \mbox{ is finite, and it holds } 0 \in S_{\shuffle_\bbP}^{(3)}.
\end{align}
On account of \eqref{eq:RT-192}, \eqref{eq:decidability-20}, \eqref{eq:RT-195}, 
\eqref{eq:decidability-18} and \eqref{eq:RT-200} it holds
\begin{align}\label{eq:decidability-21}
& S_{\shuffle_\bbP}^{(2)} = Z_\bbP(A_\bbP), 
 \ \shuffle_\bbP^{(2)'} = \shuffle_\bbP \cap (Z_\bbP(A_\bbP) \times \Sigma \times Z_\bbP(A_\bbP))
\mbox{ and } \nonumber \\
& \shuffle_\bbP^{(3)'} = \shuffle_{\check{\bbP}}^E \cap 
  (Z_\bbP(E_\bbP) \times \check{\iota}^{-1}(\Sigma) \times Z_\bbP(E_\bbP)) = \shuffle_{\check{\bbP}}^\sigma.
\end{align}
\begin{equation}\label{eq:decidability-22}
\mbox{So } \shuffle_\bbP^{(3)'} = \shuffle_{\check{\bbP}}^\sigma 
\mbox{ is finite and can be effectively determined, if } 
\bbP \mbox{ is finite.}
\end{equation}
By \eqref{eq:grave-P-2} \eqref{eq:RT-209} can be rephrased. Let therefore 
the mappings $\varphi_\Delta^{(1.1)'}$, $\varphi_\Delta^{(1.2)'}$ and $\varphi_\Delta^{(1.3)'}$ be defined by
\begin{align}\label{eq:decidability-23}
& \varphi_\Delta^{(1.1)'} : \dsN_0^Q \times \Sigma \times \dsN_0^Q \to \dsN_0^Q \mbox{ with } 
  \varphi_\Delta^{(1.1)'}((f,a,g)) := f, \nonumber \\
& \varphi_\Delta^{(1.2)'} : \dsN_0^Q \times \Sigma \times \dsN_0^Q \to \Sigma \mbox{ with } 
  \varphi_\Delta^{(1.2)'}((f,a,g)) := a, \nonumber \\
& \varphi_\Delta^{(1.3)'} : \dsN_0^Q \times \Sigma \times \dsN_0^Q \to \dsN_0^Q \mbox{ with } 
  \varphi_\Delta^{(1.3)'}((f,a,g)) := g \nonumber \\
& \mbox{for each } (f,a,g) \in \dsN_0^Q \times \Sigma \times \dsN_0^Q.
\end{align}
Then \eqref{eq:RT-209} becomes
\begin{align}\label{eq:decidability-24}
\Delta^{()} = \{\ &(x_1,x_2,x_3) \in (\dsN_0^Q \times \Sigma \times \dsN_0^Q) \times \Delta^{(2)} \times \Delta^{(3)} \  | \nonumber \\
& \varphi_\Delta^{(1.1)'}(x_1) = \varphi_\Delta^{(2.1)}(x_2) + \varphi_\Delta^{(3.1)}(x_3), \nonumber \\
& \varphi_\Delta^{(1.3)'}(x_1) = \varphi_\Delta^{(2.3)}(x_2) + \varphi_\Delta^{(3.3)}(x_3) \mbox{ and } \nonumber \\
& \mbox{either } x_2 \in \Delta^{(2)'}, \ x_3 \in S_\Delta^{(3)} \dotcup \{ \check{0} \} 
  \mbox{ and } \varphi_\Delta^{(1.2)'}(x_1) = \varphi_\Delta^{(2.2)}(x_2) \nonumber \\
& \mbox{or } x_2 \in S_\Delta^{(2)}, \ x_3 \in \Delta^{(3)'} 
  \mbox{ and } \varphi_\Delta^{(1.2)'}(x_1) = \check{\iota}(\varphi_\Delta^{(3.2)}(x_3)) \ \}. 
\end{align}
Now \eqref{eq:decidability-24} together with \eqref{eq:RT-196}, \eqref{eq:RT-200}, 
\eqref{eq:decidability-19} and \eqref{eq:decidability-21} implies
\begin{align}\label{eq:decidability-25}
\shuffle_\bbP^{()} = & \shuffle_\bbP^{(S)} \dotcup \shuffle_\bbP^{(E)} \mbox{ with } \nonumber \\
\shuffle_\bbP^{(S)} := 
\{\ &(x_1,x_2,x_3) \in (\dsN_0^Q \times \Sigma \times \dsN_0^Q) \times \shuffle_\bbP^{(2)'} 
 \times (Z_\bbP(E_\bbP) \dotcup \{\check{0}\}) \  | \nonumber \\
& \varphi_{\shuffle_\bbP}^{(1.1)'}(x_1) = \varphi_{\shuffle_\bbP}^{(2.1)}(x_2) + \varphi_{\shuffle_\bbP}^{(3.1)}(x_3), \nonumber \\
& \varphi_{\shuffle_\bbP}^{(1.3)'}(x_1) = \varphi_{\shuffle_\bbP}^{(2.3)}(x_2) + \varphi_{\shuffle_\bbP}^{(3.3)}(x_3) \mbox{ and } \nonumber \\
& \varphi_{\shuffle_\bbP}^{(1.2)'}(x_1) = \varphi_{\shuffle_\bbP}^{(2.2)}(x_2) \ \} \mbox{ and } \nonumber \\
\shuffle_\bbP^{(E)} := 
\{\ &(x_1,x_2,x_3) \in (\dsN_0^Q \times \Sigma \times \dsN_0^Q) \times Z_\bbP(A_\bbP) \times \shuffle_{\check{\bbP}}^\sigma \  | \nonumber \\
& \varphi_{\shuffle_\bbP}^{(1.1)'}(x_1) = \varphi_{\shuffle_\bbP}^{(2.1)}(x_2) + \varphi_{\shuffle_\bbP}^{(3.1)}(x_3), \nonumber \\
& \varphi_{\shuffle_\bbP}^{(1.3)'}(x_1) = \varphi_{\shuffle_\bbP}^{(2.3)}(x_2) + \varphi_{\shuffle_\bbP}^{(3.3)}(x_3) \mbox{ and } \nonumber \\
& \varphi_{\shuffle_\bbP}^{(1.2)'}(x_1) = \check{\iota}(\varphi_{\shuffle_\bbP}^{(3.2)}(x_3)) \ \}. 
\end{align}
Because of \eqref{eq:decidability-21} and \eqref{eq:A11.3} holds
\begin{align}\label{eq:decidability-26}
& (f,a,g) \in \shuffle_\bbP^{(2)'} \mbox{ iff there exists } h \in \dsN_0^Q  \mbox{ and } 
  (f',a,g') \in \shuffle_\bbP^\sigma \mbox{ with } \nonumber \\ 
& f = f' + h \in Z_\bbP(A_\bbP) \mbox{ and } g = g' + h.
\end{align}
This implies
\begin{align}\label{eq:decidability-28}
\shuffle_\bbP^{(S)} = 
\{\ &(x_1,x_2,x_3) \in (\dsN_0^Q \times \Sigma \times \dsN_0^Q) \times 
 \shuffle_\bbP^{(2)'} \times (Z_\bbP(E_\bbP) \dotcup \{\check{0}\}) \  | \nonumber \\
& \mbox{there exist } (f',a,g') \in \shuffle_\bbP^\sigma \mbox{ and } h \in \dsN_0^Q,  \mbox{ such that } \nonumber \\
& f' + h \in Z_\bbP(A_\bbP),\ x_2 = (f'+h,a,g'+h) \mbox{ and } \nonumber \\
& x_1 = (f' + h + \varphi_{\shuffle_\bbP}^{(3.1)}(x_3) , a ,
  g' + h + \varphi_{\shuffle_\bbP}^{(3.3)}(x_3)) \ \}.
\end{align}
Similar to \eqref{eq:decidability-28} $\shuffle_\bbP^{(E)}$ can be represented by
\begin{align}\label{eq:decidability-29}
\shuffle_\bbP^{(E)} = 
\{\ &(x_1,x_2,x_3) \in (\dsN_0^Q \times \Sigma \times \dsN_0^Q) \times 
Z_\bbP(A_\bbP) \times \shuffle_{\check{\bbP}}^\sigma \  | \nonumber \\
& \mbox{there exist } (f',\check{a},g') \in \shuffle_{\check{\bbP}}^\sigma \mbox{ and } h \in Z_\bbP(A_\bbP), \mbox{ such that } \nonumber \\
& x_3 = (f',\check{a},g'),\ x_2 = h \mbox{ and } x_1 = (f' + h , \check{\iota}(\check{a}) , g' + h) \ \}.
\end{align}
On account of \eqref{eq:A11.4} the representation \eqref{eq:decidability-28} is ambiguous.
Contrary to \eqref{eq:decidability-28}, the representation \eqref{eq:decidability-29} is unique. 
To capture the ambiguity of \eqref{eq:decidability-28} let the mapping
\begin{align}\label{eq:decidability-30}
\sigma_\bbP^{(S)} & : \shuffle_\bbP^{(S)} \to 
  2^{\shuffle_\bbP^\sigma} \setminus \{ \emptyset \} \mbox{ be defined by } \nonumber \\
\sigma_\bbP^{(S)} ((x_1,x_2,x_3)) := & \{(f',a,g') \in \shuffle_\bbP^\sigma \ | 
 \mbox{ there exists } h \in \dsN_0^Q  \mbox{ such that } \nonumber \\
& \ f' + h \in Z_\bbP(A_\bbP),\ x_2 = (f'+h,a,g'+h) \mbox{ and } \nonumber \\
& \ x_1 = (f' + h + \varphi_{\shuffle_\bbP}^{(3.1)}(x_3) , a ,
  g' + h + \varphi_{\shuffle_\bbP}^{(3.3)}(x_3)) \ \} \nonumber \\ 
\mbox{for each } & \ (x_1,x_2,x_3) \in \shuffle_\bbP^{(S)}.
\end{align}
As $\shuffle_\bbP^{()} = \shuffle_\bbP^{(S)} \ \dotcup \ \shuffle_\bbP^{(E)}$, 
for technical reasons $\sigma_\bbP^{(S)}$ can be extended to
\begin{align}\label{eq:decidability-31}
\sigma_\bbP^{()} & : \shuffle_\bbP^{()} \to 
  (2^{\shuffle_\bbP^\sigma} \setminus \{ \emptyset \}) \ 
  \dotcup \ (2^{\shuffle_{\check{\bbP}}^\sigma} \setminus \{ \emptyset \}) \mbox{ by } \nonumber \\
\sigma_\bbP^{()}((x_1,x_2,x_3)) & := \sigma_\bbP^{(S)}((x_1,x_2,x_3)) \ \ \ \ \ \ \mbox{ for } 
 (x_1,x_2,x_3) \in \shuffle_\bbP^{(S)}  \mbox{ and } \nonumber \\
\sigma_\bbP^{()}((x_1,x_2,x_3)) & := \{(f',\check{a},g') \in \shuffle_{\check{\bbP}}^\sigma \ | 
 \mbox{ there exists } h \in Z_\bbP(A_\bbP) \mbox{ such that } \nonumber \\ 
& \ \ \ \ \ \  x_3 = (f',\check{a},g'),\ x_2 = h \mbox{ and } x_1 = (f' + h , \check{\iota}(\check{a}) , g' + h) \ \} \nonumber \\
& \ \ \ \ \ \ \ \ \ \ \ \ \ \ \ \ \ \ \ \ \ \ \ \ \ \ \ \ \ \ \ \ \ \mbox{ for } (x_1,x_2,x_3) \in \shuffle_\bbP^{(E)}, \nonumber \\
\mbox{which implies }& \#(\sigma_\bbP^{()}((x_1,x_2,x_3))) = 1 \mbox{ for } (x_1,x_2,x_3) \in \shuffle_\bbP^{(E)}.
\end{align}
Now \eqref{eq:decidability-28}, \eqref{eq:decidability-29} and \eqref{eq:decidability-31} 
present an appropriate characterization of $\shuffle_\bbP^{()}$ to simulate $\bbW_\bbP^\bbV$ 
by a Petri net $N_\bbP^\bbV$, which is the final step to decide $\SP(P \cup \{\varepsilon\},V)$. 
For that purpose we additionally assume finiteness of $\bbP$ and $\bbV$. To define the 
set of places of $N_\bbP^\bbV$, let
\begin{align}\label{eq:decidability-32}
& Q^{(i)} \mbox{ and } Q_\bbV^{(i)} \mbox{ for each } i \in \{ 1,2 \} 
  \mbox{ be copies of } Q \mbox{ and } Q_\bbV \mbox{ with } \nonumber \\
& Q^{(1)} \cap Q^{(2)} = \emptyset = Q_\bbV^{(1)} \cap Q_\bbV^{(2)} \mbox{ and } 
  Q^{(i)} \cap Q_\bbV^{(j)} = \emptyset \mbox{ for each } i,j \in \{ 1,2 \}, 
  \mbox{ and let } \nonumber \\
& \tau^{(i)} : Q^{(i)} \dotcup Q_\bbV^{(i)} \to Q \dotcup Q_\bbV \mbox{ for each } 
  i \in \{ 1,2 \} \mbox{ be the corresponding bijections } \nonumber \\
& \mbox{with } \tau^{(i)}(Q^{(i)}) = Q \mbox{ and } \tau^{(i)}(Q_\bbV^{(i)}) = Q_\bbV 
  \mbox{ for each } i \in \{ 1,2 \}.
\end{align}
Corresponding to the state set $S_\bbP^\bbV$ of the semiautomaton $W_\bbP^\bbV$, 
which by \eqref{eq:decidability-11}, \eqref{eq:decidability-14}, \eqref{eq:decidability-19} 
and \eqref{eq:decidability-21} is represented by 
\[S_\bbP^\bbV = Q_\bbV \times Q_\bbV \times 
(Z_\bbP(A_\bbP) \times Z_\bbP(A_\bbP) \times (Z_\bbP(E_\bbP) \dotcup \{\check{0}\})), \]
the set $R_\bbP^\bbV$ of places of $N_\bbP^\bbV$ is defined by 
\begin{equation}\label{eq:decidability-33}
R_\bbP^\bbV := Q_\bbV^{(1)} \dotcup Q_\bbV^{(2)} \dotcup (Q^{(1)} \dotcup Q^{(2)} 
  \dotcup (Z_\bbP(E_\bbP) \dotcup \{\check{0}\})).
\end{equation}
By this definition there exists an injective mapping from $S_\bbP^\bbV$ into the set 
of markings of $N_\bbP^\bbV$. Let therefore the injection
\begin{align}\label{eq:decidability-34}
& \iota_\bbP^\bbV : S_\bbP^\bbV \to \dsN_0^{Q_\bbV^{(1)} \dotcup Q_\bbV^{(2)} \dotcup (Q^{(1)} \dotcup Q^{(2)} 
  \dotcup (Z_\bbP(E_\bbP) \dotcup \{\check{0}\}))} \mbox{ be defined by } \nonumber \\
& \iota_\bbP^\bbV((q_1,q_2,(s_1,s_2,s_3)))(x) := 0 \mbox{ for } 
   x \in Q_\bbV^{(1)} \setminus \{(\tau^{(1)})^{-1}(q_1) \},   \nonumber \\
& \iota_\bbP^\bbV((q_1,q_2,(s_1,s_2,s_3)))(x) := 1 \mbox{ for } 
   x \in Q_\bbV^{(1)} \cap \{(\tau^{(1)})^{-1}(q_1) \},   \nonumber \\
& \iota_\bbP^\bbV((q_1,q_2,(s_1,s_2,s_3)))(x) := 0 \mbox{ for } 
   x \in Q_\bbV^{(2)} \setminus \{(\tau^{(2)})^{-1}(q_2) \},   \nonumber \\
& \iota_\bbP^\bbV((q_1,q_2,(s_1,s_2,s_3)))(x) := 1 \mbox{ for } 
   x \in Q_\bbV^{(2)} \cap \{(\tau^{(2)})^{-1}(q_2) \},   \nonumber \\
& \iota_\bbP^\bbV((q_1,q_2,(s_1,s_2,s_3)))(x) := s_1(\tau^{(1)}(x)) \mbox{ for } 
   x \in Q^{(1)},   \nonumber \\
& \iota_\bbP^\bbV((q_1,q_2,(s_1,s_2,s_3)))(x) := s_2(\tau^{(2)}(x)) \mbox{ for } 
   x \in Q^{(2)},   \nonumber \\
& \iota_\bbP^\bbV((q_1,q_2,(s_1,s_2,s_3)))(x) := 0 \mbox{ for } 
   x \in (Z_\bbP(E_\bbP) \dotcup \{\check{0}\}) \setminus \{s_3 \}, \mbox{ and }   \nonumber \\
& \iota_\bbP^\bbV((q_1,q_2,(s_1,s_2,s_3)))(x) := 1 \mbox{ for } 
   x \in (Z_\bbP(E_\bbP) \dotcup \{\check{0}\}) \cap \{s_3 \} \mbox{ for each } \nonumber \\
& (q_1,q_2,(s_1,s_2,s_3)) \in S_\bbP^\bbV \subset 
   Q_\bbV \times Q_\bbV \times (\dsN_0^Q \times \dsN_0^Q \times (Z_\bbP(E_\bbP) \dotcup \{\check{0}\})).
\end{align}
The set $T_\bbP^\bbV$ of transitions of $N_\bbP^\bbV$ will be defined such that there exists 
a bijective mapping $\chi_\bbP^\bbV : 
(Q_\bbV \times Q_\bbV \times \shuffle_\bbP^\sigma) \ \dotcup \ 
(Q_\bbV \times \shuffle_{\check{\bbP}}^\sigma) \to T_\bbP^\bbV$. 
For this purpose let
\[T_\bbP^\bbV := 
  \ \tilde{T}_\bbP^{\bbV(S)} \ \dotcup \ \mathring{T}_\bbP^{\bbV(S)} 
\ \dotcup \ \bar{T}_\bbP^{\bbV(S)} \ \dotcup \ \tilde{\bar{T}}_\bbP^{\bbV(S)} \ \dotcup
 \ \tilde{T}_\bbP^{\bbV(E)} \ \dotcup \ \mathring{T}_\bbP^{\bbV(E)} 
\ \dotcup \ \bar{T}_\bbP^{\bbV(E)} \ \dotcup \ \tilde{\bar{T}}_\bbP^{\bbV(E)}, \mbox{ where } \]
\begin{align}\label{eq:decidability-35}
\tilde{T}_\bbP^{\bbV(S)} := & \{((q_1,q_2),(a,p),(p_1,p_2)) \in 
  (Q_\bbV \times Q_\bbV) \times (\Sigma \times Q) \times (Q_\bbV \times Q_\bbV) | \nonumber \\ 
  & \ (0,a,1_p) \in \tilde{\shuffle}_\bbP^\sigma, \ \delta_\bbV(q_1,a) = p_1 
  \mbox{ and } \delta_\bbV(q_2,a) = p_2  \}, \nonumber \\
\mathring{T}_\bbP^{\bbV(S)} := & \{((q_1,q_2),(q,a,p),(p_1,p_2)) \in 
  (Q_\bbV \times Q_\bbV) \times (Q \times \Sigma \times Q) \times (Q_\bbV \times Q_\bbV) | \nonumber \\ 
  & \ (1_q,a,1_p) \in \mathring{\shuffle}_\bbP^\sigma, \ \delta_\bbV(q_1,a) = p_1 
  \mbox{ and } \delta_\bbV(q_2,a) = p_2  \}, \nonumber \\
\bar{T}_\bbP^{\bbV(S)} := & \{((q_1,q_2),(q,a),(p_1,p_2)) \in 
  (Q_\bbV \times Q_\bbV) \times (Q \times \Sigma) \times (Q_\bbV \times Q_\bbV) | \nonumber \\ 
  & \ (1_q,a,0) \in \bar{\shuffle}_\bbP^\sigma, \ \delta_\bbV(q_1,a) = p_1 
  \mbox{ and } \delta_\bbV(q_2,a) = p_2  \}, \nonumber \\
\tilde{\bar{T}}_\bbP^{\bbV(S)} := & \{((q_1,q_2),a,(p_1,p_2)) \in 
  (Q_\bbV \times Q_\bbV) \times \Sigma \times (Q_\bbV \times Q_\bbV) | \nonumber \\ 
  & \ (0,a,0) \in \tilde{\bar{\shuffle}}_\bbP^\sigma, \ \delta_\bbV(q_1,a) = p_1 
  \mbox{ and } \delta_\bbV(q_2,a) = p_2  \}, \nonumber \\
\tilde{T}_\bbP^{\bbV(E)} := & \{(q_1,(\check{a},p),p_1) \in 
  Q_\bbV \times (\check{\Sigma} \times Q) \times Q_\bbV | \nonumber \\
  & \ (0,\check{a},1_p) \in \tilde{\shuffle}_{\check{\bbP}}^\sigma  
  \mbox{ and } \delta_\bbV(q_1,a) = p_1  \}, \nonumber \\
\mathring{T}_\bbP^{\bbV(E)} := & \{(q_1,(q,\check{a},p),p_1) \in 
  Q_\bbV \times (Q \times \check{\Sigma} \times Q) \times Q_\bbV | \nonumber \\ 
  & \ (1_q,\check{a},1_p) \in \mathring{\shuffle}_{\check{\bbP}}^\sigma  
  \mbox{ and } \delta_\bbV(q_1,a) = p_1  \}, \nonumber \\
\bar{T}_\bbP^{\bbV(E)} := & \{(q_1,(q,\check{a}),p_1) \in 
  Q_\bbV \times (Q \times \check{\Sigma}) \times Q_\bbV | \nonumber \\
  & \ (1_q,\check{a},0) \in \bar{\shuffle}_{\check{\bbP}}^\sigma  
  \mbox{ and } \delta_\bbV(q_1,a) = p_1  \} \mbox{ and } \nonumber \\
\tilde{\bar{T}}_\bbP^{\bbV(E)} := & \{(q_1,\check{a},p_1) \in 
  Q_\bbV \times \check{\Sigma} \times Q_\bbV | 
  (0,\check{a},0) \in \tilde{\bar{\shuffle}}_{\check{\bbP}}^\sigma  
  \mbox{ and } \delta_\bbV(q_1,a) = p_1  \}.
\end{align}
Let now the bijective mapping $\chi_\bbP^\bbV : 
(Q_\bbV \times Q_\bbV \times \shuffle_\bbP^\sigma) \ \dotcup \ 
(Q_\bbV \times \shuffle_{\check{\bbP}}^\sigma) \to T_\bbP^\bbV$ be defined by
\begin{align}\label{eq:decidability-36}
\chi_\bbP^\bbV ((q_1,q_2,(0,a,1_p))) := & ((q_1,q_2),(a,p),(\delta_\bbV(q_1,a),\delta_\bbV(q_2,a))) \nonumber \\
  & \mbox{ for } (q_1,q_2,(0,a,1_p)) \in Q_\bbV \times Q_\bbV \times \tilde{\shuffle}_\bbP^\sigma,  \nonumber \\
\chi_\bbP^\bbV ((q_1,q_2,(1_q,a,1_p))) := & ((q_1,q_2),(q,a,p),(\delta_\bbV(q_1,a),\delta_\bbV(q_2,a))) \nonumber \\
  & \mbox{ for } (q_1,q_2,(1_q,a,1_p)) \in Q_\bbV \times Q_\bbV \times \mathring{\shuffle}_\bbP^\sigma,  \nonumber \\
\chi_\bbP^\bbV ((q_1,q_2,(1_q,a,0))) := & ((q_1,q_2),(q,a),(\delta_\bbV(q_1,a),\delta_\bbV(q_2,a))) \nonumber \\
  & \mbox{ for } (q_1,q_2,(1_q,a,0)) \in Q_\bbV \times Q_\bbV \times \bar{\shuffle}_\bbP^\sigma,  \nonumber \\
\chi_\bbP^\bbV ((q_1,q_2,(0,a,0))) := & ((q_1,q_2),a,(\delta_\bbV(q_1,a),\delta_\bbV(q_2,a))) \nonumber \\
  & \mbox{ for } (q_1,q_2,(0,a,0)) \in Q_\bbV \times Q_\bbV \times \tilde{\bar{\shuffle}}_\bbP^\sigma,  \nonumber \\
\chi_\bbP^\bbV ((q_1,(0,\check{a},1_p))) := & (q_1,(\check{a},p),\delta_\bbV(q_1,a)) \nonumber \\
  & \mbox{ for } (q_1,(0,\check{a},1_p)) \in Q_\bbV \times \tilde{\shuffle}_{\check{\bbP}}^\sigma,  \nonumber \\
\chi_\bbP^\bbV ((q_1,(1_q,\check{a},1_p))) := & (q_1,(q,\check{a},p),\delta_\bbV(q_1,a)) \nonumber \\
  & \mbox{ for } (q_1,(1_q,\check{a},1_p)) \in Q_\bbV \times \mathring{\shuffle}_{\check{\bbP}}^\sigma,  \nonumber \\
\chi_\bbP^\bbV ((q_1,(1_q,\check{a},0))) := & (q_1,(q,\check{a}),\delta_\bbV(q_1,a)) \nonumber \\
  & \mbox{ for } (q_1,(1_q,\check{a},0)) \in Q_\bbV \times \bar{\shuffle}_{\check{\bbP}}^\sigma \mbox{ and } \nonumber \\
\chi_\bbP^\bbV ((q_1,(0,\check{a},0))) := & (q_1,\check{a},\delta_\bbV(q_1,a)) \nonumber \\
  & \mbox{ for } (q_1,(0,\check{a},0)) \in Q_\bbV \times \tilde{\bar{\shuffle}}_{\check{\bbP}}^\sigma.
\end{align}
The set $K_\bbP^\bbV$ of edges of $N_\bbP^\bbV$ let be defined by
\[K_\bbP^\bbV := 
  \ \tilde{K}_\bbP^{\bbV(S)} \ \dotcup \ \mathring{K}_\bbP^{\bbV(S)} 
\ \dotcup \ \bar{K}_\bbP^{\bbV(S)} \ \dotcup \ \tilde{\bar{K}}_\bbP^{\bbV(S)} \ \dotcup
 \ \tilde{K}_\bbP^{\bbV(E)} \ \dotcup \ \mathring{K}_\bbP^{\bbV(E)} 
\ \dotcup \ \bar{K}_\bbP^{\bbV(E)} \ \dotcup \ \tilde{\bar{K}}_\bbP^{\bbV(E)}, \]
where
\begin{align}\label{eq:decidability-37}
&\tilde{K}_\bbP^{\bbV(S)} := \bigcup\limits_{((q_1,q_2),(a,p),(p_1,p_2)) \in \tilde{T}_\bbP^{\bbV(S)}} \nonumber \\
&\ \ \ \ \ \ \ \ \ \ \ \ \ \ \ \ \ \ \ \ \ \ \ \ \ \{((\tau^{(1)})^{-1}(q_1),((q_1,q_2),(a,p),(p_1,p_2))), \nonumber \\
&\ \ \ \ \ \ \ \ \ \ \ \ \ \ \ \ \ \ \ \ \ \ \ \ \ \ ((\tau^{(2)})^{-1}(q_2),((q_1,q_2),(a,p),(p_1,p_2))), \nonumber \\
&\ \ \ \ \ \ \ \ \ \ \ \ \ \ \ \ \ \ \ \ \ \ \ \ \ \ (((q_1,q_2),(a,p),(p_1,p_2)),(\tau^{(1)})^{-1}(p_1)), \nonumber \\
&\ \ \ \ \ \ \ \ \ \ \ \ \ \ \ \ \ \ \ \ \ \ \ \ \ \ (((q_1,q_2),(a,p),(p_1,p_2)),(\tau^{(2)})^{-1}(p_2)), \nonumber \\
&\ \ \ \ \ \ \ \ \ \ \ \ \ \ \ \ \ \ \ \ \ \ \ \ \ \ (((q_1,q_2),(a,p),(p_1,p_2)),(\tau^{(1)})^{-1}(p)), \nonumber \\
&\ \ \ \ \ \ \ \ \ \ \ \ \ \ \ \ \ \ \ \ \ \ \ \ \ \ (((q_1,q_2),(a,p),(p_1,p_2)),(\tau^{(2)})^{-1}(p)) \} \subset \nonumber \\
&((Q_\bbV^{(1)} \dotcup Q_\bbV^{(2)}) \times \tilde{T}_\bbP^{\bbV(S)})
 \dotcup(\tilde{T}_\bbP^{\bbV(S)}\times
(Q_\bbV^{(1)} \dotcup Q_\bbV^{(2)} \dotcup Q^{(1)} \dotcup Q^{(2)})),
\end{align}
\begin{align}\label{eq:decidability-38}
&\mathring{K}_\bbP^{\bbV(S)} := \bigcup\limits_{((q_1,q_2),(q,a,p),(p_1,p_2)) \in \mathring{T}_\bbP^{\bbV(S)}} \nonumber \\
&\ \ \ \ \ \ \ \ \ \ \ \ \ \ \ \ \ \ \ \ \ \ \ \ \ \{((\tau^{(1)})^{-1}(q_1),((q_1,q_2),(q,a,p),(p_1,p_2))), \nonumber \\
&\ \ \ \ \ \ \ \ \ \ \ \ \ \ \ \ \ \ \ \ \ \ \ \ \ \ ((\tau^{(2)})^{-1}(q_2),((q_1,q_2),(q,a,p),(p_1,p_2))), \nonumber \\
&\ \ \ \ \ \ \ \ \ \ \ \ \ \ \ \ \ \ \ \ \ \ \ \ \ \ ((\tau^{(1)})^{-1}(q),((q_1,q_2),(q,a,p),(p_1,p_2))), \nonumber \\
&\ \ \ \ \ \ \ \ \ \ \ \ \ \ \ \ \ \ \ \ \ \ \ \ \ \ ((\tau^{(2)})^{-1}(q),((q_1,q_2),(q,a,p),(p_1,p_2))), \nonumber \\
&\ \ \ \ \ \ \ \ \ \ \ \ \ \ \ \ \ \ \ \ \ \ \ \ \ \ (((q_1,q_2),(q,a,p),(p_1,p_2)),(\tau^{(1)})^{-1}(p_1)), \nonumber \\
&\ \ \ \ \ \ \ \ \ \ \ \ \ \ \ \ \ \ \ \ \ \ \ \ \ \ (((q_1,q_2),(q,a,p),(p_1,p_2)),(\tau^{(2)})^{-1}(p_2)), \nonumber \\
&\ \ \ \ \ \ \ \ \ \ \ \ \ \ \ \ \ \ \ \ \ \ \ \ \ \ (((q_1,q_2),(q,a,p),(p_1,p_2)),(\tau^{(1)})^{-1}(p)), \nonumber \\
&\ \ \ \ \ \ \ \ \ \ \ \ \ \ \ \ \ \ \ \ \ \ \ \ \ \ (((q_1,q_2),(q,a,p),(p_1,p_2)),(\tau^{(2)})^{-1}(p)) \} \subset \nonumber \\
&((Q_\bbV^{(1)} \dotcup Q_\bbV^{(2)} \dotcup Q^{(1)} \dotcup Q^{(2)}) \times \mathring{T}_\bbP^{\bbV(S)})
 \dotcup(\mathring{T}_\bbP^{\bbV(S)}\times
(Q_\bbV^{(1)} \dotcup Q_\bbV^{(2)} \dotcup Q^{(1)} \dotcup Q^{(2)})),
\end{align}
\begin{align}\label{eq:decidability-39}
&\bar{K}_\bbP^{\bbV(S)} := \bigcup\limits_{((q_1,q_2),(q,a),(p_1,p_2)) \in \bar{T}_\bbP^{\bbV(S)}} \nonumber \\
&\ \ \ \ \ \ \ \ \ \ \ \ \ \ \ \ \ \ \ \ \ \ \ \ \ \{((\tau^{(1)})^{-1}(q_1),((q_1,q_2),(q,a),(p_1,p_2))), \nonumber \\
&\ \ \ \ \ \ \ \ \ \ \ \ \ \ \ \ \ \ \ \ \ \ \ \ \ \ ((\tau^{(2)})^{-1}(q_2),((q_1,q_2),(q,a),(p_1,p_2))), \nonumber \\
&\ \ \ \ \ \ \ \ \ \ \ \ \ \ \ \ \ \ \ \ \ \ \ \ \ \ ((\tau^{(1)})^{-1}(q),((q_1,q_2),(q,a),(p_1,p_2))), \nonumber \\
&\ \ \ \ \ \ \ \ \ \ \ \ \ \ \ \ \ \ \ \ \ \ \ \ \ \ ((\tau^{(2)})^{-1}(q),((q_1,q_2),(q,a),(p_1,p_2))), \nonumber \\
&\ \ \ \ \ \ \ \ \ \ \ \ \ \ \ \ \ \ \ \ \ \ \ \ \ \ (((q_1,q_2),(q,a,p),(p_1,p_2)),(\tau^{(1)})^{-1}(p_1)), \nonumber \\
&\ \ \ \ \ \ \ \ \ \ \ \ \ \ \ \ \ \ \ \ \ \ \ \ \ \ (((q_1,q_2),(q,a,p),(p_1,p_2)),(\tau^{(2)})^{-1}(p_2)) \} \subset \nonumber \\
&((Q_\bbV^{(1)} \dotcup Q_\bbV^{(2)} \dotcup Q^{(1)} \dotcup Q^{(2)}) \times \bar{T}_\bbP^{\bbV(S)})
 \dotcup(\bar{T}_\bbP^{\bbV(S)}\times
(Q_\bbV^{(1)} \dotcup Q_\bbV^{(2)})),
\end{align}
\begin{align}\label{eq:decidability-40}
&\tilde{\bar{K}}_\bbP^{\bbV(S)} := \bigcup\limits_{((q_1,q_2),a,(p_1,p_2)) \in \tilde{\bar{T}}_\bbP^{\bbV(S)}} \nonumber \\
&\ \ \ \ \ \ \ \ \ \ \ \ \ \ \ \ \ \ \ \ \ \ \ \ \ \{((\tau^{(1)})^{-1}(q_1),((q_1,q_2),a,(p_1,p_2))), \nonumber \\
&\ \ \ \ \ \ \ \ \ \ \ \ \ \ \ \ \ \ \ \ \ \ \ \ \ \ ((\tau^{(2)})^{-1}(q_2),((q_1,q_2),a,(p_1,p_2))), \nonumber \\
&\ \ \ \ \ \ \ \ \ \ \ \ \ \ \ \ \ \ \ \ \ \ \ \ \ \ (((q_1,q_2),a,(p_1,p_2)),(\tau^{(1)})^{-1}(p_1)), \nonumber \\
&\ \ \ \ \ \ \ \ \ \ \ \ \ \ \ \ \ \ \ \ \ \ \ \ \ \ (((q_1,q_2),a,(p_1,p_2)),(\tau^{(2)})^{-1}(p_2)) \} \subset \nonumber \\
&((Q_\bbV^{(1)} \dotcup Q_\bbV^{(2)}) \times \tilde{\bar{T}}_\bbP^{\bbV(S)})
 \dotcup(\tilde{\bar{T}}_\bbP^{\bbV(S)}\times
(Q_\bbV^{(1)} \dotcup Q_\bbV^{(2)})),
\end{align}
\begin{align}\label{eq:decidability-41}
&\tilde{K}_\bbP^{\bbV(E)} := \bigcup\limits_{(q_1,(\check{a},p),p_1) \in \tilde{T}_\bbP^{\bbV(E)}} \nonumber \\
&\ \ \ \ \ \ \ \ \ \ \ \ \ \ \ \ \ \ \ \ \ \ \ \ \ \{((\tau^{(1)})^{-1}(q_1),(q_1,(\check{a},p),p_1)), \nonumber \\
&\ \ \ \ \ \ \ \ \ \ \ \ \ \ \ \ \ \ \ \ \ \ \ \ \ \ (0,(q_1,(\check{a},p),p_1)), \nonumber \\
&\ \ \ \ \ \ \ \ \ \ \ \ \ \ \ \ \ \ \ \ \ \ \ \ \ \ ((q_1,(\check{a},p),p_1),(\tau^{(1)})^{-1}(p_1)), \nonumber \\
&\ \ \ \ \ \ \ \ \ \ \ \ \ \ \ \ \ \ \ \ \ \ \ \ \ \ ((q_1,(\check{a},p),p_1),1_p), \nonumber \\
&\ \ \ \ \ \ \ \ \ \ \ \ \ \ \ \ \ \ \ \ \ \ \ \ \ \ ((q_1,(\check{a},p),p_1),(\tau^{(1)})^{-1}(p)) \} \subset \nonumber \\
&((Q_\bbV^{(1)} \dotcup (Z_\bbP(E_\bbP) \dotcup \{\check{0}\})) \times \tilde{T}_\bbP^{\bbV(E)})\dotcup \nonumber \\
&(\tilde{T}_\bbP^{\bbV(E)}\times
(Q_\bbV^{(1)} \dotcup (Z_\bbP(E_\bbP) \dotcup \{\check{0}\}) \dotcup Q^{(1)})),
\end{align}
\begin{align}\label{eq:decidability-42}
&\mathring{K}_\bbP^{\bbV(E)} := \bigcup\limits_{(q_1,(q,\check{a},p),p_1) \in \mathring{T}_\bbP^{\bbV(E)}} \nonumber \\
&\ \ \ \ \ \ \ \ \ \ \ \ \ \ \ \ \ \ \ \ \ \ \ \ \ \{((\tau^{(1)})^{-1}(q_1),(q_1,(q,\check{a},p),p_1)), \nonumber \\
&\ \ \ \ \ \ \ \ \ \ \ \ \ \ \ \ \ \ \ \ \ \ \ \ \ \ (1_q,(q_1,(q,\check{a},p),p_1)), \nonumber \\
&\ \ \ \ \ \ \ \ \ \ \ \ \ \ \ \ \ \ \ \ \ \ \ \ \ \{((\tau^{(1)})^{-1}(q),(q_1,(q,\check{a},p),p_1)), \nonumber \\
&\ \ \ \ \ \ \ \ \ \ \ \ \ \ \ \ \ \ \ \ \ \ \ \ \ \ ((q_1,(q,\check{a},p),p_1),(\tau^{(1)})^{-1}(p_1)), \nonumber \\
&\ \ \ \ \ \ \ \ \ \ \ \ \ \ \ \ \ \ \ \ \ \ \ \ \ \ ((q_1,(q,\check{a},p),p_1),1_p), \nonumber \\
&\ \ \ \ \ \ \ \ \ \ \ \ \ \ \ \ \ \ \ \ \ \ \ \ \ \ ((q_1,(q,\check{a},p),p_1),(\tau^{(1)})^{-1}(p)) \} \subset \nonumber \\
&((Q_\bbV^{(1)} \dotcup (Z_\bbP(E_\bbP) \dotcup \{\check{0}\}) \dotcup Q^{(1)}) \times \mathring{T}_\bbP^{\bbV(E)})\dotcup \nonumber \\
&(\mathring{T}_\bbP^{\bbV(E)}\times
(Q_\bbV^{(1)} \dotcup (Z_\bbP(E_\bbP) \dotcup \{\check{0}\}) \dotcup Q^{(1)})),
\end{align}
\begin{align}\label{eq:decidability-43}
&\bar{K}_\bbP^{\bbV(E)} := \bigcup\limits_{(q_1,(q,\check{a}),p_1) \in \bar{T}_\bbP^{\bbV(E)}} \nonumber \\
&\ \ \ \ \ \ \ \ \ \ \ \ \ \ \ \ \ \ \ \ \ \ \ \ \ \{((\tau^{(1)})^{-1}(q_1),(q_1,(q,\check{a}),p_1)), \nonumber \\
&\ \ \ \ \ \ \ \ \ \ \ \ \ \ \ \ \ \ \ \ \ \ \ \ \ \ (1_q,(q_1,(q,\check{a}),p_1)), \nonumber \\
&\ \ \ \ \ \ \ \ \ \ \ \ \ \ \ \ \ \ \ \ \ \ \ \ \ \{((\tau^{(1)})^{-1}(q),(q_1,(q,\check{a}),p_1)), \nonumber \\
&\ \ \ \ \ \ \ \ \ \ \ \ \ \ \ \ \ \ \ \ \ \ \ \ \ \ ((q_1,(q,\check{a}),p_1),(\tau^{(1)})^{-1}(p_1)), \nonumber \\
&\ \ \ \ \ \ \ \ \ \ \ \ \ \ \ \ \ \ \ \ \ \ \ \ \ \ ((q_1,(q,\check{a}),p_1),\check{0}) \} \subset \nonumber \\
&((Q_\bbV^{(1)} \dotcup (Z_\bbP(E_\bbP) \dotcup \{\check{0}\}) \dotcup Q^{(1)}) \times \bar{T}_\bbP^{\bbV(E)})\dotcup \nonumber \\
&(\bar{T}_\bbP^{\bbV(E)}\times
(Q_\bbV^{(1)} \dotcup (Z_\bbP(E_\bbP) \dotcup \{\check{0}\}))) \ \ \ \ \mbox{ and }
\end{align}
\begin{align}\label{eq:decidability-44}
&\tilde{\bar{K}}_\bbP^{\bbV(E)} := \bigcup\limits_{(q_1,\check{a},p_1) \in \tilde{\bar{T}}_\bbP^{\bbV(E)}} \nonumber \\
&\ \ \ \ \ \ \ \ \ \ \ \ \ \ \ \ \ \ \ \ \ \ \ \ \ \{((\tau^{(1)})^{-1}(q_1),(q_1,\check{a},p_1)), \nonumber \\
&\ \ \ \ \ \ \ \ \ \ \ \ \ \ \ \ \ \ \ \ \ \ \ \ \ \ (0,(q_1,\check{a},p_1)), \nonumber \\
&\ \ \ \ \ \ \ \ \ \ \ \ \ \ \ \ \ \ \ \ \ \ \ \ \ \ ((q_1,\check{a},p_1),(\tau^{(1)})^{-1}(p_1)), \nonumber \\
&\ \ \ \ \ \ \ \ \ \ \ \ \ \ \ \ \ \ \ \ \ \ \ \ \ \ ((q_1,\check{a},p_1),\check{0}) \} \subset \nonumber \\
&((Q_\bbV^{(1)} \dotcup (Z_\bbP(E_\bbP) \dotcup \{\check{0}\})) \times \tilde{\bar{T}}_\bbP^{\bbV(E)})\dotcup \nonumber \\
&(\tilde{\bar{T}}_\bbP^{\bbV(E)}\times
(Q_\bbV^{(1)} \dotcup (Z_\bbP(E_\bbP) \dotcup \{\check{0}\}))).
\end{align}
Let now the sets $\Omega_\bbP^\bbV$, $\mcO_\bbP^\bbV$ and $\mcE_\bbP^\bbV(M)$ as well as the functions 
$\mcI_\bbP^\bbV$ and $\mcF_\bbP^\bbV$ be defined corresponding to \eqref{eq:A11.6}, \eqref{eq:A11.7} 
and \eqref{eq:A11.7'}. By induction on the length of occurrence sequences 
$o \in (\mcI_\bbP^\bbV)^{ -1}(\iota_\bbP^\bbV((q_{\bbV_0}, q_{\bbV_0},(0,0,0))))$ it can be shown that
\begin{align}\label{eq:decidability-45}
& \sum\limits_{x \in Q_\bbV^{(1)}} M(x) = \sum\limits_{x \in Q_\bbV^{(2)}} M(x) = 
  \sum\limits_{x \in Z_\bbP(E_\bbP) \dotcup \{\check{0}\}} M(x) = 1    \nonumber \\
& \mbox{for each } M \in \mcE_\bbP^\bbV(\iota_\bbP^\bbV((q_{\bbV_0}, q_{\bbV_0},(0,0,0))).
\end{align}
Therefore the function $\zeta_\bbP^{\bbV(3)} : 
\mcE_\bbP^\bbV(\iota_\bbP^\bbV((q_{\bbV_0}, q_{\bbV_0},(0,0,0))) \to \dsN_0^Q$ is  well defined for each 
$M \in \mcE_\bbP^\bbV(\iota_\bbP^\bbV((q_{\bbV_0}, q_{\bbV_0},(0,0,0)))$ by
\begin{align}\label{eq:decidability-46}
& \zeta_\bbP^{\bbV(3)}(M) := 0 \mbox{ if } M(\check{0}) = 1 \mbox{ and } \nonumber \\
& \zeta_\bbP^{\bbV(3)}(M) := f \mbox{ if } M(f) = 1 \mbox{ for } f \in Z_\bbP(E_\bbP).
\end{align}
For $i \in \{1,2\}$ let the functions $\zeta_\bbP^{\bbV(i)} : 
\mcE_\bbP^\bbV(\iota_\bbP^\bbV((q_{\bbV_0}, q_{\bbV_0},(0,0,0))) \to \dsN_0^Q$ 
for each $M \in \mcE_\bbP^\bbV(\iota_\bbP^\bbV((q_{\bbV_0}, q_{\bbV_0},(0,0,0)))$ 
be defined by
\begin{align}\label{eq:decidability-47}
& \zeta_\bbP^{\bbV(1)}(M)(q) := M((\tau^{(1)})^{-1}(q)) \mbox{ and } \nonumber \\
& \zeta_\bbP^{\bbV(2)}(M)(q) := M((\tau^{(2)})^{-1}(q)) \mbox{ for each } q \in Q.
\end{align}
An induction as for \eqref{eq:decidability-45} proves
\begin{align}\label{eq:decidability-48}
& \zeta_\bbP^{\bbV(1)}(M) = \zeta_\bbP^{\bbV(2)}(M) + \zeta_\bbP^{\bbV(3)}(M), \nonumber \\
& \zeta_\bbP^{\bbV(1)}(M) \in Z_\bbP(A_\bbP) \mbox{ and } 
  \zeta_\bbP^{\bbV(2)}(M) \in Z_\bbP(A_\bbP) \nonumber \\
& \mbox{for each } M \in \mcE_\bbP^\bbV(\iota_\bbP^\bbV((q_{\bbV_0}, q_{\bbV_0},(0,0,0))).
\end{align}
To formulate the main theorem about the simulation of 
$\bbW_\bbP^\bbV$ by $N_\bbP^\bbV$ let the mapping
\begin{align}\label{eq:decidability-49}
& \sigma_\bbP^{\bbV()} : \Psi_\bbP^\bbV \to 
  (2^{Q_\bbV \times Q_\bbV \times \shuffle_\bbP^\sigma} \setminus \{ \emptyset \}) \ 
  \dotcup \ (2^{Q_\bbV \times \shuffle_{\check{\bbP}}^\sigma} \setminus \{ \emptyset \}) \mbox{ be defined by } \nonumber \\
& \sigma_\bbP^{\bbV()}(((y_1,y_2),x,(y'_1,y'_2))) := \{y_1\} \times \{y_2\} \times \sigma_\bbP^{()}  
  \mbox{ for } \nonumber \\ 
& ((y_1,y_2),x,(y'_1,y'_2)) \in \Psi_\bbP^\bbV \cap 
  (Q_\bbV \times Q_\bbV) \times \shuffle_\bbP^{(S)} \times (Q_\bbV \times Q_\bbV) \mbox{ and } \nonumber \\ 
& \sigma_\bbP^{\bbV()}(((y_1,y_2),x,(y'_1,y'_2))) := \{y_1\} \times \sigma_\bbP^{()}  
  \mbox{ for } \nonumber \\ 
& ((y_1,y_2),x,(y'_1,y'_2)) \in \Psi_\bbP^\bbV \cap 
  (Q_\bbV \times Q_\bbV) \times \shuffle_\bbP^{(E)} \times (Q_\bbV \times Q_\bbV).
\end{align}
Now, together with \eqref{eq:decidability-48} and \eqref{eq:decidability-45} an induction on the 
length of $w \in W_\bbP^\bbV$ proves 
\begin{theorem}\label{thm:decidability-50}
\begin{align*}
& \mbox{For each }o = o_1...o_{|o|} \in (\dsN_0^{R_\bbP^\bbV} \times T_\bbP^\bbV \times \dsN_0^{R_\bbP^\bbV})^+ \mbox{with } \nonumber \\
& o_i \in \dsN_0^{R_\bbP^\bbV} \times T_\bbP^\bbV \times \dsN_0^{R_\bbP^\bbV} 
  \mbox{ for } 1 \leq i \leq |o| \mbox{ holds } o 
  \in (\mcI_\bbP^\bbV)^{-1}(\iota_\bbP^\bbV(q_{\bbP_0}^\bbV)),\nonumber \\
& \mbox{iff there exists } w \in W_\bbP^\bbV 
\mbox{ with } |w| = |o| \mbox{ such that for } 1\leq i \leq \abs{o} \mbox{ holds:}\nonumber \\
& o_i=(\iota_\bbP^\bbV(\lambda_\bbP^\bbV(q_{\bbP_0}^\bbV,w'_{i-1})),
  t_i,\iota_\bbP^\bbV(\lambda_\bbP^\bbV(q_{\bbP_0}^\bbV,w'_i))) \mbox{ with } w'_j\in \pre(w), \nonumber \\
& \abs{w'_j}=j \mbox{ for } 0\leq j \leq \abs{o},\mbox{ and } 
  t_i \in \chi_\bbP^\bbV(\sigma_\bbP^{\bbV()}(w_i)), \mbox{ where } \nonumber \\
& w=w_1\ldots{}w_{\abs{o}} \mbox{ and } w_i\in \Psi_\bbP^\bbV \mbox{ for } 1\leq i \leq \abs{o}.
\end{align*}
\end{theorem}
Theorem~\ref{thm:decidability-50} implies
\begin{equation}\label{eq:decidability-51}
\iota_\bbP^\bbV(\lambda_\bbP^\bbV(q_{\bbP_0}^\bbV,W_\bbP^\bbV)) = 
\mcE_\bbP^\bbV(\iota_\bbP^\bbV(q_{\bbP_0}^\bbV)).
\end{equation}
As the reachability problem of Petri nets is decidable \cite{reutenauer90}, \cite{Wimmel08}, by 
\eqref{eq:decidability-17} and \eqref{eq:decidability-51} follows 
\begin{corollary}\label{cor:decidability-52} \ \\
$\SP(P \cup \{\varepsilon\},V)$ is decidable for regular languages $P$ and $V$.
\end{corollary}
The decidability of $\SP(P \cup \{\varepsilon\},V)$ essentially 
depends on the decidability of the Petri net reachability problem.
In \cite{reutenauer90} this decidability result is annotated as double complex: 
in the proof and in the algorithm. 
For practical applications it is therefore important, to have 
simpler sufficient conditions for $\SP(P \cup \{\varepsilon\},V)$,
as demonstrated in Example~\ref{ex:i-s-compatible}, Example~\ref{ex:easy-compatible} 
and in Example~\ref{ex:RT}.

\clearpage
\section*{Appendix}

\appendix

\section{Shuffle Projection in Terms of Shuffle Factors}
The \emph{ shuffle product } $U \Sha V$ \cite{berstel79} for languages $U$ and $V$ can be defined 
in terms of the homomorphisms $\tau_n^I$ and $\Theta^I$.
\begin{definition}\label{def:shuffle-prod}\ \\
For $U, V \subset \Sigma^*$ the \emph{ shuffle product } $U \Sha V \subset \Sigma^*$ 
is defined by 
\[U \Sha V := \Theta^{\{1,2\}}[(\tau_1^{\{1,2\}})^{-1}(U) \cap (\tau_2^{\{1,2\}})^{-1}(V)] .\]
\end{definition}
It is easy to see that
\begin{align} \label{eq:shuffle-fac-0}
& \Sha \mbox{ is commutative, } \{w\} = \{w\} \Sha \{\varepsilon\} 
   \mbox{ for } w \in \Sigma^*, \nonumber \\
& |w| = |u| + |v| \mbox{ for } w \in \{u\} \Sha \{v\} 
  \mbox{ and } u,v \in \Sigma^*, \nonumber \\
& \pre(U \Sha V)=\pre(U) \Sha \pre(V) , \mbox{ and } \nonumber \\
& U \Sha V = \bigcup\limits_{u \in U, v \in V} \{u\} \Sha \{v\} \mbox{ for } U,V \subset \Sigma^*. \nonumber \\
& \mbox{By Lemma~\ref{lemma:shuffle-lemma2} } \nonumber \\
& \{1,2\} \mbox{ can be replaced by any set } S \mbox{ with } \#(S) = 2.
\end{align}
The following lemma is the key to a relation between shuffle products and shuffle projection.
\begin{lemma} \label{lemma:shuffle-fac-1}\ \\
Let $P\subset \Sigma^*$. Then $w \in \{u\} \Sha \{v\}$ for $u,v \in P^\shuffle$, iff there exist 
\begin{align*}
& x \in \bigcap\limits_{t \in \dsN} (\tau_t^\dsN)^{-1}(P \cup \{\varepsilon\}) \mbox{ and } 
K \subset \dsN \mbox{ with } w = \Theta^\dsN(x) \in P^\shuffle , \nonumber \\
& u = \Theta^K(\Pi_K^\dsN(x)) \mbox{ and } 
 v = \Theta^{\dsN \setminus K}(\Pi_{\dsN \setminus K}^\dsN(x)).
\end{align*}
\end{lemma}
\begin{proof}
\ \\
Let $x \in \bigcap\limits_{t \in \dsN} (\tau_t^\dsN)^{-1}(P \cup \{\varepsilon\})$ and $K \subset \dsN$, 
then $w := \Theta^\dsN(x) \in P^\shuffle$ and by Lemma~\ref{lemma:struc-rep-proj} 
$u := \Theta^K(\Pi_K^\dsN(x)) \in P^\shuffle$ and 
$v := \Theta^{\dsN \setminus K}(\Pi_{\dsN \setminus K}^\dsN(x)) \in P^\shuffle$.\\ 
Let $\omega_K : \Sigma_\dsN^* \to  \Sigma_{\{1,2\}}^*$ be defined by 
$\omega_K(a) := (\tau_1^{\{1\}})^{-1}(\Theta^K(a)) \mbox{ for } a \in \Sigma_K \mbox{ and }$  
$\omega_K(a) := (\tau_2^{\{2\}})^{-1}(\Theta^{\dsN \setminus K}(a))\mbox{ for } a \in \Sigma_{\dsN \setminus K}$, 
then $w = \Theta^\dsN(x) = \Theta^{\{1,2\}}(\omega_K(x))$ and 
$\omega_K(x) \in (\tau_1^{\{1,2\}})^{-1}(\{u\}) \cap (\tau_2^{\{1,2\}})^{-1}(\{v\})$. 
This implies $w \in \{u\} \Sha \{v\}$. \\
\ \\
Let now $u,v \in P^\shuffle$ and $w \in \{u\} \Sha \{v\}$. Then there exist \\ 
$\breve{w} \in (\tau_1^{\{1,2\}})^{-1}(\{u\}) \cap (\tau_2^{\{1,2\}})^{-1}(\{v\})$ with
$\Theta^{\{1,2\}}(\breve{w}) = w$,\\ 
$\breve{u} \in \bigcap\limits_{t \in \dsN} (\tau_t^\dsN)^{-1}(P \cup \{\varepsilon\})$ 
with $\Theta^\dsN(\breve{u}) = u$ and \\
$\breve{v} \in \bigcap\limits_{t \in \dsN} (\tau_t^\dsN)^{-1}(P \cup \{\varepsilon\})$ 
with $\Theta^\dsN(\breve{v}) = v$. \\
Then, by ``combining the structures'' of $\breve{w}$, $\breve{u}$ and $\breve{v}$ there exists
\begin{align} \label{eq:shuffle-fac-1}
& \breve{x} \in \bigcap\limits_{n \in \dsN \times \{1,2\}} (\tau_n^{\dsN \times \{1,2\}})^{-1}(P \cup \{\varepsilon\}) 
\mbox{ with } \Theta_{\{1,2\}}^{\dsN \times \{1,2\}}(\breve{x}) = \breve{w} , \nonumber \\
& \Theta_\dsN^{\dsN \times \{1\}}(\Pi_{\dsN \times \{1\}}^{\dsN \times \{1,2\}}(\breve{x})) = \breve{u}, \  
 \Theta_\dsN^{\dsN \times \{2\}}(\Pi_{\dsN \times \{2\}}^{\dsN \times \{1,2\}}(\breve{x})) = \breve{v}, \nonumber \\
& |\Pi_{\dsN \times \{1\}}^{\dsN \times \{1,2\}}(x')| = |\tau_1^{\{1,2\}}(w')| \mbox{ and } 
  |\Pi_{\dsN \times \{2\}}^{\dsN \times \{1,2\}}(x')| = |\tau_2^{\{1,2\}}(w')| \nonumber \\
& \mbox{for each } x' \in \pre(\breve{x}) \mbox{ and } w' \in \pre(\breve{w}) \mbox{ with } |x'| = |w'|.
\end{align}
This implies $w = \Theta^{\dsN \times \{1,2\}}(\breve{x})$,
\begin{equation} \label{eq:shuffle-fac-2}
u = \Theta^{\dsN \times \{1\}}(\Pi_{\dsN \times \{1\}}^{\dsN \times \{1,2\}}(\breve{x})) \mbox{ and } 
v =  \Theta^{\dsN \times \{2\}}(\Pi_{\dsN \times \{2\}}^{\dsN \times \{1,2\}}(\breve{x})).
\end{equation} 
Each bijection $\iota : N \to N'$ defines an isomorphism $\iota_{N'}^N : \Sigma_N^* \to \Sigma_{N'}^*$ 
by \ \ \ \ \ \ \ \ \ \ \ \ $\iota_{N'}^N := (\tau_{\iota(i)}^{\{\iota(i)\}})^{-1}(\tau_i^{\{i\}}(a))$ 
for $a \in \Sigma_{\{i\}}$ and $i \in N$. Then it is easy to see \cite{SysMea14} that
\begin{align} \label{eq:shuffle-fac-3}
& \iota_{N'}^N(x) \in \bigcap\limits_{t \in N'} (\tau_t^{N'})^{-1}(P \cup \{\varepsilon\}) 
\mbox{ and }
 \Theta^K(\Pi_K^N(x)) = \Theta^{\iota(K)}(\Pi_{\iota(K)}^{N'}(\iota_{N'}^N(x))) \nonumber \\
& \mbox{for } x \in \bigcap\limits_{t \in N} (\tau_t^N)^{-1}(P \cup \{\varepsilon\}) \mbox{ and } 
K \subset N.
\end{align}
Applying \eqref{eq:shuffle-fac-3} with $N = \dsN \times \{1,2\}$ and $N' = \dsN$ to 
\eqref{eq:shuffle-fac-1} and \eqref{eq:shuffle-fac-2} completes the proof of the lemma.
\end{proof}
Moreover, the second part of this proof shows
\begin{corollary}\label{cor:shuffle-fac-1}\ \\
Let $P\subset \Sigma^*$. Then $w \in \{u\} \Sha \{v\}$ for $u,v \in P^\shuffle$, iff there exist 
\begin{align*}
& x \in \bigcap\limits_{t \in \dsN} (\tau_t^\dsN)^{-1}(P \cup \{\varepsilon\}) \mbox{ and } 
K \subset \dsN \mbox{ with } \#(K) = \#(\dsN \setminus K) = \#(\dsN) , \nonumber \\
& w = \Theta^\dsN(x) \in P^\shuffle ,\ u = \Theta^K(\Pi_K^\dsN(x)) \mbox{ and } 
 v = \Theta^{\dsN \setminus K}(\Pi_{\dsN \setminus K}^\dsN(x)).
\end{align*}
\end{corollary}
Lemma~\ref{lemma:shuffle-fac-1} and corollary~\ref{cor:shuffle-fac-1} motivates
\begin{definition}\label{def:shuffle-factors}\ \\
For $P \subset \Sigma^*$ let $\SWF_P : 2^{P^\shuffle} \to 2^{P^\shuffle}$ 
be defined by 
\[\SWF_P(M) := \{u \in P^\shuffle | \mbox{ there exist } w \in M \mbox{ and } v \in P^\shuffle 
\mbox{ such that } w \in \{u\} \Sha \{v\} \} \]
for $M \subset P^\shuffle$. The elements of $\SWF_P(M)$ are called \emph{ shuffle factors } of $M$.
\end{definition}
It is an immediate consequence of this definition that 
\begin{align} \label{eq:shuffle-fac-4}
& M \subset \SWF_P(M), \ \SWF_P(M) = \bigcup\limits_{w \in M } \SWF_P(\{w\}) \nonumber \\
& \mbox{and therefore } \SWF_P(U) \subset \SWF_P(M) \mbox{ for } U \subset M \subset P^\shuffle.
\end{align}
\begin{theorem}\label{thm:shuffle-fac-1}\ \\
Let $P,V \subset \Sigma^*$. Then $\SP(P \cup \{\varepsilon\},V)$  iff 
$\SWF_P(P^\shuffle \cap V) \subset P^\shuffle \cap V$.
\end{theorem}
\begin{proof} \ \\
By Corollary~\ref{cor:shuffle-fac-1} $\SP(P \cup \{\varepsilon\},V)$ implies 
$\SWF_P(P^\shuffle \cap V) \subset P^\shuffle \cap V$, which implies $\SP(P \cup \{\varepsilon\},V)$ 
on account of Lemma~\ref{lemma:shuffle-fac-1}.
\end{proof}
\noindent \textbf{Remark.} This proof shows that in Definition~\ref{def:shuffle-projection} 
the restriction $K \neq \emptyset$ can be omitted, or $K$ can also be restricted by 
$\#(K) = \#(\dsN \setminus K) = \#(\dsN)$. \\
\ \\
Additionally to commutativity also associativity of $\Sha$ is well known, see for example \cite{Jantzen85}. 
Because of 
$U \Sha V = \bigcup\limits_{u \in U, v \in V} \{u\} \Sha \{v\}$, the following 
lemma is sufficient for its proof. 
\begin{lemma} \label{lemma:shuffle-fac-2}\ \\
Let $u,v,w \in \Sigma^*$. Then
\begin{align*}
& (\{u\} \Sha \{v\}) \Sha \{w\} = \nonumber \\
& \Theta^{\{1,2,3\}}[(\tau_1^{\{1,2,3\}})^{-1}(\{u\}) 
\cap (\tau_2^{\{1,2,3\}})^{-1}(\{v\}) 
\cap (\tau_3^{\{1,2,3\}})^{-1}(\{w\})] = \nonumber \\
& \{u\} \Sha (\{v\} \Sha \{w\}).
\end{align*}
\end{lemma}
\begin{proof}
$x \in (\{u\} \Sha \{v\}) \Sha \{w\}$, iff
\begin{equation} \label{eq:shuffle-fac-5}
\mbox{there exists } y \in \{u\} \Sha \{v\} 
\mbox{ with } x \in \{y\} \Sha \{w\}.
\end{equation}
\eqref{eq:shuffle-fac-5} is equivalent to:
\begin{align} \label{eq:shuffle-fac-6}
& \mbox{There exist } \breve{y} \in (\tau_1^{\{1,2\}})^{-1}(\{u\}) 
\cap (\tau_2^{\{1,2\}})^{-1}(\{v\}) \mbox{ and } \nonumber \\
& \breve{x} \in (\tau_{\{1,2\}}^{\{{\{1,2\}},3\}})^{-1}(\{y\}) 
   \cap (\tau_3^{\{{\{1,2\}},3\}})^{-1}(\{w\}) \mbox{ with } \nonumber \\
& y = \Theta^{\{1,2\}}(\breve{y}) \mbox{ and } 
   x = \Theta^{\{{\{1,2\}},3\}}(\breve{x}).
\end{align}
\eqref{eq:shuffle-fac-6} is equivalent to:
\begin{align} \label{eq:shuffle-fac-7}
& \mbox{There exists } \breve{z} \in (\tau_1^{\{1,2,3\}})^{-1}(\{u\}) 
   \cap (\tau_2^{\{1,2,3\}})^{-1}(\{v\}) 
   \cap (\tau_3^{\{1,2,3\}})^{-1}(\{w\}) \mbox{ with } \nonumber \\
& \Pi_{\{1,2\}}^{\{1,2,3\}}(\breve{z}) = \breve{y}, \ \Theta^{\{1,2,3\}}(\breve{z}) = x, \nonumber \\
& |\Pi_{\{1,2\}}^{\{1,2,3\}}(z')| = |\tau_{\{1,2\}}^{\{{\{1,2\}},3\}}(x')| \mbox{ and } 
  |\Pi_{\{3\}}^{\{1,2,3\}}(z')| = |\tau_3^{\{{\{1,2\}},3\}}(x')| \nonumber \\
& \mbox{for each } z' \in \pre(\breve{z}) \mbox{ and } x' \in \pre(\breve{x}) \mbox{ with } |z'| = |x'|.
\end{align}
where $\breve{z}$ result by ``combining the structures'' of $\breve{y}$ and $\breve{x}$. 
\eqref{eq:shuffle-fac-5} - \eqref{eq:shuffle-fac-7} proves the first equation of the lemma. 
The second equation can be shown by an analogous argument. 
\end{proof}
Lemma~\ref{lemma:shuffle-fac-2} shows:
\begin{equation*}
u \in \SWF_P(\{w\}) \mbox{ and } x \in \SWF_P(\{u\}) \mbox{ implies } x \in \SWF_P(\{w\})  
\mbox{ for each } w \in P^\shuffle.
\end{equation*}
Therefore $\SWF_P(\SWF_P(M)) \subset \SWF_P(M)$ for each $M \subset P^\shuffle$, which 
by \eqref{eq:shuffle-fac-4} implies   
\begin{equation}\label{eq:shuffle-fac-8}
\SWF_P(\SWF_P(M)) = \SWF_P(M) \mbox{for each M} \subset P^\shuffle. 
\end{equation}
On account of \eqref{eq:shuffle-fac-4} and \eqref{eq:shuffle-fac-8} $\SWF_P$ is a closure operator \cite{DP02}. 
For $P,V \subset \Sigma^*$ and $M \subset P^\shuffle$, by Theorem~\ref{thm:shuffle-fac-1} $\SWF_P(M)$ 
is the smallest $V$ with $X \subset V$ and $\SP(P \cup \{\varepsilon\},V)$. 
On account of \eqref{eq:sp-6} holds
\begin{equation}\label{eq:shuffle-fac-9}
\SWF_P(M) = \bigcap\limits_{\substack
{M \subset V \subset \Sigma^*\\ \SP(P \cup \{\varepsilon\},V)}} V . 
\end{equation}
For $P \subset \Sigma^*$, $\SWF_P$ is a generalization of $\mcC_\Sigma$, where 
\begin{align*}
\mcC_\Sigma(X) := 
\{ u \in \Sigma^* | & \mbox{ there exist } n \in \dsN \mbox{ and } u_i,v_i \in \Sigma^* \mbox{ for } 
1 \leq i \leq n \mbox{ such that } \nonumber \\ 
& u = u_1...u_n  \mbox{ and } u_1v_1...u_nv_n \in X \} = \SWF_\Sigma(X)
\end{align*}
for $X \subset \Sigma^\shuffle = \Sigma^*$ \cite{Haines69}, 
which is called the \emph{downward closure} of $X$. \\
\ \\
In preparation for the next section we show the following 
\begin{lemma} \label{lemma:shuffle-fac-3}\ \\
\[\{ua\} \Sha \{vb\} = (\{u\} \Sha \{vb\})a \cup (\{ua\} \Sha \{v\})b 
\mbox{ \ for } u,v \in \Sigma^* \mbox{ and } a,b \in \Sigma. \]
\end{lemma}
\begin{proof} \ \\
On account of \eqref{eq:shuffle-fac-0} $\{ua\} \Sha \{vb\} \subset \Sigma^+$, and therefore
\begin{align*}
& \{ua\} \Sha \{vb\} = \Theta^{\{1,2\}}[(\tau_1^{\{1,2\}})^{-1}(\{ua\}) \cap (\tau_2^{\{1,2\}})^{-1}(\{vb\})] = \nonumber \\
& \Theta^{\{1,2\}}[(\tau_1^{\{1,2\}})^{-1}(\{ua\}) \cap (\tau_2^{\{1,2\}})^{-1}(\{vb\}) 
   \cap \Sigma_{\{1,2\}}^*\Sigma_{\{1\}}] \cup \nonumber \\ 
& \Theta^{\{1,2\}}[(\tau_1^{\{1,2\}})^{-1}(\{ua\}) \cap (\tau_2^{\{1,2\}})^{-1}(\{vb\}) 
   \cap \Sigma_{\{1,2\}}^*\Sigma_{\{2\}}] = \nonumber \\
& \Theta^{\{1,2\}}[(\tau_1^{\{1,2\}})^{-1}(\{u\}) \cap (\tau_2^{\{1,2\}})^{-1}(\{vb\})]a \cup \nonumber \\
& \Theta^{\{1,2\}}[(\tau_1^{\{1,2\}})^{-1}(\{ua\}) \cap (\tau_2^{\{1,2\}})^{-1}(\{v\})]b = \nonumber \\
& (\{u\} \Sha \{vb\})a \cup (\{ua\} \Sha \{v\})b.
\end{align*}
\end{proof}
The properties of \eqref{eq:shuffle-fac-0} and Lemma~\ref{lemma:shuffle-fac-3} completely characterize $\Sha$. 
It is well known that
\begin{align} \label{eq:shuffle-fac-10}
& \{w\} = \{w\} \Sha \{\varepsilon\} = \{\varepsilon\} \Sha \{w\}
   \mbox{ for } w \in \Sigma^*, \nonumber \\
& \{ua\} \Sha \{vb\} = (\{u\} \Sha \{vb\})a \cup (\{ua\} \Sha \{v\})b 
\mbox{ \ for } u,v \in \Sigma^* \mbox{ and } a,b \in \Sigma, \mbox{ and } \nonumber \\
& U \Sha V = \bigcup\limits_{u \in U, v \in V} \{u\} \Sha \{v\} \mbox{ \ for } U,V \subset \Sigma^*
\end{align}
inductively defines $\Sha$, see for example \cite{JedrzejowiczS01a}. \\
\ \\
Lemma~\ref{lemma:shuffle-fac-1} shows
\begin{align}\label{eq:shuffle-fac-11}
& \mbox{For each } \varepsilon \neq w \in P^\shuffle \mbox{ there exist } \varepsilon \neq e \in P 
\mbox{ and } v \in P^\shuffle \mbox{ with } \nonumber \\
& w \in \{e\} \Sha \{v\}.
\end{align}
\eqref{eq:shuffle-fac-11} together with Lemma~\ref{lemma:shuffle-fac-1} implies 
the following well known inductive definition of $P^\shuffle$, see for example \cite{JedrzejowiczS01a}:
\begin{align}\label{eq:shuffle-fac-12}
& P^\shuffle = \bigcup\limits_{n \in \dsN} P^{(\shuffle,n)} \mbox{ where } \nonumber \\
& P^{(\shuffle,1)} := P \cup \{\varepsilon\} \mbox{ and } 
P^{(\shuffle,n+1)} := P^{(\shuffle,n)} \Sha (P \cup \{\varepsilon\}) \mbox{ for } n \in \dsN.
\end{align}
\eqref{eq:shuffle-fac-12} motivates
\begin{definition}\label{def:shuffle-factors-n}\ \\
For $P \subset \Sigma^*$ and $n \in \dsN$ let $\SWF_P^{(n)} : 2^{P^\shuffle} \to 2^{P^\shuffle}$ 
be defined by $\SWF_P^{(n)}(M) :=$ 
\[ \{u \in P^\shuffle | \mbox{ there exist } w \in M \mbox{ and } v \in P^{(\shuffle,n)} 
\mbox{ such that } w \in \{u\} \Sha \{v\} \} \]
for $M \subset P^\shuffle$.
\end{definition}
It is an immediate consequence of this definition that 
\begin{align} \label{eq:shuffle-fac-13}
& \SWF_P^{(n)}(M) = \bigcup\limits_{w \in M } \SWF_P^{(n)}(\{w\}) \mbox{ and therefore } \nonumber \\
& \SWF_P^{(n)}(U) \subset \SWF_P^{(n)}(M) \mbox{ for } U \subset M \subset P^\shuffle \mbox{ and } n \in \dsN.
\end{align}
Since $\{\varepsilon\} \subset P^{(\shuffle,n)} \subset P^{(\shuffle,n+1)}$ for $n \in \dsN$, 
\eqref{eq:shuffle-fac-12} implies 
\begin{align} \label{eq:shuffle-fac-14}
& M \subset \SWF_P^{(n)}(M) \subset \SWF_P^{(n+1)}(M) \mbox{ for } n \in \dsN, \nonumber \\
& \mbox{and } \SWF_P(M) = \bigcup\limits_{n \in \dsN} \SWF_P^{(n)}(M) \mbox{ for } M \subset P^\shuffle.
\end{align}
The iterative definition of $P^{(\shuffle,n)}$ together with the commutativity and associativity of $\Sha$ shows:
\begin{align} \label{eq:shuffle-fac-15}
& \SWF_P^{(n+1)}(M) = \SWF_P^{(1)}(\SWF_P^{(n)}(M)) = \SWF_P^{(n)}(\SWF_P^{(1)}(M)) \nonumber \\
& \mbox{for } M \subset P^\shuffle \mbox{ and } n \in \dsN.
\end{align}
For $M \subset P^\shuffle$ \eqref{eq:shuffle-fac-15} by induction implies 
\begin{equation*}
\SWF_P^{(1)}(M) \subset M  \mbox{ iff } \SWF_P^{(n)}(M) \subset M \mbox{ for each } n \in \dsN.
\end{equation*}
Therefore, by \eqref{eq:shuffle-fac-14} and Theorem~\ref{thm:shuffle-fac-1} holds
\begin{corollary}\label{cor:shuffle-fac-2}\ \\
Let $P,V \subset \Sigma^*$. Then $\SP(P \cup \{\varepsilon\},V)$  iff 
$\SWF_P^{(1)}(P^\shuffle \cap V) \subset P^\shuffle \cap V$.
\end{corollary}
By Lemma~\ref{lemma:shuffle-fac-1} Corollary~\ref{cor:shuffle-fac-2} is a reformulation of Theorem~\ref{thm:mod-shuffle-projection}.
\section{Shuffled Runs of Computations in S-Automata}
To represent $\SWF_{\pre(P)}^{(1)}$ and $\SWF_{\pre(P)}$ for $\emptyset \neq P \subset \Sigma^*$ 
in terms of computations in S-automata, 
now a kind of shuffle product will be defined on $2^{A_\bbP}$. 
Guideline for this definition is \eqref{eq:shuffle-fac-10}. Preparatively let $\bbP$ and 
$A_\bbP$ be defined as in section~\ref{sec:automata}, and let
\begin{align*}
& \{c\} \Sha^{(\bbP)} \{\varepsilon\} := \{\varepsilon\} \Sha^{(\bbP)} \{c\} := \{c\} 
   \mbox{ for } c \in A_\bbP, \mbox{ and } \nonumber \\
& \{c(f,a,f')\} \Sha^{(\bbP)} \{d(g,b,g')\} := \nonumber \\ 
& (\{c\} \Sha^{(\bbP)} \{d(g,b,g')\})(f+g',a,f'+g') \cup (\{c(f,a,f')\} \Sha^{(\bbP)} \{d\})(f'+g,b,f'+g') \nonumber \\ 
& \mbox{for } c(f,a,f'),d(g,b,g') \in A_\bbP \mbox{ with } (f,a,f'),(g,b,g') \in \shuffle_\bbP.
\end{align*}
Then \eqref{eq:grave-P-2} and induction show 
\begin{align}\label{eq:shuffle-fac-16}
& \{x\} \Sha^{(\bbP)} \{y\} = \{y\} \Sha^{(\bbP)} \{x\}, \ \{x\} \Sha^{(\bbP)} \{y\} \subset A_\bbP, \nonumber \\
& |c| = |x| + |y| \mbox{ and } Z_\bbP(c) = Z_\bbP(x) + Z_\bbP(y) \nonumber \\
& \mbox{for } x,y \in A_\bbP 
  \mbox{ and } c \in \{x\} \Sha^{(\bbP)} \{y\}.
\end{align}
\begin{definition}\label{def:shuffle-prod-P}\ \\
Using \eqref{eq:shuffle-fac-16}, let the commutativ operation 
$\Sha^{(\bbP)} : 2^{A_\bbP} \times 2^{A_\bbP} \to 2^{A_\bbP}$ in infix notation be inductively defined by 
\begin{align*}
& \{c\} \Sha^{(\bbP)} \{\varepsilon\} := \{\varepsilon\} \Sha^{(\bbP)} \{c\} := \{c\} 
   \mbox{ for } c \in A_\bbP, \nonumber \\
& \{c(f,a,f')\} \Sha^{(\bbP)} \{d(g,b,g')\} := \nonumber \\ 
& (\{c\} \Sha^{(\bbP)} \{d(g,b,g')\})(f+g',a,f'+g') \cup (\{c(f,a,f')\} \Sha^{(\bbP)} \{d\})(f'+g,b,f'+g') \nonumber \\ 
& \mbox{for } c(f,a,f'),d(g,b,g') \in A_\bbP \mbox{ with } (f,a,f'),(g,b,g') \in \shuffle_\bbP, \mbox{ and } \nonumber \\
& X \Sha^{(\bbP)} Y = \bigcup\limits_{x \in X, y \in Y} \{x\} \Sha^{(\bbP)} \{y\} \mbox{ for } X,Y \subset A_\bbP.
\end{align*}
$X \Sha^{(\bbP)} Y$ is called the \emph{shuffled runs} of $X$ and $Y$.
\end{definition}
The name \emph{shuffled runs} is justified by the relation to section~\ref{sec:automata-shuffle projection}, 
as will be demonstrated in the last theorem of this section. \\
\ \\
Definition~\ref{def:shuffle-prod-P} allows to transfer Definition~\ref{def:shuffle-factors} to $A_\bbP$:
\begin{definition}\label{def:shuffle-factors-bbP}\ \\
Let $\SRF_\bbP : 2^{A_\bbP} \to 2^{A_\bbP}$ and $\SRF_\bbP^{(1)} : 2^{A_\bbP} \to 2^{A_\bbP}$ 
be defined by
\begin{align*}
& \SRF_\bbP(M) := \nonumber \\
& \{u \in A_\bbP | \mbox{ there exist } w \in M \mbox{ and } v \in A_\bbP 
  \mbox{ such that } w \in \{u\} \Sha^{(\bbP)} \{v\} \} \nonumber \\
& \mbox{and } \nonumber \\
& \SRF_\bbP^{(1)}(M) := \nonumber \\
& \{u \in A_\bbP | \mbox{ there exist } w \in M \mbox{ and } e \in E_\bbP 
  \mbox{ such that } w \in \{u\} \Sha^{(\bbP)} \{e\} \} \nonumber \\
& \mbox{for } M \subset A_\bbP.
\end{align*}
\end{definition}
The following two lemmas are the key to express $\SWF_{\pre(P)}$ resp. $\SWF_{\pre(P)}^{(1)}$ by 
$\SRF_\bbP$ resp. $\SRF_\bbP^{(1)}$.
\begin{lemma} \label{lemma:shuffle-fac-4}\ \\
For $c,d \in A_\bbP$ and $x \in \{c\} \Sha^{(\bbP)} \{d\}$ holds 
$\alpha_\bbP(x) \in \{\alpha_\bbP(c)\} \Sha \{\alpha_\bbP(d)\}$.
\end{lemma}
\begin{proof} [by induction] \ \\
\emph{Induction base}  \\
Let $c=\varepsilon$ or $d=\varepsilon$. On account of commutativity of $\Sha^{(\bbP)}$ 
let $d=\varepsilon$. Then $x = c$ and $\alpha_\bbP(d) = \varepsilon$, which implies 
$\alpha_\bbP(x) \in \{\alpha_\bbP(c)\} \Sha \{\alpha_\bbP(d)\}$. \\
\emph{Induction step}  \\
$c \neq \varepsilon \neq d$ implies $c = c'(f,a,f')$ and $d = d'(g,b,g')$ 
with $c',d' \in A_\bbP$ and $(f,a,f'),(g,b,g') \in \shuffle_\bbP$. Therefore 
$x \in (\{c'\} \Sha^{(\bbP)} \{d'(g,b,g')\})(f+g',a,f'+g') \cup (\{c'(f,a,f')\} \Sha^{(\bbP)} \{d'\})(f'+g,b,f'+g')$. 
On account of symmetry it is sufficient to prove the induction step for 
$x \in (\{c'\} \Sha^{(\bbP)} \{d'(g,b,g')\})(f+g',a,f'+g')$, which implies $x = x'(f+g',a,f'+g')$ 
with $x' \in \{c'\} \Sha^{(\bbP)} \{d'(g,b,g')\}$. Now by the induction hypothesis 
$\alpha_\bbP(x') \in \{\alpha_\bbP(c')\} \Sha \{\alpha_\bbP(d')b\}$, and therefore by Lemma~\ref{lemma:shuffle-fac-3}
$\alpha_\bbP(x) \in (\{\alpha_\bbP(c')\} \Sha \{\alpha_\bbP(d')b\})a 
\subset \{\alpha_\bbP(c')a\} \Sha \{\alpha_\bbP(d')b\} = \{\alpha_\bbP(c)\} \Sha \{\alpha_\bbP(d)\}$, which completes 
the proof of Lemma~\ref{lemma:shuffle-fac-4}.
\end{proof}
\begin{lemma} \label{lemma:shuffle-fac-5}\ \\
For $u,v \in \pre(P^\shuffle) = (\pre(P))^\shuffle$, $w \in \{u\} \Sha \{v\}$, $c \in \alpha_\bbP^{-1}(u)$ 
and $d \in \alpha_\bbP^{-1}(v)$ there exists 
$x \in \{c\} \Sha^{(\bbP)} \{d\}$ with $\alpha_\bbP(x) = w$.
\end{lemma}
\begin{proof} [by induction] \ \\
\emph{Induction base}  \\
Let $u=\varepsilon$ or $v=\varepsilon$. On account of commutativity of $\Sha$ 
let $v=\varepsilon$. Then $w = u$ and $d = \varepsilon$, which implies 
$c \in \{c\} \Sha^{(\bbP)} \{d\}$ with $\alpha_\bbP(c) = w$. \\
\emph{Induction step}  \\
$u \neq \varepsilon \neq v$ implies $u = u'a$ and $v = v'b$ 
with $u',v' \in \pre(P^\shuffle)$ and $a,b \in \Sigma$. Therefore 
$w \in (\{u'\} \Sha \{v'b\})a \cup (\{u'a\} \Sha \{v'\})b$, 
$c = c'(f,a,f')$ and $d = d'(g,b,g')$ 
with $c' \in \alpha_\bbP^{-1}(u')$, $d' \in \alpha_\bbP^{-1}(v')$, 
$(f,a,f'),(g,b,g') \in \shuffle_\bbP$, $Z_\bbP(c') = f$ and $Z_\bbP(d') = g$. 
On account of symmetry it is sufficient to prove the induction step for 
$w \in (\{u'\} \Sha \{v'b\})a$, which implies $w = w'a$ 
with $w' \in \{u'\} \Sha \{v'b\}$. 
Now by the induction hypothesis there exists
$x' \in \{c'\} \Sha^{(\bbP)} \{d'(g,b,g')\}$ with $\alpha_\bbP(x') = w'$. 
Then $x := x'(f+g',a,f'+g') \in \{c\} \Sha^{(\bbP)} \{d\}$ with $\alpha_\bbP(x) = w$, 
which completes the proof of Lemma~\ref{lemma:shuffle-fac-5}.
\end{proof}
\begin{theorem}\label{thm:shuffle-fac-2}\ \\ 
$\SWF_{\pre(P)}(M) = \alpha_\bbP(\SRF_\bbP(\alpha_\bbP^{-1}(M)))$ for 
$M \subset \pre(P^\shuffle) = (\pre(P))^\shuffle$.
\end{theorem}
\begin{proof} \ \\
For $u \in \SWF_{\pre(P)}(M) \subset \pre(P^\shuffle)$ there exist $w \in M$ and $v \in \pre(P^\shuffle)$ 
such that $w \in \{u\} \Sha \{v\}$. 
Now, by Corollary~\ref{cor:cp-surj} and Lemma~\ref{lemma:shuffle-fac-5} there exist 
$c \in \alpha_\bbP^{-1}(u) \subset A_\bbP$, $d \in \alpha_\bbP^{-1}(v) \subset A_\bbP$, 
and $x \in \alpha_\bbP^{-1}(w) \subset \alpha_\bbP^{-1}(M)$ with $x \in \{c\} \Sha^{(\bbP)} \{d\}$. 
This implies $c \in \SRF_\bbP(\alpha_\bbP^{-1}(M))$, which proves 
$u \in \alpha_\bbP(\SRF_\bbP(\alpha_\bbP^{-1}(M)))$. \\
For $c \in \SRF_\bbP(\alpha_\bbP^{-1}(M)) \subset A_\bbP$ there exist $x \in \alpha_\bbP^{-1}(M)$ 
and $d \in A_\bbP$ such that $x \in \{c\} \Sha^{(\bbP)} \{d\}$. 
Now, by Corollary~\ref{cor:cp-surj} and Lemma~\ref{lemma:shuffle-fac-4} 
$\alpha_\bbP(c), \alpha_\bbP(d) \in (\pre(P))^\shuffle$, and 
$\alpha_\bbP(x) \in \{\alpha_\bbP(c)\} \Sha \{\alpha_\bbP(d)\}$, which shows 
$\alpha_\bbP(c) \in \SWF_{\pre(P)}(M)$. 
This completes the proof of Theorem~\ref{thm:shuffle-fac-2}. 
\end{proof}
The proof of Theorem~\ref{thm:shuffle-fac-2} together with $P^\shuffle = \alpha_\bbP(Z_\bbP^{-1}(0))$ 
(Corollary~\ref{cor:cp-surj}) and \eqref{eq:shuffle-fac-16} shows
\begin{corollary}\label{cor:shuffle-fac-3}\ \\
$\SWF_P(M) = \alpha_\bbP(\SRF_\bbP(\alpha_\bbP^{-1}(M) \cap Z_\bbP^{-1}(0)))$ for 
$M \subset P^\shuffle$.
\end{corollary}
Together with $\alpha_\bbP(E_\bbP) = \pre(P)$ and $\alpha_\bbP(E_\bbP \cap Z_\bbP^{-1}(0)) = P$ \eqref{eq:char-E_P}, 
the proofs of Theorem~\ref{thm:shuffle-fac-2} and Corollary~\ref{cor:shuffle-fac-3} shows
\begin{corollary}\label{cor:shuffle-fac-4}\ \\
$\SWF_P^{(1)}(M) = \alpha_\bbP(\SRF_\bbP^{(1)}(\alpha_\bbP^{-1}(M) \cap Z_\bbP^{-1}(0)))$ for 
$M \subset P^\shuffle$, and \\ 
$\SWF_{\pre(P)}^{(1)}(M) = \alpha_\bbP(\SRF_\bbP^{(1)}(\alpha_\bbP^{-1}(M)))$ for 
$M \subset \pre(P^\shuffle) = (\pre(P))^\shuffle$.
\end{corollary}
Because of $\alpha_\bbP^{-1}((\pre(P))^\shuffle \cap V) = \alpha_\bbP^{-1}(V)$, 
Corollary~\ref{cor:shuffle-fac-2} and Corollary~\ref{cor:shuffle-fac-4} imply
\begin{corollary}\label{cor:shuffle-fac-5}\ \\
$\SP(\pre(P),V)$  iff $\SRF_\bbP^{(1)}(\alpha_\bbP^{-1}(V)) \subset \alpha_\bbP^{-1}(V)$.
\end{corollary}
Because of $Z_\bbP^{-1}(0) \subset \alpha_\bbP^{-1}(P^\shuffle)$, it holds
\begin{align}\label{eq:shuffle-fac-17}
\alpha_\bbP^{-1}(P^\shuffle \cap V) \cap Z_\bbP^{-1}(0) 
& = \alpha_\bbP^{-1}(P^\shuffle) \cap \alpha_\bbP^{-1}(V) \cap Z_\bbP^{-1}(0) \nonumber \\
& = \alpha_\bbP^{-1}(V) \cap Z_\bbP^{-1}(0).
\end{align}
\eqref{eq:shuffle-fac-16} implies
\begin{equation}\label{eq:shuffle-fac-18}
\SRF_\bbP^{(1)}(\alpha_\bbP^{-1}(V) \cap Z_\bbP^{-1}(0)) \subset Z_\bbP^{-1}(0).
\end{equation}
By \eqref{eq:shuffle-fac-17} and \eqref{eq:shuffle-fac-18} it holds 
\begin{align}\label{eq:shuffle-fac-19}
& \SRF_\bbP^{(1)}(\alpha_\bbP^{-1}(P^\shuffle \cap V) \cap Z_\bbP^{-1}(0)) 
\subset \alpha_\bbP^{-1}(P^\shuffle \cap V) \mbox{ iff } \nonumber \\
& \SRF_\bbP^{(1)}(\alpha_\bbP^{-1}(V) \cap Z_\bbP^{-1}(0)) 
\subset \alpha_\bbP^{-1}(V) \cap Z_\bbP^{-1}(0).
\end{align}
Now, because of \eqref{eq:shuffle-fac-19}, 
Corollary~\ref{cor:shuffle-fac-2} and Corollary~\ref{cor:shuffle-fac-4} imply
\begin{corollary}\label{cor:shuffle-fac-6}\ \\
$\SP(P \cup \{\varepsilon\},V)$  iff 
$\SRF_\bbP^{(1)}(\alpha_\bbP^{-1}(V) \cap Z_\bbP^{-1}(0)) \subset \alpha_\bbP^{-1}(V) \cap Z_\bbP^{-1}(0)$.
\end{corollary}
To show that Corollary~\ref{cor:shuffle-fac-5} and Corollary~\ref{cor:shuffle-fac-6} 
are equivalent to Corollary~\ref{cor:eq-rel-RP'} and Corollary~\ref{cor:eq-rel-R0P'}, we prove
\begin{theorem}\label{thm:shuffle-fac-3}\  
$\SRF_\bbP^{(1)} = \mcR'_\bbP$.
\end{theorem}
\begin{proof} \ \\
Since $\SRF_\bbP^{(1)}(M) = \bigcup\limits_{x \in M} \SRF_\bbP^{(1)}(\{x\})$ and 
$\mcR'_\bbP(M) = \bigcup\limits_{x \in M} \mcR'_\bbP(\{x\})$ it is sufficient to prove 
$\SRF_\bbP^{(1)}(\{x\}) = \mcR'_\bbP(\{x\})$ for each $x \in A_\bbP$. 
For this purpose we show the following:
\begin{align}\label{eq:shuffle-fac-20}
& \mbox{For } x,u \in A_\bbP \mbox{ and } e \in E_\bbP \mbox{ holds }
  x \in \{u\} \Sha^{(\bbP)} \{e\} \mbox{ iff there exists a shuffled } \nonumber \\
& \mbox{representation } b \in 
  \pi_{\shuffle_{\bbP}}^{-1}(A_{\bbP}) \cap \pi_{\shuffle_{\check{\bbP}}}^{-1}(E_{\check{\bbP}}) 
  \mbox{ of } x \mbox{ by } u \mbox{ and } \check e := \check{\iota}_{\shuffle_{\check{\bbP}}}^{-1}(e).
\end{align}
Because of $\shuffle_{\check{\bbP}} \cap \shuffle_{\bbP} = \emptyset$ it holds:
\begin{align}\label{eq:shuffle-fac-21}
& b \in 
  \pi_{\shuffle_{\bbP}}^{-1}(A_{\bbP}) \cap \pi_{\shuffle_{\check{\bbP}}}^{-1}(E_{\check{\bbP}}) 
  \mbox{ with } \pi_{\shuffle_{\check{\bbP}}}(b) = \check e \mbox{ and } \pi_{\shuffle_{\bbP}}(b) = u \nonumber \\
& \mbox{iff } b  \in \{u\} \Sha \{\check e\}.
\end{align}
Now \eqref{eq:shuffle-fac-21} allows to prove \eqref{eq:shuffle-fac-20} inductively using the 
inductive definitions of $\Sha^{(\bbP)}$ and $\Sha$.\\
\emph{Induction base}  \\
Let $u=\varepsilon$ or $e=\varepsilon$. We only consider $e=\varepsilon$, because $u=\varepsilon$ 
can be treated analogously. 
Then $x \in \{u\} \Sha^{(\bbP)} \{e\}$, iff $x = u$, iff there exists a shuffled representation $b \in 
\pi_{\shuffle_{\bbP}}^{-1}(A_{\bbP}) \cap \pi_{\shuffle_{\check{\bbP}}}^{-1}(E_{\check{\bbP}})$ 
  of $x$ by $u$ and $\check e = \varepsilon$.\\ 
\emph{Induction step}  \\
$u \neq \varepsilon \neq e$ implies $u = u'(f,a,f')$ and $e = e'(g,b,g')$ 
with $u' \in A_\bbP$, $e' \in E_\bbP$, and $(f,a,f'),(g,b,g') \in \shuffle_\bbP$. Therefore, 
$x \in \{u\} \Sha^{(\bbP)} \{e\}$ implies 
$x \in (\{u'\} \Sha^{(\bbP)} \{e'(g,b,g')\})(f+g',a,f'+g') \cup 
(\{u'(f,a,f')\} \Sha^{(\bbP)} \{e'\})(f'+g,b,f'+g')$. 
On account of symmetry it is sufficient to prove the induction step for 
$x \in (\{u'\} \Sha^{(\bbP)} \{e'(g,b,g')\})(f+g',a,f'+g')$, which implies $x = x'(f+g',a,f'+g')$ 
with $x' \in \{u'\} \Sha^{(\bbP)} \{e'(g,b,g')\}$. 
Now by the induction hypothesis $x' \in \{u'\} \Sha^{(\bbP)} \{e'(g,b,g')\}$ implies the existence of 
a shuffled representation $b' \in 
\pi_{\shuffle_{\bbP}}^{-1}(A_{\bbP}) \cap \pi_{\shuffle_{\check{\bbP}}}^{-1}(E_{\check{\bbP}})$ 
of $x'$ by $u'$ and $\check e = \check{\iota}_{\shuffle_{\check{\bbP}}}^{-1}(e')(g,\check b,g')$. 
But then $b := b'(f,a,f')$  is a shuffled representation of $x = x'(f+g',a,f'+g')$ by 
$u = u'(f,a,f')$ and $\check e = \check{\iota}_{\shuffle_{\check{\bbP}}}^{-1}(e')(g,\check b,g')$. \\
Let now $b$ be a shuffled representation of $x$ by 
$u = u'(f,a,f')$ and $\check e = \check{\iota}_{\shuffle_{\check{\bbP}}}^{-1}(e')(g,\check b,g')$. 
Then $b \in \{u\} \Sha \{\check e\} = (\{u'\} \Sha 
\{\check{\iota}_{\shuffle_{\check{\bbP}}}^{-1}(e')(g,\check b,g')\})(f,a,f') \cup 
(\{u'(f,a,f')\} \Sha^{(\bbP)} \{\check{\iota}_{\shuffle_{\check{\bbP}}}^{-1}(e')\})(g,\check b,g')$. 
On account of symmetry it is sufficient to prove the induction step for 
$b \in (\{u'\} \Sha \{\check{\iota}_{\shuffle_{\check{\bbP}}}^{-1}(e')(g,\check b,g')\})(f,a,f')$,
which implies $b = b'(f,a,f')$ with 
$b' \in (\{u'\} \Sha \{\check{\iota}_{\shuffle_{\check{\bbP}}}^{-1}(e')(g,\check b,g')\})$. Additionally
$b'$ is a shuffled representation of $x'$ by 
$u'$ and $\check e = \check{\iota}_{\shuffle_{\check{\bbP}}}^{-1}(e')(g,\check b,g')$ with 
$x = x'(f+g',a,f'+g')$. Now by the induction hypothesis $x' \in \{u'\} \Sha^{(\bbP)} \{e'(g,b,g')\}$, 
which implies $x = x'(f+g',a,f'+g') \in \{u'(f,a,f'\} \Sha^{(\bbP)} \{e'(g,b,g')\}$.
This completes the induction step and the proof of Theorem~\ref{thm:shuffle-fac-3}. 
\end{proof}
Analogously to the proofs of Theorem~\ref{thm:shuffle-fac-3} and Theorem~\ref{thm:RT} a
representation of $\SRF_\bbP$ can be constructed like such in Theorem~\ref{thm:RT}.
%
%
\bibliographystyle{plain}
\bibliography{shuffle-projection}

\begin{thebibliography}{10}

\bibitem{berstel79}
Jean Berstel.
\newblock {\em Transductions and Context-Free Languages}.
\newblock B.G. Teubner Stuttgart, 1979.

\bibitem{BjorklundB07}
Henrik Bj{\"o}rklund and Mikolaj Bojanczyk.
\newblock {Shuffle Expressions and Words with Nested Data}.
\newblock In {\em Mathematical Foundations of Computer Science 2007}, pages
  750--761, 2007.

\bibitem{DP02}
B.A. Davey and H.A. Priestley.
\newblock {\em Introduction to Lattices and Order}.
\newblock Cambridge University Press, 2002.

\bibitem{Haines69}
L.H. Haines.
\newblock On free monoids partially ordered by embedding.
\newblock {\em Journal of Combinatorial Theory}, 6:94--98, 1969.

\bibitem{Jantzen85}
Matthias Jantzen.
\newblock {Extending Regular Expressions with Iterated Shuffle}.
\newblock {\em Theor. Comput. Sci.}, 38:223--247, 1985.

\bibitem{Jedrzejowicz99}
Joanna Jedrzejowicz.
\newblock {Structural Properties of Shuffle Automata}.
\newblock {\em Grammars}, 2(1):35--51, 1999.

\bibitem{JedrzejowiczS01a}
Joanna Jedrzejowicz and Andrzej Szepietowski.
\newblock {Shuffle languages are in P}.
\newblock {\em Theor. Comput. Sci.}, 250(1-2):31--53, 2001.

\bibitem{lothaire83}
M.~Lothaire.
\newblock {\em {Combinatorics on Words}}.
\newblock Cambridge University Press, 1983.

\bibitem{OR2010t}
Peter Ochsenschl\"{a}ger and Roland Rieke.
\newblock {Uniform Parameterisation of Phase Based Cooperations}.
\newblock Technical Report SIT-TR-2010/1, Fraunhofer SIT, 2010.

\bibitem{OR2011}
Peter Ochsenschl\"{a}ger and Roland Rieke.
\newblock Security properties of self-similar uniformly parameterised systems
  of cooperations.
\newblock In {\em Proceedings of the 19th Euromicro International Conference on
  Parallel, Distributed and Network-Based Computing (PDP)}. IEEE Computer
  Society, February 2011.

\bibitem{SysMea14}
Peter Ochsenschl\"ager and Roland Rieke.
\newblock Safety by construction: Well-behaved scalable systems.
\newblock {\em International Journal On Advances in Systems and Measurements},
  7(3 \& 4):239 -- 257, 2014.

\bibitem{reutenauer90}
Christophe Reutenauer.
\newblock {\em The Mathematics of Petri Nets}.
\newblock Masson and Prentice Hall International, 1990.

\bibitem{Wimmel08}
Harro Wimmel.
\newblock {\em {Entscheidbarkeit bei Petri Netzen -- {\"U}berblick und
  Kompendium}}.
\newblock Springer-Verlag, 2008.

\end{thebibliography}
\end{document}